%% file: main.tex
\documentclass[11pt]{article}

\usepackage[a4paper,width=159mm,top=25mm,bottom=25mm,bindingoffset=0mm]{geometry}

\usepackage[sorting=none,style=numeric-comp]{biblatex}
\usepackage[utf8]{inputenc}
\usepackage[T1]{fontenc}
\usepackage[english]{babel}

\usepackage{amsmath,amscd,amssymb,amsfonts}
\usepackage{mathrsfs}
\usepackage{braket}
\usepackage{array}
\usepackage{graphicx}
\usepackage{float}
\usepackage{subcaption}
\usepackage{tikz}
\usepackage{booktabs}
\usepackage{multirow}
\usepackage{siunitx}
\usepackage{tabularx}
\usepackage{mdframed}
\usepackage{xurl}
\usepackage{authblk}

\usepackage{hyperref}
\usepackage[dvipsnames]{xcolor}
\hypersetup{
    colorlinks=true,
    linkcolor=RoyalBlue,
    filecolor=magenta,      
    urlcolor=cyan,
    citecolor=RoyalPurple
    }

\usepackage[toc,page]{appendix}

\usepackage{amsthm}
\newtheorem{theorem}{Theorem}[section]
\newtheorem{corollary}{Corollary}[theorem]
\newtheorem{lemma}[theorem]{Lemma}
\newtheorem{proposition}[theorem]{Proposition}

\usepackage{dsfont}
\newcommand{\thetaB}{\boldsymbol{\theta}}
\newcommand{\Var}{{\rm Var}}
\newcommand{\Cov}{{\rm Cov}}
\newcommand{\Tr}{{\rm Tr}}
\newcommand{\dd}{\text{d}}
\renewcommand{\Re}{{\rm Re}}
\newcommand{\poly}{{\rm poly}}
\newcommand{\polylog}{{\rm polylog}}
\newcommand{\rk}{{\rm rank}}
\newcommand{\im}{{\rm im}}

\newcommand{\qimpl}{\quad \Rightarrow \quad}

\graphicspath{{images/}}

\usepackage[acronyms,toc]{glossaries}

\makeglossaries
\glsdisablehyper
\renewcommand*{\glossaryentrynumbers}[1]{}

\loadglsentries{glossary}

\begin{document}

\title{Quantum Graph Convolutional Networks: Implementation and Trainability Analysis}

\author[1,2]{Paul San Sebastian Sein$^*$}
\author[3,4]{Theodor Iosif$^*$}
\author[3]{Tilen G. Limbäck-Stokin}
\author[3]{Kin Ian Lo}
\author[5,6,7]{Yidong Liao}

\affil[1]{Ikerlan Technology Research Centre, Basque Research and Technology Alliance (BRTA), Arrasate-Mondragon, Spain}
\affil[2]{University of the Basque Country/Euskal Herriko Unibertsitatea-EHU}

\affil[3]{Quantum Learning Labs, Department of Computer Science, University College London, London, United Kingdom}
\affil[4]{London Centre for Nanotechnology, London, United Kingdom}
\affil[5]{Centre for Quantum Software and Information, University of Technology Sydney, Sydney, NSW, Australia}
\affil[6]{Sydney Quantum Academy, Sydney, NSW, Australia}
\affil[7]{Laboratoire d'Informatique de Paris 6, CNRS, Sorbonne Université, 4 Place Jussieu, 75005 Paris, France}

\date{}

\maketitle
\def\thefootnote{*}\footnotetext{These authors contributed equally to this work.}\def\thefootnote{\arabic{footnote}}

\begin{abstract}
Graph Neural Networks (GNNs) achieve state-of-the-art performance on graph-structured data, but training and inference on large graphs are often bottlenecked by memory constraints and sparse linear-algebra workloads. Quantum computing offers an alternative set of primitives that may improve scalability for graph learning. Building on the quantum graph neural network (QGNN) framework of Liao \textit{et al.}~\cite{liao_graph_2024}, this work implements two representative architectures --- the \gls{sgc} and \gls{lgc} models --- and evaluates them on open benchmark graph datasets and semi-supervised learning tasks using quantum simulation. We compare predictive performance and optimization behavior against classical baselines, showing that the quantum models achieve competitive performance with fewer parameters. 
Finally, we present a cost gradient analysis that identifies the tasks for which the models showcased are trainable. This is followed by a classical simulability study to find regimes in which the proposed circuits remain robust during training.
\end{abstract}

\section{Introduction}\label{qgnn_intro}

Graph Neural Networks (GNNs) are powerful machine learning models for analyzing structured data represented as graphs. They have shown remarkable success in various applications including social network analysis \cite{borisyuk_lignn_2024,fan_graph_2019}, recommendation systems \cite{de_nadai_personalized_2024}, drug discovery \cite{stokes_deep_2020,zitnik_modeling_2018}, and traffic prediction \cite{derrow-pinion_eta_2021}. From a theoretical perspective, GNNs have been posited as a universal framework for various neural network architectures: Convolutional Neural Networks, Recurrent Neural Networks, Transformers etc. can be viewed as special cases of GNNs \cite{bronstein_geometric_2021, joshi_transformers_2025}.\\

Despite their success, classical GNNs face several challenges when dealing with large-scale graphs. One major challenge is the memory limitation that arises when handling giant graphs. Large and complex graphs become increasingly difficult to fit in the conventional memory used by most classical computing hardware \cite{chiang_cluster-gcn_2019}. Another issue lies in the inherent sparse matrix operations of GNNs, which pose challenges for efficient computation on modern hardware like GPUs that are optimized for dense matrix operations\footnote{Customized hardware accelerators for sparse matrices can improve GNNs' latency and scalability, but their design remains an open question.}. Moreover, the common method of managing large graphs through graph subsampling techniques (e.g. dividing large graphs into smaller, more manageable subgraphs \cite{chiang_cluster-gcn_2019}) may encounter reliability issues, since it is challenging to guarantee that these subgraphs preserve the semantics of the entire graph and provide reliable gradients for training GNNs. In summary, the memory and computational requirements of processing large-scale graphs often exceed the capabilities of classical computing hardware, limiting the practical scalability of GNNs. The need for efficient and scalable graph learning has motivated ongoing efforts in developing specialized hardware accelerators for GNNs \cite{kiningham_grip_2023,auten_hardware_2020,abadal_computing_2021} as well as the exploration of utilizing alternative computing paradigms, such as quantum computing, to address these challenges.\\

Quantum computers hold the promise of significantly improving machine learning by providing computational speed-ups or improved model scalability \cite{cerezo_challenges_2022}. In the context of graph learning, quantum computing provides new opportunities to design quantum machine learning architectures tailored for graph-structured data \cite{beer_quantum_2023,skolik_equivariant_2022,verdon_quantum_2019}. Motivated by this potential, Liao et al.~\cite{liao_graph_2024} proposed novel \gls{qgnn} architectures using quantum computing primitives that serve as a starting point for developing more complex \gls{qgnn} models.\\

In this work, two branches of the Quantum Graph Convolutional Networks (QGCNs) in \cite{liao_graph_2024} are considered --- the simple (SGC) and linear (LGC) graph convolutions. They are implemented, trained, and tested using simulators, and their performance is compared to that of their classical counterparts. Open benchmark graph datasets were employed for this task, in a semi-supervised learning context. Additionally, a detailed trainability analysis is provided, showing the robustness of these architectures to the loss concentration problem.\\

The complexity analysis by Liao et al.~\cite{liao_graph_2024} indicates that the quantum implementation of a SGC network can potentially achieve significant improvements in time and/or space complexity compared to its classical counterpart, under conditions commonly encountered in real-world applications. The quantum LGC closely follows, with an additional exchange between a more expressive message-passing and runtime. Following the trainability calculation, our work extends this complexity analysis to compare the full models, also taking into account the backpropagation and shot noise factors. Here we flag that these quantum models suffer from an additional circuit compilation complexity which on its own is comparable to the classical SGC performance. This is called the \textit{input problem} and is shared between several quantum models approaching classical data  \cite{Aaronson_2015, Biamonte_Wittek_Pancotti_Rebentrost_Wiebe_Lloyd_2017}. As this problem is generally mitigated by assumption of an efficient QRAM, we add a dequantization scheme to even out the comparison. With this scheme we can find in which regimes the quantum models' improvements are maintained. 

To this end, we emphasize that the parameter-count reduction demonstrated throughout this work concerns the trainable variational component; any claim of an end-to-end time or memory advantage additionally depends on the data-loading, block-encoding, and readout assumptions discussed in Sec.~\ref{sec: Note on Complexity and Dequantization}. \\

In summary, the main contributions of this work are as follows:
\begin{itemize}
    \item Detailed implementation of two novel \gls{qgnn} models (\gls{sgc} and \gls{lgc}) in a quantum simulator.
    \item Experimental comparison of the implemented models and their classical versions for open benchmarking datasets.
    \item Detailed trainability analysis of the proposed QGCN models.
    \item Complexity comparison between the quantum and classical convolutional architectures and conditions for low-rank simulability.
\end{itemize}

The paper is organized as follows: Sec.~\ref{sec:classical_gnn} reviews the foundations of classical graph neural networks and introduces the convolutional perspective adopted throughout the work. Next, Sec.~\ref{sec:semi-sup-node} presents the semi-supervised node-classification setting and recalls the classical graph convolutional models that motivate the quantum constructions. Then, Sec.~\ref{sec:QGCN} introduces the proposed quantum graph convolutional architectures, with dedicated subsections for the \gls{sgc} and \gls{lgc} models. Afterwards, Sec.~\ref{sec:experiments_qgnn} describes the experimental setup and discusses the numerical results obtained on benchmark datasets. Sec.~\ref{sec: Trainability and Gradient Variance Bounds} provides a theoretical analysis of the trainability and the loss landscape of the proposed models. This is followed by a complexity comparison with the classical counterparts as well as a dequantization scheme to reason on the regimes of advantage of the QGCNs. Finally, Sec.~\ref{sec:conclusion_qgnn} summarizes the main conclusions and outlines future research directions.

\section{Classical Graph Neural Networks}
\label{sec:classical_gnn}

Following Refs.~\cite{bronstein_geometric_2021,velickovic_everything_2023,battaglia_relational_2018}, a brief introduction to classical Graph Neural Networks is provided, which serve as the foundation for the development of the proposed quantum GNNs. \\

Graphs are a natural way to represent complex systems of interacting entities. Formally, a graph $G = (V, E)$ consists of a set of nodes $V$ and a set of edges $E \subseteq V \times V$ that connect pairs of nodes. In many real-world applications, graphs are used to model relational structure, with nodes representing entities (e.g., users, proteins, web pages) and edges representing relationships or interactions between them (e.g., friendships, molecular bonds, hyperlinks). To enable rich feature representations, nodes are often endowed with attribute information in the form of real-valued feature vectors. Given a graph with $N = |V|$ nodes, we can summarize the node features as a matrix $X \in \mathbb{R}^{N \times C}$, where the $u$-th row $\mathbf{x}_u \in \mathbb{R}^C$ corresponds to the $C$-dimensional feature vector of node $u$. The connectivity of the graph can be represented by an adjacency matrix $A \in \mathbb{R}^{N \times N}$, where $a_{uv} = 1$ if there is an edge between nodes $u$ and $v$, and $a_{uv} = 0$ otherwise.\newline

Graph Neural Networks (GNNs) are a family of machine learning models that operate on the graph structure $(X, A)$. The key defining property of GNNs is \textit{permutation equivariance}. Formally, let $P \in \{0,1\}^{N \times N}$ be an permutation matrix. A GNN layer, denoted by $\textbf{F}(X, A)$, is a permutation-equivariant function in the sense that:
\begin{equation}
    \textbf{F}(PX, PAP^\top) = P\textbf{F}(X, A)
\end{equation}
Permutation equivariance is a desirable inductive bias for graph representation learning, as it ensures that the GNN output will be equivariant to arbitrary reorderings of the nodes. This property arises naturally from the unordered nature of graph data, i.e., a graph is intrinsically defined by its connectivity and not by any particular node ordering.\newline

 In each GNN layer, nodes update their features by aggregating information from their local neighborhoods ((undirected) neighborhood of node $u$ is defined as $\mathcal{N}_{u} = \{v | (u,v) \in E \text{ or } (v,u) \in E \}$). This local computation is performed identically (i.e., shared) across all nodes in the graph. Mathematically, a GNN layer computes a new feature matrix $H \in \mathbb{R}^{N \times C'}$ from the input features $X$ as follows:
\begin{equation}
H = \textbf{F}(X, A) = [\phi(\mathbf{x}_1, X_{\mathcal{N}_1}), \phi(\mathbf{x}_2, X_{\mathcal{N}_2}), ..., \phi(\mathbf{x}_N, X_{\mathcal{N}_N})]^T
\end{equation}
where $\phi$ is a local function often called the \textit{neighborhood aggregation} or \textit{message passing} function, and $X_{\mathcal{N}_{u}} = \{\!\!\{\mathbf{x}_v \ | \ v \in \mathcal{N}_{u}\}\!\!\} $ denotes the multiset of all neighborhood features of node $u$. In other words, the new feature vector $\mathbf{h}_{u}:=\phi(\mathbf{x}_u, X_{\mathcal{N}_u})$ of node $u$ is computed by applying $\phi$ to the current feature $\mathbf{x}_u$ and the features of its neighbors $X_{\mathcal{N}_{u}}$. Since $\phi$ is shared across all nodes and only depends on local neighborhoods, it can be shown that if $\phi$ is permutation invariant in $X_{\mathcal{N}_{u}}$, then $\textbf{F}$ will be permutation equivariant. Stacking multiple GNN layers allows information to propagate over longer graph distances, enabling the network to capture high-order interaction effects. \newline 

While the general blueprint of GNNs based on local neighborhood aggregation is quite simple and natural, there are many possible choices for the aggregation function $\phi$. The design and study of GNN layers is a rapidly expanding area of deep learning, and the literature can be divided into three \textit{flavors} \cite{bronstein_geometric_2021}: convolutional, attentional, and message-passing. These flavors determine the extent to which $\phi$ transforms the neighborhood features, allowing for varying levels of complexity when modeling interactions across the graph.\newline

This work will focus on convolutional models, despite the associated paper \cite{liao_graph_2024} including additional study on attentional and message-passing models. In the convolutional flavor \cite{kipf_semi-supervised_2017}, the features of the neighboring nodes are directly combined with fixed weights,
\begin{equation} \label{eq:graph_convolution}
\mathbf{h}_{u}=\phi\left(\mathbf{x}_{u}, \bigoplus_{v \in \mathcal{N}_{u}} c_{u v} \psi\left(\mathbf{x}_{v}\right)\right).
\end{equation}
Here, $c_{u v}$ is a constant indicating the significance of node $v$ to node $u^{\prime}$ s representation. $\bigoplus$ is the aggregation operator which is often chosen to be a simple summation. $\psi$ and $\phi$ are learnable transformations\footnote{Note the activation function in the original definition in \cite{bronstein_geometric_2021} is omitted,  as well as $\mathbf{b} $ in the quantum case, for simplicity.}: $\psi(\mathbf{x})=W \mathbf{x}+\mathbf{b}$, $\phi(\mathbf{x}, \mathbf{z})=W \mathbf{x}+U \mathbf{z}+\mathbf{b}$.  \newline

Classical GNNs have been shown to be highly effective in a variety of graph-related tasks including~\cite{wu_graph_2022,velickovic_everything_2023}:

\begin{enumerate}
    \item \textit{Node classification}: the goal is to assign labels to nodes based on their attributes and the graph structure. For example, in a social network, the task could be to classify users into different categories by leveraging their profile information and social connections. In a biological context, a canonical example is classifying protein functions in a protein-protein interaction network \cite{hamilton_inductive_2017}.

    \item \textit{Link prediction}: the goal is to determine whether an edge exists between two nodes, or to infer edge properties. In a social network, for instance, this translates to anticipating potential interactions between users; in a biological context, it can involve predicting associations between drugs and diseases — a task commonly known as drug repurposing.

    \item \textit{Graph classification}: the goal is to classify entire graphs based on their structures and attributes. A typical example is classifying molecules in terms of their quantum-chemical properties, which holds significant promise for applications in drug discovery and materials science \cite{gilmer_neural_2017}.
\end{enumerate}

\section{Semi-supervised node classification with Graph Convolutional Networks}\label{sec:semi-sup-node}

From the aforementioned applications, this work will hold to the problem of \textit{semi-supervised node classification}. In this scenario, all nodes in the graph belong to one of the possible categories; however, only a small subgroup of nodes is labeled.\footnote{For trainability purposes, the fraction of labeled nodes to all nodes is non-vanishing as the graph size scales --- see Sec.~\ref{sec: Trainability and Gradient Variance Bounds}.} The objective of a GNN is to process the nodes' features and their connectivity to infer the correct class each node belongs to (see Fig.~\ref{fig:semi-supervised}). In contrast to conventional multiclass classification tasks, where patterns in the features are learned to infer over new data, GNNs exploit relational information provided by the graph geometry. This allows GNNs to generalize the classification patterns from a few labeled instances and propagate them through the entire graph.

\begin{figure}[H]
    \centering    \includegraphics[width=1\linewidth]{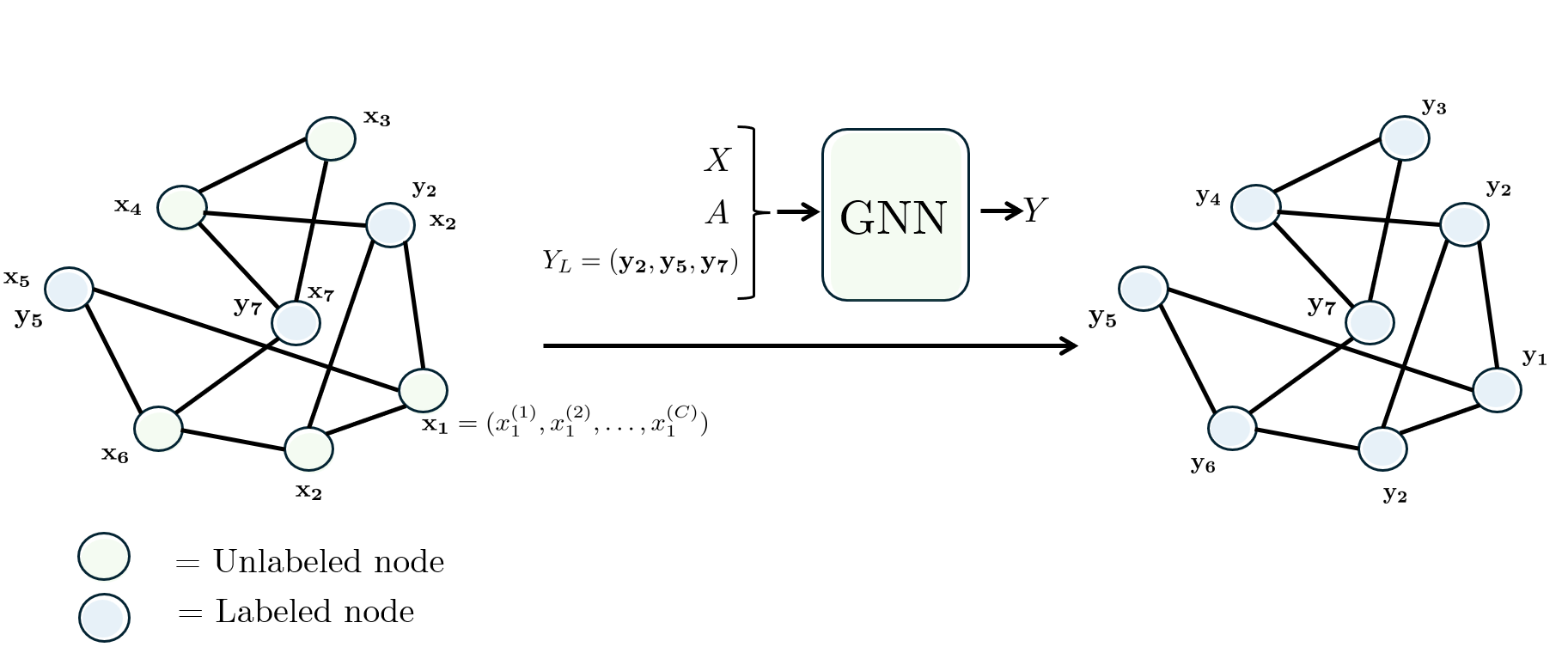}
    \caption{Schematic illustration of the semi-supervised node classification.}
    \label{fig:semi-supervised}
\end{figure}

Graph Convolutional Networks (GCNs) interleave graph convolutional layers, described in Eq.~\ref{eq:graph_convolution}, carried out  by the following matrix operation \cite{kipf_semi-supervised_2017}:
\begin{equation}
    H^{(l+1)}=\sigma\left(\hat{A} H^{(l)} W^{(l)}\right).
\label{gcneq}
\end{equation}
Here, $\hat{A}=\tilde{D}^{-\frac{1}{2}} \tilde{A} \tilde{D}^{-\frac{1}{2}}$ in which $\tilde{A}=A+I_{N}$ is the adjacency matrix of graph $G$ with added self-connections ($I_{N}$ is the identity matrix), $\tilde{D}_{i i}=\sum_{j} \tilde{A}_{i j}$, and $W^{(l)}$ is a layer-specific trainable weight matrix. $\sigma(\cdot)$ denotes a nonlinear activation function. This way, the GNN learns richer node representations $H^{(l)}$ with every layer. Naturally, in the first layer, the node representation corresponds to the input feature-matrix $H^{(1)}=X$. At the output of the last layer, the softmax function, defined as $\operatorname{softmax}\left(x_{i}\right)=\frac{1}{\mathcal{Z}} \exp \left(x_{i}\right)$ with $\mathcal{Z}=\sum_{i} \exp \left(x_{i}\right)$, is applied row-wise to the node feature matrix, producing the final output of the network:
\begin{equation}
    Z=\text{softmax}(\hat{A} H^{(p-1)} W^{(p-1)}).
\end{equation}
Here, $p$ represents the total number of layers applied. Each layer adds more trainable weights to the final cost and mixes the information of a node with its $p$-nearest neighboring nodes.

For semi-supervised multiclass classification, the cost function is defined by the cross-entropy error over all labeled examples \cite{kipf_semi-supervised_2017}:
\begin{equation} 
   \mathcal{L}=-\sum_{s \in {Y}_{L}} \sum_{f=1}^{F_K} Y_{s f} \ln Z_{s f},
   \label{cost}
\end{equation}
where ${Y}_{L}$ is the set of node indices that have labels, $Y\in \mathbb{B}^{N \times F_K}$ denotes the one-hot encoding of the labels.

\section{Quantum Graph Convolutional Networks}\label{sec:QGCN}

An algorithmic design to leverage quantum circuits to create GCNs is proposed in this section. The models comprising this design are termed \textit{Quantum Graph Convolutional Networks} (QGCNs). The objective of the quantum circuit is to implement the graph convolution operation, described in Eq.~\ref{gcneq}, to the node feature matrix $X$. To do so, one must encode the data into the quantum system.

\subsubsection{Data encoding}\label{de}

The GCN model's node features $X \in \mathbb{R}^{N \times C}$, whose entries are denoted by $X_{ik}$, can be encoded in a quantum state $\left|\psi_{X}\right\rangle$ (after normalization)\footnote{Note that the normalization factors in quantum states is often omitted during this work.} as follows:
\begin{equation}\label{eq: input encoding}
\left|\psi_{X}\right\rangle=\sum_{i=1}^{N} |i\rangle\ket{\mathbf{x}_{i}} = \sum_{i=1}^{N} \sum_{k=1}^{C}X_{ik}\ket{i}\ket{k}.
\end{equation}
The entire state is prepared on two quantum registers hosting the channel index $k$ and node index $i$, which are denoted as $Reg(k)$ and $Reg(i)$, respectively. Therefore, this encoding corresponds to an \textit{amplitude encoding}, of the vectorized node feature matrix. In this work, the method from Ref.~\cite{zhang_circuit_2024} is chosen, as their work provides a tunable trade-off between the number of ancillary qubits and the circuit depth for state preparation. Note that $Reg(k)$ initially hosts the $C$-dimensional feature index; after the layer-wise transformations described in Sec.~\ref{mgcn}, the same register is used as a readout register, and the $F_K$ class labels are read out by projecting onto $F_K$ of its computational basis states (see Eq.~\ref{eq: target encoding} and the prediction rule introduced later in Sec.~\ref{sec:QGCN}).

\subsubsection{Layer-wise transformation}\label{mgcn}

The layer-wise linear transformation for a multi-channel GCN (i.e. $H'^{(l)}=\hat{A} H^{(l)} W^{(l)}$), can be implemented by applying the block-encoding of $\hat{A}$ and a parameterized quantum circuit implementing $W^{(l)}$ on the two quantum registers $Reg(i)$ and $Reg(k)$ respectively, as depicted in Fig.~\ref{fig:qgnn_circuit}. This is proven following the \textit{Kronecker-vec identity} \cite{magnus_matrix_2019}:
\begin{equation}
    vec(ABC) = \left( C^T \otimes A\right)vec(B) \longrightarrow vec\left({H'^{(l)}}^T\right) = \left( \hat{A} \otimes {W^{(l)}}^T  \right)vec\left({H^{(l)}}^T\right),
\end{equation}
where $vec(\cdot)$ operator transforms a matrix into a column vector by stacking its columns on top of one another, from left to right.\\

\begin{figure}
    \centering\includegraphics[width=0.8\linewidth]{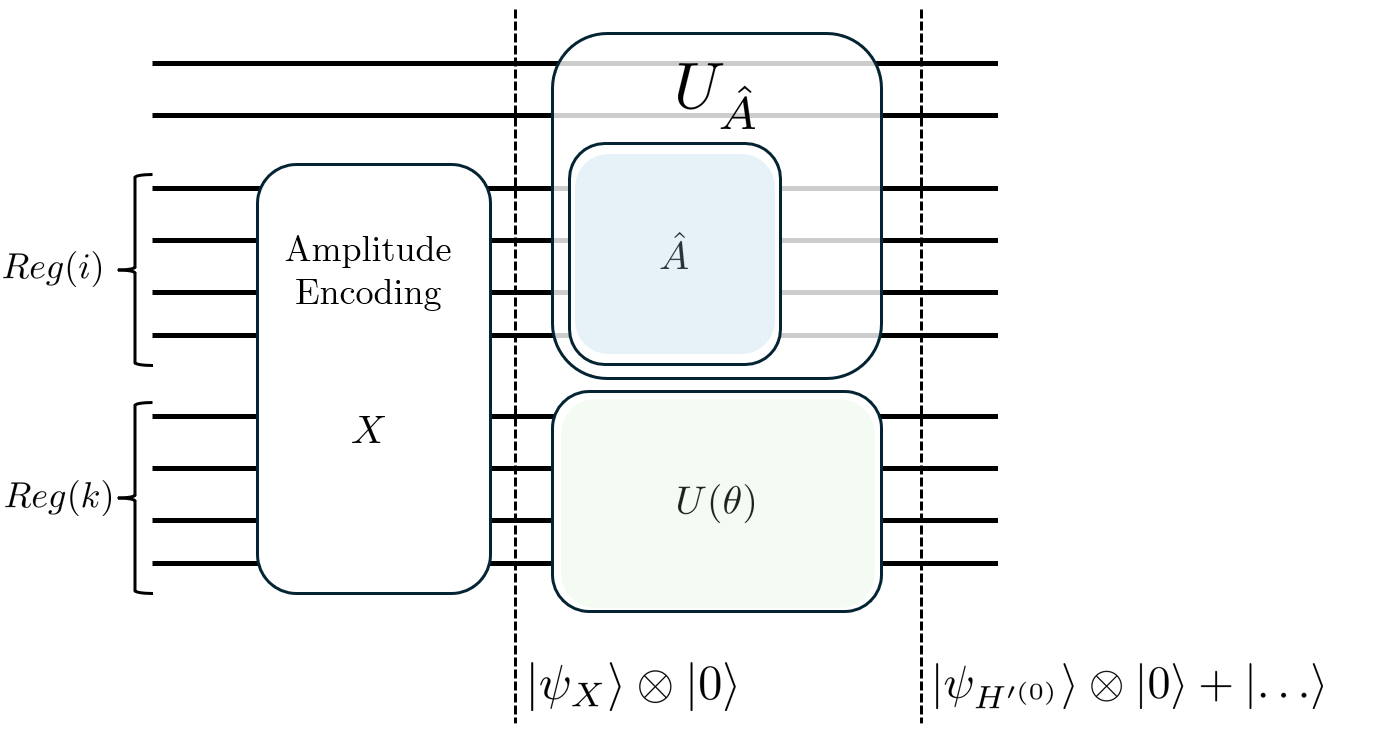}
    \caption{Quantum circuit implementing the proposed graph convolution operation.}
    \label{fig:qgnn_circuit}
\end{figure}

It has been observed in many experiments that deeper GNNs do not always yield to an improved performance and can even lead to worse outcomes than shallow models, as a consequence of a problem called \textit{oversmoothing} \cite{keriven_not_2022}. Accordingly, Ref.~\cite{kipf_semi-supervised_2017} considered a two-layer GCN where the nonlinear activation function is applied only once and the forward model takes the following form:
\begin{equation}
Z = \text{softmax}\left(\hat{A} \sigma\left(\hat{A}XW^{(0)}\right) W^{(1)}\right). 
\end{equation}
Taking this model as reference, we describe the state evolution for the quantum analogue of a two-layer GCN. Denote the block-encoding of $\hat{A}$ by $U_{\hat{A}}$ and the parameterized quantum circuit implementing $W^{(0)}$ by $U_{W^{(0)}}$. Applying these operations to the quantum state $\ket{\psi_X}$ yields
\begin{equation}
 \left(U_{\hat{A}} \otimes U_{W^{(0)}}\right) \ket{\psi_X}\otimes \ket{0} = \ket{\psi_{H'^{(0)}}} \otimes \ket{0} + \ket{\ldots} ,
\end{equation}
where $\ket{\psi_{H'^{(0)}}}=\sum_{i=1}^{N} \ket{i}\ket{\mathbf{h'}^{(0)}_{i}} = \sum_{i=1}^{N} \sum_{k=1}^{C}H'^{(0)}_{ik} \ket{i}\ket{k}$ is the amplitude encoding of the linearly transformed node features for node $i$ over the feature index $k\in\{1,\dots,C\}$. The term ``$\ket{\ldots}$''\footnote{Throughout this paper, the terms ``$\ket{\ldots}$'' in the quantum states are consistently used as defined here.} denotes a quantum state orthogonal to the first term in the sum. The block-encoding of $\hat{A}$ aggregates neighboring node features, while $U_{W^{(0)}}$ applies the trainable weight matrix to the node features.\newline

After the linear layer-wise transformation, the element-wise nonlinear activation function can be applied using an established technique called Nonlinear Transformation of Complex Amplitudes (NTCA) \cite{guo_nonlinear_2024}. One can also potentially utilize the techniques from Ref.~\cite{rattew_non-linear_2023} to apply the nonlinear activation function and achieve better performance. In our case, an explicit quantum implementation of nonlinear activation functions is left for future work.\\

By using a nonlinear activation function on the amplitudes of the state $\ket{\psi_{H'^{(0)}}} \otimes \ket{0} + \ket{\ldots}$, the state $\ket{\psi_{H^{(1)}}} \otimes \ket{0} + \ket{\ldots}$ is obtained. Finally, applying $U_{\hat{A}}$ and $U_{W^{(1)}}$ to the state $\ket{\psi_{H^{(1)}}} \otimes \ket{0} + \ket{\ldots}$ yields the output state $\ket{\psi_{\text{out}}}=\ket{\psi_{H'^{(1)}}} \otimes \ket{0} + \ket{\ldots}$.

\subsubsection{Cost function}

Having computed the feature state of the last layer $\ket{\psi_{\text{out}}}$, to evaluate the loss function defined as in Equation \ref{cost}, full state tomography must be performed to recover all $Z_{sf}$ values that are encoded in the amplitudes of computational basis states. As this can be prohibitively expensive for real quantum devices \cite{schmied_quantum_2014}, we instead define a new cost function:
\begin{equation} \label{eq:inner_product_loss}
\mathcal{L}_{QGCN} = -\mathrm{Re}\big(\left\langle\psi_{out}|\psi_Y\right\rangle\big).
\end{equation}
We stress that this overlap-based loss is a surrogate for the classical row-wise softmax cross-entropy of Eq.~\ref{cost}, chosen because it can be estimated efficiently via the Modified Hadamard Test described below. The quantum model does not implement the classical cross-entropy loss directly; final class labels are instead inferred from projections of the output state onto class basis states, as described later in this section.\\

For this, the one-hot-encoded target state
\begin{equation}\label{eq: target encoding}
\ket{\psi_Y} = \sum_{s \in Y_L} \sum_{f=1}^{F_K} Y_{sf} \, \ket{s}\ket{f}\otimes \ket{0}
\end{equation}
can be efficiently amplitude-encoded using the algorithm proposed in Ref.~\cite{shukla_efficient_2024}. This cost function can be evaluated via the \textit{Modified Hadamard Test}~\cite{knorzer_cross-platform_2023, luongo_quantum_2022}. Assuming that $\ket{\psi_1}= U_1\ket{0}$ and $\ket{\psi_2}= U_2\ket{0}$, an ancillary qubit is prepared in a uniform superposition using a Hadamard gate. Then, both $U_1$ and $U_2$ are applied, each controlled by a different branch of the ancilla. Finally, another Hadamard gate is applied to the ancilla, resulting in the state (see Fig.~\ref{fig:modified_hadamard_test}):
\begin{equation}
    \ket{\psi} = \frac{1}{2}\Bigl[ \ket{0}(\ket{\psi_1}+\ket{\psi_2})+\ket{1}(\ket{\psi_1}-\ket{\psi_2})\Bigr].
\end{equation}
Now, it is straightforward to see that the probability of observing the ancilla qubit in the state $\ket{0}$ is
\begin{equation}
    p(0)=\frac{1}{4}(\bra{\psi_1}+\bra{\psi_2})(\ket{\psi_1}+\ket{\psi_2}) = \frac{2+2\text{Re}[\braket{\psi_1|\psi_2}]}{4}.
\end{equation}

\begin{figure}[H]
    \centering
    \includegraphics[width=0.85\linewidth]{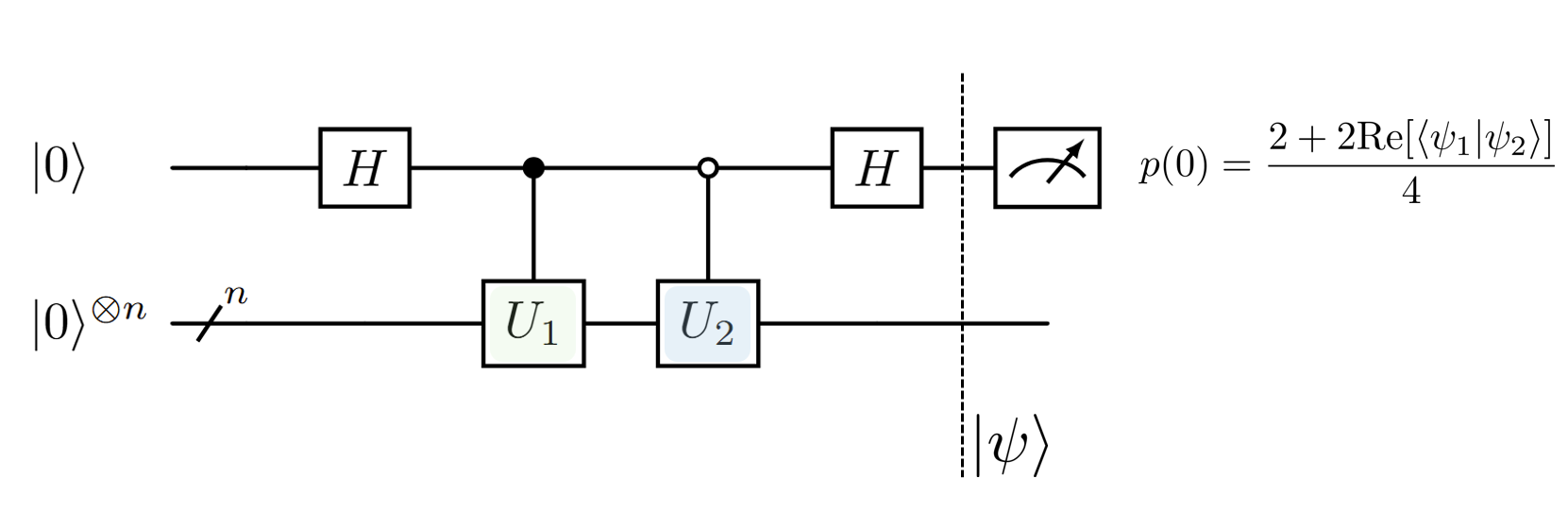}
    \caption{Schematic representation of the quantum circuit applying the Modified Hadamard Test.}
    \label{fig:modified_hadamard_test}
\end{figure}

\paragraph{Note:} One can alternatively find the modulus squared of the overlap,
\begin{equation}
    \mathcal{L}_{QGCN}^\text{alternative}\equiv |\braket{\psi_1|\psi_2}|^2.
\end{equation}
For this, the imaginary part of the overlap can also be measured by adding an $S^\dagger$ phase gate to the ancilla after the first Hadamard gate; consequently, this variant of the loss can be reconstructed from the two measurements. This protocol is related to the \textit{Swap test}, but applies both controlled unitaries to the same register, reducing the number of qubits at the cost of increased circuit depth. If depth must be minimized, the Swap test can be employed, yielding an equivalent estimate of the overlap.

\subsection{Simplified Graph Convolutional Network}

Since all the transformation of quantum states are unitary transformations, and therefore, linear transformations, the application of activation functions is not an easy task for quantum computers. For instance, even if NTCA is a reasonable proposal, it adds a significant computational cost to the algorithm. For this reason, and in order to implement models that can be simulated in classical computers or run in \gls{nisq} devices, activation functions are omitted from this point forward throughout this work.\\

By removing all intermediate activation functions, all linear operations collapse into the following transformation:
\begin{equation}
    \sigma\left(\hat{A}\sigma\left(\hat{A}\sigma(\cdots\hat{A}XW^{(1)})\cdots W^{(p-1)} \right)W^{(p)} \right)\rightarrow\hat{A}^pXW,
\end{equation}
where $W=W^{(1)}W^{(2)}\dots W^{(p)}$. Since all weight matrices are trainable, we treat their product as a single trainable matrix and initialize $W$ directly. This kind of transformation leads to a well-known GCN variant called \textit{\acrfull{sgc}}, which shows competitive results despite its low model complexity \cite{wu_simplifying_2019, maekawa_beyond_2022,pasa_simple_2021}:
\begin{equation}
Y_{\text{SGC}} = \text{softmax}(\hat{A}^p X W).
\end{equation}
As $p$ now only represents how far information is mixed across the graph, we call this the \textit{propagation order}, or $p$\textit{-hop} hereafter. \\

The quantum alternative proposed here (QSGC) applies the transformation described in Sec.~\ref{mgcn} to the encoded feature matrix:
\begin{equation}
    \ket{\psi_{\text{out}}} =  \left( U_{\hat{A}^p}  \otimes U(\theta) \right) \ket{\psi_X},
\end{equation}
where $U_{\hat{A}^p}$ can be implemented by directly block-encoding the precomputed matrix $\hat{A}^p$, or by applying \gls{qsvt} to compute the power $p$ of the normalized adjacency matrix\footnote{For simplicity, the first method is used in the experimental part of this paper.}. To keep the number of ancilla qubits to a minimum, the matrix $A$ can be encoded as
\begin{equation}
    U_A = \begin{bmatrix}
        A & \sqrt{I_N-AA^\dagger}\\
        \sqrt{I_N-A^\dagger A} & -A^\dagger
    \end{bmatrix}.
\end{equation}
The efficient circuit implementation of the block-encoding protocol is still an open and active research area. We refer to \cite{camps_explicit_2023, yang_dictionary-based_2025} for the efficient block-encoding of certain sparse and structured matrices.\\

The weight matrix $U(\theta)$ is implemented by a \gls{pqc}, introducing one subtle difference between the classical and quantum models: $W$ can be any matrix in $\mathbb{R}^{C\times F_K}$, whereas $U(\theta)$ is a unitary matrix in $SU(2^{n_k})$. Here, $F_K$ denotes the number of classes and $n_k = \lceil\log_2 C\rceil$ is the number of qubits in $Reg(k)$.  We emphasize that $U(\theta)$ is therefore not a direct implementation of the classical linear map $W$. It replaces $W$ by a unitary variational feature transformation followed by class readout described below.\\

Once trained with the loss function in Eq.~\ref{eq:inner_product_loss}, the predicted label of node $i$ can be inferred as
\begin{equation}
    y_i = \arg \max_f\quad \mathrm{Re}\big((\bra{i}\bra{f})\ket{\psi_{\text{out}}} \big) =  \arg \max_f \; \mathrm{Re}\big({H'}^{(p)}_{if}\big), \quad f\in[F_K],
\end{equation}
estimated via the Modified Hadamard Test. Note that the index $f$ associated to the class labels, iterates along computational basis states of $Reg(k)$, which requires that $F_K\leq C$. \\

\begin{figure}[!htbp]
    \centering
    \includegraphics[width=0.85\linewidth]{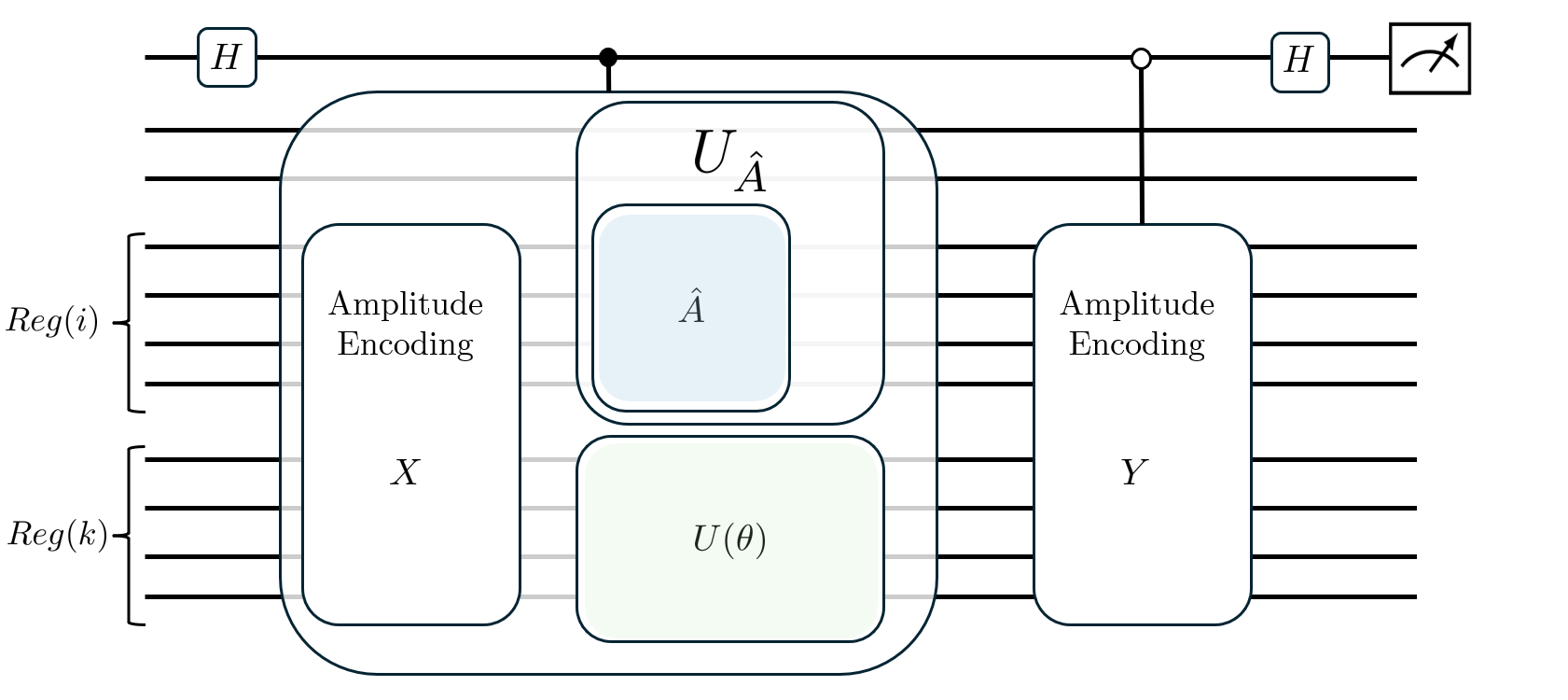}
    \caption{Quantum circuit implementing the QSGC network. The cost function is evaluated via the Modified Hadamard Test.}
    \label{fig:QSGC_circuit}
\end{figure}

Finally, the power $p$ of the normalized adjacency matrix is treated as a hyperparameter controlling the propagation of neighbor features through the graph. Fig.~\ref{fig:QSGC_circuit} shows the quantum circuit implementing the QSGC network.

\subsection{Linear Graph Convolutional Network}\label{sec: Linear Graph Convolutional Network}

Extending the \gls{sgc}, Pasa et al.~\cite{pasa_empowering_2024} proposed a more expressive variant, called \textit{\acrfull{lgc}}:
\begin{equation}
H = \sum_{i=0}^p \alpha_i L^i XW
\end{equation}
where $L\equiv I-\hat{A}$ is the normalized Laplacian matrix of the graph and $\alpha_i \in \mathbb{R}$ are treated as learnable parameters. The quantum version can be implemented as
\begin{equation}
    \ket{\psi_{\text{out}}} =  \left( U_{P(L)}  \otimes U(\theta) \right) \ket{\psi_X},
\end{equation}
where $P(L) = \sum_{i=0}^p \alpha_i L^i$ is a polynomial in the Laplacian matrix. By allowing multiple learnable weighting coefficients $\alpha_i$ for each $L^i$ up to order $p$, LGCs can represent a much richer class of graph convolution filters compared to SGCs. This increased expressivity enables LGCs to capture more complex graph structures and long-range dependencies, leading to improved performance on certain downstream tasks.\\ 

There are several ways to block-encode this polynomial. The simplest approach is to classically compute $P_K(L)$ and directly block-encode it. However, because the advantage of the quantum model arises when the graph is too large to handle classically, it is preferable in that regime to evaluate the polynomial within the quantum circuit. A second option is to implement a \textit{Linear Combination of Unitaries (LCU)} \cite{childs_hamiltonian_nodate}, but then each power $L^i$ must still be computed classically and block-encoded independently. We therefore opt to leverage \textit{Quantum Singular Value Transformation (QSVT)} \cite{gilyen_quantum_2019} to evaluate the polynomial directly in the quantum circuit.

The QSVT algorithm nevertheless imposes restrictions on the polynomials it can generate. First, it requires the polynomial to have a defined parity\footnote{The \textit{parity} of a polynomial implies that it only contains even or odd indices with non-zero terms}. Second, $\|P_K(L)\|_2\leq 1$ must be satisfied to ensure that $P_K(L)$ is a valid sub-unitary to block-encode. The first condition can be addressed by separating the polynomial into the sum of odd and even terms,
\begin{equation}
    P_K(L) = P_K^{\textrm{odd}}(L) + P_K^{\textrm{even}}(L),
\end{equation}
and combining both using an LCU with a single ancilla qubit. Regarding the second condition, even if a valid configuration (and scaling) of coefficients $\alpha_i$ is found, these coefficients are trainable parameters, so it is easy to violate this constraint during gradient descent. To mitigate such violations, we propose optimizing the Quantum Signal Processing phases $\phi_k$ instead:
\begin{equation}
    \Bigg[\prod_{k=1}^{\lfloor p/2\rfloor } \Pi_{\phi_{2k-1}} U_L^{\dagger} \tilde{\Pi}_{\phi_{2k}} U_L \Bigg]\Pi_{\phi_{p+1}} = \begin{pmatrix}
        P_K^{\textrm{even}}(L) & * \\
        * & *
    \end{pmatrix}
\end{equation}

\begin{equation}
    \tilde{\Pi}_{\phi_{1}} U_L\Bigg[\prod_{k=1}^{\lfloor( p-1)/2\rfloor} \Pi_{\phi_{2k}} U_L^{\dagger} \tilde{\Pi}_{\phi_{2k+1}} U_L \Bigg]\Pi_{\phi_{p+1}} = \begin{pmatrix}
        P_K^{\textrm{odd}}(L) & * \\
        * & *
    \end{pmatrix},
\end{equation}
where $U_L$ is the block-encoding of L. This way, even if the track of the produced polynomial is blurred, it is ensured that the outcome fulfills the required conditions. The application of $U(\theta)$ and the evaluation of the cost is performed in the same way as in SGC. Fig.~\ref{fig:QLGC_circuit} depicts the quantum circuit implementing the proposed algorithm.

\begin{figure}[H]
    \centering
    \includegraphics[width=0.9\linewidth]{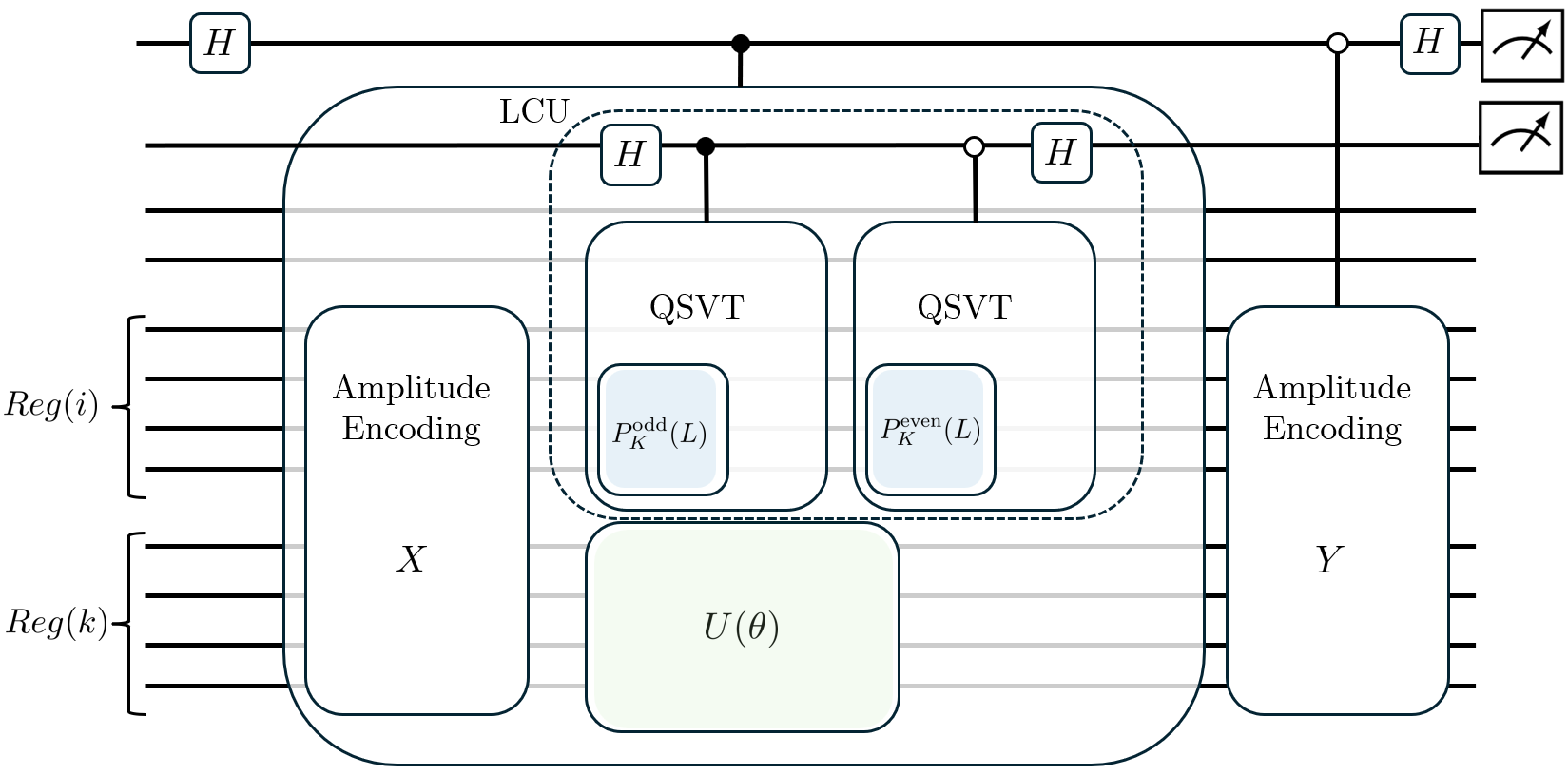}
    \caption{Quantum circuit implementing QLGC network. The polynomial of $L$ is split into odd and even parity polynomials implemented by QSVT and added by LCU. The cost function is evaluated by the Modified Hadamard test.}
    \label{fig:QLGC_circuit}
\end{figure}

\section{Experiments}\label{sec:experiments_qgnn}

\subsection{Experimental setup}

Two types of quantum models --- the \textit{Quantum Simplified Graph Convolutional (QSGC)} and \textit{Quantum Linear Graph Convolutional (QLGC)} networks --- are considered, together with their respective classical counterparts, termed \textit{CSGC} and \textit{CLGC}. Quantum models are implemented using \textit{Pennylane} \cite{bergholm_pennylane_2022} and run in a noiseless simulator, employing \textit{JAX} \cite{jax2018github} to accelerate the underlying linear algebra. Classical models are implemented using \textit{PyTorch} \cite{paszke_pytorch_2019}. For the QLGC experiments, we simulate the effective polynomial graph filter $P(\hat{A})$ resulting from the QSVT/LCU construction of Sec.~\ref{sec: Linear Graph Convolutional Network}, rather than a fully compiled, gate-level QSVT/LCU circuit. The ancilla and depth counts reported in Sec.~\ref{sec: Note on Complexity and Dequantization} refer to such a compiled circuit, not to the simulator used here.\\

\begin{table}[H]
\centering

\begin{tabular}{@{}lrrrcrrr@{}}
\toprule
\textbf{Name} & \textbf{Nodes} & \textbf{Features} & \textbf{Classes} & \textbf{Train / Val / Test} & Reference\\
\midrule
Karate Club & 34 & 34 & 4 & 4 / 10 / 20 & \cite{zachary_information_1977}\\
Cora & 2{,}708 & 1{,}433 & 7 & 140 / 500 / 1,000 & \cite{yang_revisiting_2016} \\
Texas & 183 & 1{,}703 & 5 & 109 / 37 / 37 & \cite{pei_geom-gcn_2020}  \\
Wisconsin & 251 & 1{,}703 & 5 & 150 / 50 / 51 & \cite{pei_geom-gcn_2020} \\
Cornell & 183 & 1{,}703 & 5 & 109 / 37 / 37 & \cite{pei_geom-gcn_2020} \\
\bottomrule
\end{tabular}
\caption{Summary of the datasets employed in the numerical experiments.}
\label{tab:Datasets}
\end{table}

All models are trained for the datasets listed in Tab.~\ref{tab:Datasets}. In the case of the graph \textit{Karate Club}, no features are provided originally for nodes. For this reason, nodes are one-hot-encoded and the data matrix $X=I_N$. In those cases, the amplitude-encoded quantum state can be efficiently prepared in $O(\log_2 N)$ depth and space \cite{shukla_efficient_2024}.\\

As variational ansatz in the quantum models, the widely spread \textit{Strongly Entangling Layers} \cite{schuld_circuit-centric_2020} architecture is employed, adding $3$ optimizable parameters per qubit and per ansatz layer. Those variational parameters are optimized with the classical optimizer \textit{Adam} \cite{kingma_adam_2017} using the same learning rate for all models, but ensuring their convergence to a minimum. A brief summary of the models is provided in Table \ref{tab:summary_models}.\\

\begin{table}[H]
    \centering
    \begin{tabular}{|c|c|c|c|c|}
\hline
 & \multicolumn{2}{c|}{Simplified} & \multicolumn{2}{c|}{Linear}\\
 \hline
Type & Quantum & Classical & Quantum & Classical \\
\hline
Name & QSGC & CSGC & QLGC & CLGC \\
\hline
Linear Transformation & \multicolumn{2}{c|}{$\hat{A}^p XW$} & \multicolumn{2}{c|}{$\sum_{i=0}^p\alpha_iL^iXW$}  \\
\hline
Parameters & $3\ell\lceil\log_2C\rceil $ & $CF_K$ & $3\ell\lceil\log_2C\rceil +p$ & $CF_K+p$\\
\hline
\end{tabular}
    \caption{Summary of the type of models implemented for the experiments. $p$ is treated as a hyperparameter, that controls the (maximum) power of the normalized adjacency matrix. $\ell$ is the number of layers in the quantum ansatz, $C$ is the number of features of nodes and $F_K$ is the number of classes. }
    \label{tab:summary_models}
\end{table}

\subsection{Hyperparameter selection:} We evaluate the classification performance of classical (CSGC, CLGC) and quantum (QSGC, QLGC) models across the five benchmark datasets given in Table \ref{tab:Datasets}. Table~\ref{tab:score_table} presents the comparative analysis, displaying mean accuracy, Macro F1, and Geometric Mean (G Mean) scores averaged over 5 random initializations. For each dataset and architecture we perform a grid search over the hyperparameter configuration $(p, \ell)$ and report the corresponding results. Here $p$ represents the propagation order and $\ell$ denotes the number of ansatz layers. The overall top performing model for each dataset is highlighted in \textbf{bold}, and the top performing quantum model is \underline{underlined}.\\

The selection of the $p$-hop propagation order ($p \in \{1, 2, 3, 4, 8, 16\}$) dictate the amount of global graph feature aggregation. The nonlinear progression ($1, 2, 3, 4, 8, 16$) efficiently spans both localized and deep structural regimes.
\begin{itemize}
    \item \textbf{Low ($p \in \{1, 2\}$):} Captures immediate local neighborhood interactions. Good for highly homophilous graphs (such as Cora), where direct neighbors share identical labels and higher orders risk over smoothing.
    \item \textbf{Intermediate ($p \in \{3, 4\}$):} Captures more global structure, such as cluster boundaries. Serves a transition between localized and global feature mixing.
    \item \textbf{High-order ($p \in \{8, 16\}$):} Captures long-range structural dependencies. This is crucial for heterophilous graphs (Texas, Wisconsin, Cornell), where immediate neighbors mainly belong to different classes, requiring distant node features to classify.
\end{itemize}

The selection of the quantum ansatz depth ($\ell \in \{1, 2, 4, 8, 12, 16\}$) determines the expressivity of the PQC by controlling the entangling capacity and parameter count of the model, while still being reasonably efficiently computable.
\begin{itemize}
    \item \textbf{Shallow ($\ell \in \{1, 2, 4\}$):} Provide low-parameter, low-entanglement baselines. Shallow circuits suppress over-parameterization, mitigating severe overfitting on small datasets (such as Karate Club and WebKB), and tend to avoid high loss concentration.
    \item \textbf{Moderate ($\ell = 8$):} A middle ground balancing sufficient expressivity and feature entanglement against against optimization complexity.
    \item \textbf{Deep ($\ell \in \{12, 16\}$):} Maximize expressivity and allows empirical identification of the depth threshold where noise, trainability bottlenecks, or capacity saturation begin to cause performance degradation.
\end{itemize}

We fix the learning rate to $\eta = 0.01$ for stable convergence. For the classical models we set the $L_2$ penalty to $5 \times 10^{-4}$ to avoid parameter drift during backpropagation. Finally, we initialize all PQC parameters uniformly from $\mathcal{U}(0, 2\pi)$ covering the complete domain of single-qubit Bloch sphere rotations. We perform this initialization, and the subsequent evaluation, across 5 distinct random seeds for each hyperparameter combination to ensure that results shown in the experimental figures reflect real generalization rather than anomalies.

\subsection{Numerical results}

The results of numerical experiments are presented in Tab. \ref{tab:score_table}, after an extensive training phase of various models with different hyperparameters and for different datasets. The distribution of the accuracy score of the best performing hyperparameter combinations is depicted in Fig.~\ref{fig:boxplotz}.

\begin{table}[!htb]
    \centering
    \begin{tabular}{llccccc}
    \toprule
    Dataset & Model & $p$ & $\ell$ & Accuracy & Macro F1 & G Mean \\
    \midrule
    \multirow{4}{*}{Karate} & QSGC & 3 & 8 & \textbf{\underline{0.912}} & 0.909 & 0.938 \\
    & CLGC & 4 & --- & 0.882 & 0.879 & 0.872 \\
    & QLGC & 2 & 16 & 0.882& 0.867 & 0.873 \\
    & CSGC & 4 & --- & 0.853 & 0.859 & 0.849 \\
    \midrule
    \multirow{4}{*}{Texas} & CLGC & 1 & --- & \textbf{1.000} & 1.000 & 1.000 \\
    & QLGC & 2 & 16 & \underline{0.942} & 0.869 & 0.622 \\
    & QSGC & 2 & 16 & 0.879 & 0.665 & 0.003 \\
    & CSGC & 2 & --- & 0.812 & 0.579 & 0.003 \\
    \midrule
    \multirow{4}{*}{Wisconsin} & CLGC & 3 & --- & \textbf{0.994} & 0.996 & 0.994 \\
    & QLGC & 3 & 16 & \underline{0.985} & 0.974 & 0.958 \\
    & QSGC & 2 & 16 & 0.953 & 0.937 & 0.910 \\
    & CSGC & 1 & --- & 0.718 & 0.612 & 0.538 \\
    \midrule
    \multirow{4}{*}{Cornell} & CLGC & 1 & --- & \textbf{1.000} & 1.000 & 1.000 \\
    & CSGC & 1 & --- & 0.977 & 0.968 & 0.962 \\
    & QLGC & 1 & 16 & \underline{0.897} & 0.876 & 0.850 \\
    & QSGC & 4 & 16 & 0.866 & 0.832 & 0.769 \\
    \midrule
    \multirow{4}{*}{Cora} & CSGC & 8 & --- & \textbf{0.788} & 0.779 & 0.795 \\
    & QSGC & 8 & 16 & \underline{0.678} & 0.675 & 0.710 \\
    & CLGC & 8 & --- & 0.658 & 0.645 & 0.668 \\
    & QLGC & 1 & 16 & 0.376 & 0.370 & 0.392 \\
    \bottomrule
    \end{tabular}
    \caption{
    Comparative analysis of classical (CSGC, CLGC) and quantum (QSGC, QLGC) Graph Convolutional Networks. Results display mean accuracy, Macro F1, and Geometric Mean (G Mean) scores across five datasets, averaged over 5 initializations. The best performing model for each dataset is in \textbf{bold} and best performing quantum model is \underline{underlined}. Parameters $p$ and $\ell$ denote the $p$-hop order and number of ansatz layers, respectively.}
    \label{tab:score_table}
\end{table}

\begin{figure}
    \centering
    \includegraphics[width=1\linewidth]{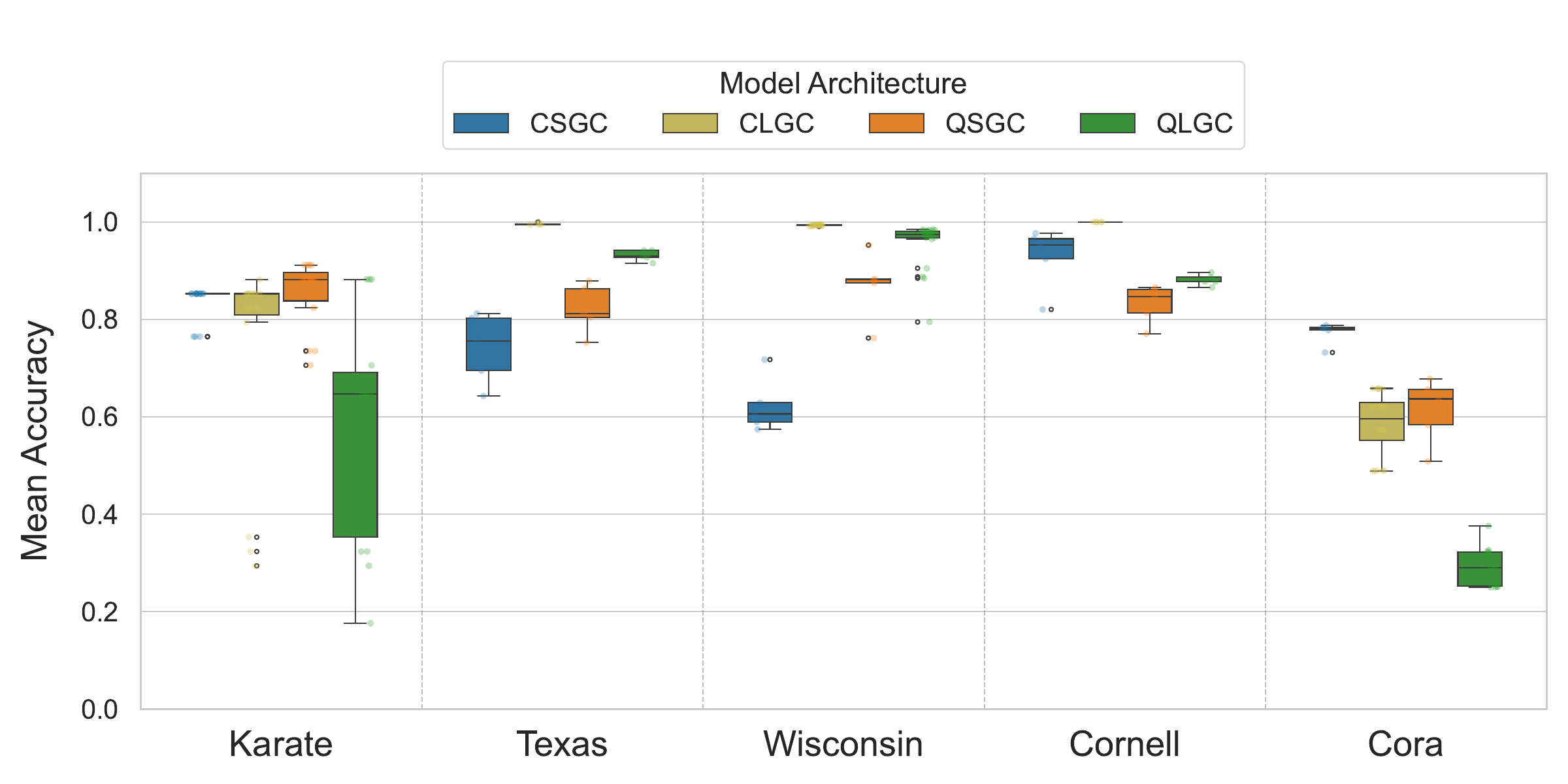}
    \caption{Distribution of model accuracy across all experimental runs for best performing hyperparameter configurations. The box plots compare, in order, classical (CSGC, CLGC) and quantum (QSGC, QLGC) architectures across five datasets. Error bars and data points indicate the stability and variance of each model type over multiple random initialization. While quantum models (orange, green) utilize fewer parameters they maintain competitive accuracies and stability compared to their classical counterparts (blue, yellow).}
    \label{fig:boxplotz}
\end{figure}

From these results, two main conclusions can be extracted. First, the proposed quantum models are able to learn, and therefore, the results validate the theoretical proposals. The second conclusion is that the quantum model can approximate the performance of the classical ones, even surpass it in some cases. This is remarkable, considering that the quantum models achieve this using substantially fewer trainable parameters in the variational component (see \ Table~\ref{tab:summary_models}).\\

We also note that some high-accuracy results are accompanied by a G Mean close to zero, e.g.\ QSGC and CSGC on Texas (G Mean $0.003$, Table~\ref{tab:score_table}). This indicates a collapse in recall for at least one minority class despite a competitive overall accuracy. We therefore report Macro F1 and G Mean alongside accuracy throughout, and interpret cases where accuracy and G Mean disagree cautiously rather than describing them simply as competitive.\\

Fig.~\ref{fig:acc_hyperparameters} shows the influence of hyperparameter in the performance of the models. The same results are depicted in a two-dimensional heatmap in Fig.~\ref{fig:heatmap}. It can be seen that the performance of the quantum models is improved with the number of layers in the ansatz. This happens because adding layers increases the expressivity of the parameterized quantum circuit and this adds degrees of freedom in the weight matrix that the model can generate in the graph convolution. When it comes to the power $p$ of the adjacency matrix (or the degree of polynomial), increasing its value does not always result in an improvement of the model. On the contrary, it starts with an initial increase, but it suffers a downgrade after an early maximum. This behavior, however, is dependent on the dataset and we believe it has a strong relation with the \textit{oversmoothing} \cite{keriven_not_2022} and \textit{overfitting} problems.  \\

\begin{figure}[H]
    \centering
    \includegraphics[width=1\linewidth]{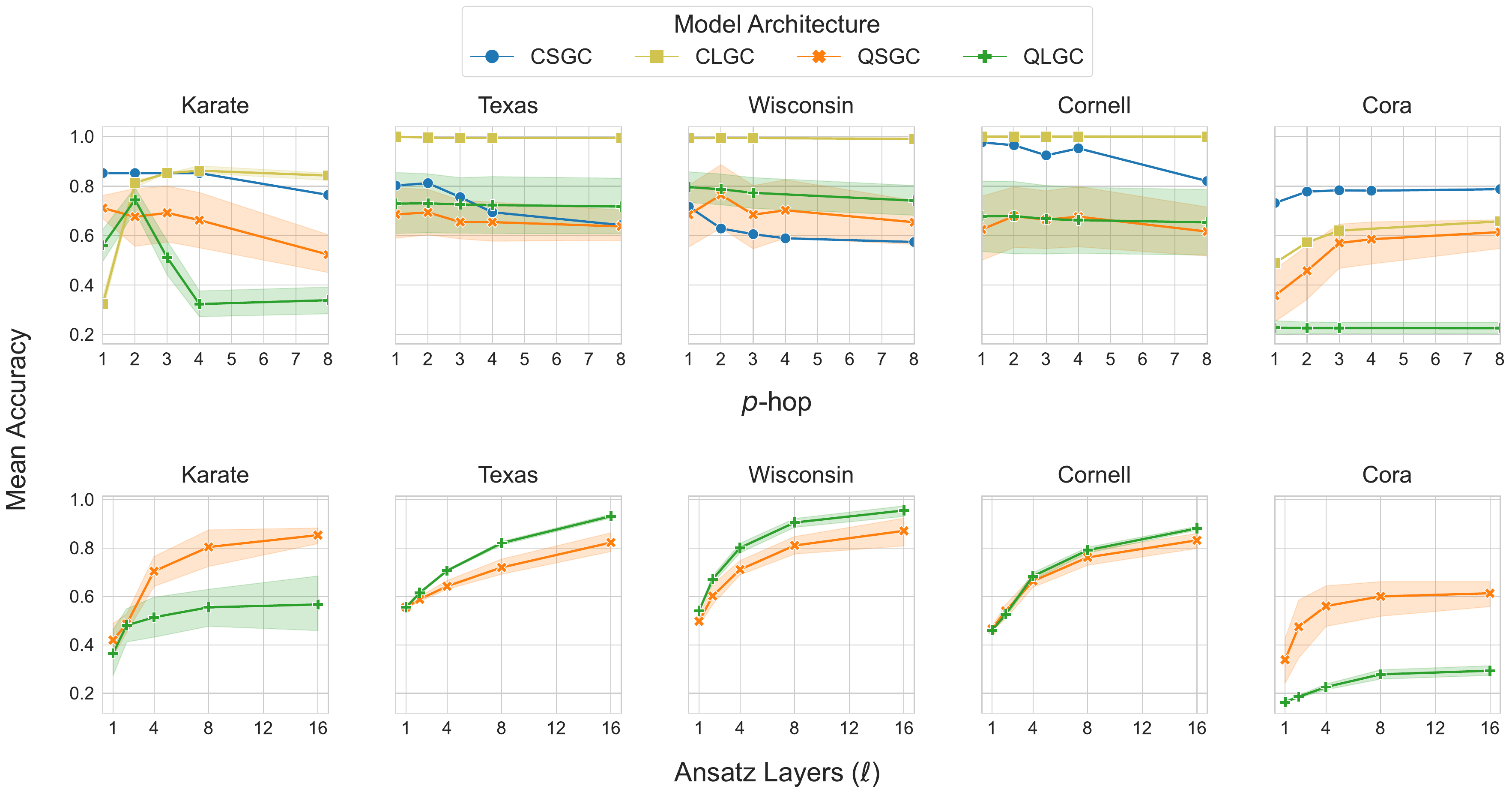}
    \label{fig:sensitivity}
    
 \caption{Sensitivity analysis of model accuracy with respect to different hyperparameters across five datasets. The top row displays the effect of the $p$-hop order, or rather the power of $\hat{A}$ in SGC models and the maximum degree of $P(L)$ for LGC models, on mean accuracy across all classical and quantum architectures. The bottom shows the influence of the number of ansatz layers $\ell$ on performance. Shaded regions represent the confidence interval (95\%) across many independent trials, highlighting the trade-off between structural information and circuit expressivity.}
\label{fig:acc_hyperparameters}
\end{figure}

\begin{figure}[H]
    \centering
    \includegraphics[width=1\linewidth]{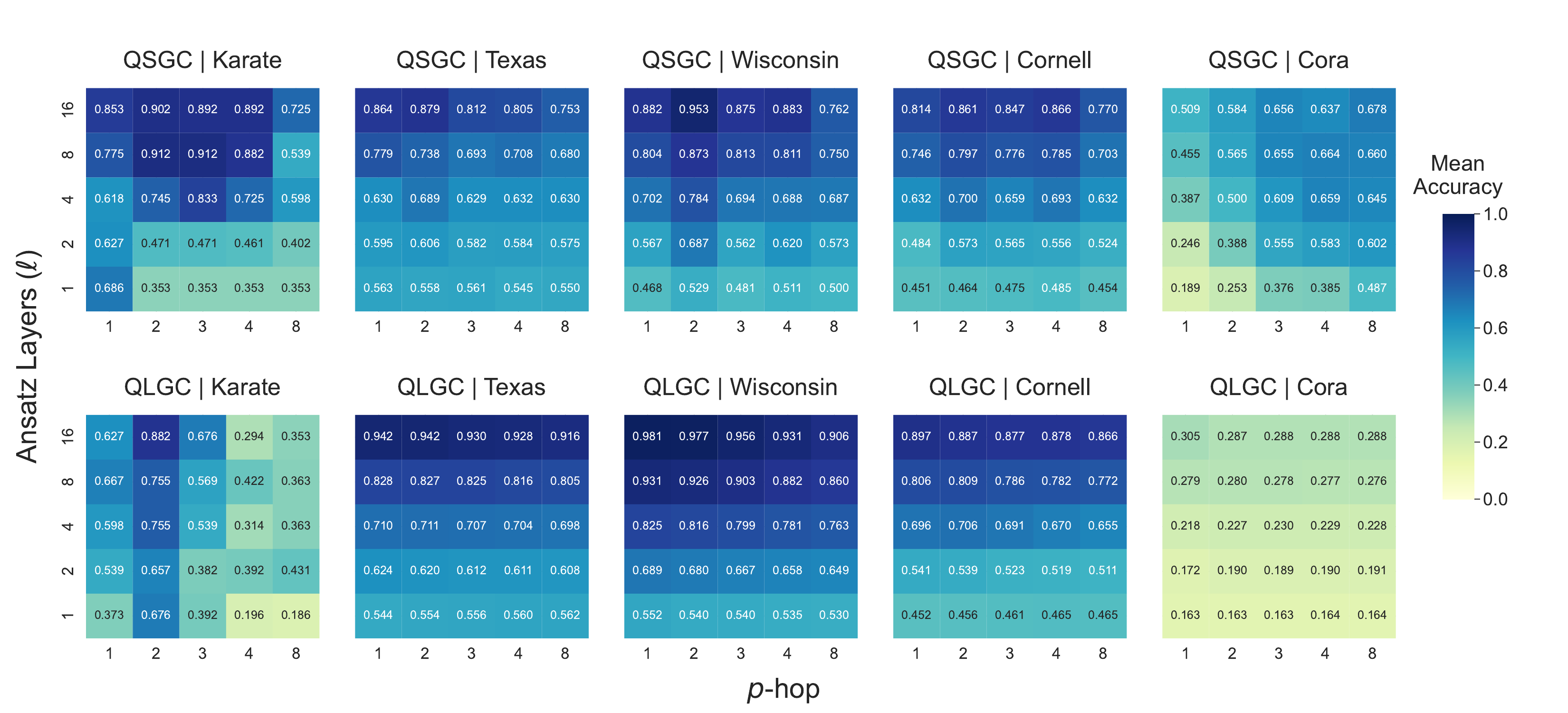}
    \caption{Grid of performance heatmaps across the joint hyperparameter space of $p$-hop order and ansatz layers ($\ell$). Each heatmap illustrates the mean accuracy for a specific quantum architecture (QSGC top and QLGC bottom) on each dataset. The color gradient represents the mean accuracy, where darker regions indicate optimal hyperparameter configurations. This visualization demonstrates that while increasing circuit depth ($\sim \ell$) consistently improves accuracy across most datasets, the optimal $p$-hop distance is dataset-dependent and prone to over-smoothing. QLGC (bottom) generally shows a broader stable region for higher accuracy than QSGC (top).}
    \label{fig:heatmap}
\end{figure}

To conclude with the numerical experiments, the problem of vanishing gradients, commonly known as \textit{barren plateaus} \cite{larocca_barren_2025} is studied. For this, various subgraphs of \textit{Karate} and \textit{Wisconsin} datasets are randomly sampled, in order to encode them in different qubit-sized circuits. For each number of qubits and layers in the quantum circuit, $1000$ random initializations of parameters are made, for which the gradient of the loss function is computed. Fig.~\ref{fig:BP} shows the scaling of the variance of those gradients depending on the number of qubits in the quantum circuit. The phenomenon of Barren Plateaus is characterized by an exponential decrease in the variance, but the results show that the proposed models reveal a more gentle descent. For the \textit{Wisconsin} dataset, it even seems that the variance converges into a constant value. These results suggest that the proposed quantum models might not suffer from Barren Plateaus, but analytical results are required to confirm this claim. In the following section we add a thorough trainability analysis for all QGCNs considered in this work.\\

\begin{figure}[H]
    \centering
    \includegraphics[width=0.83\linewidth]{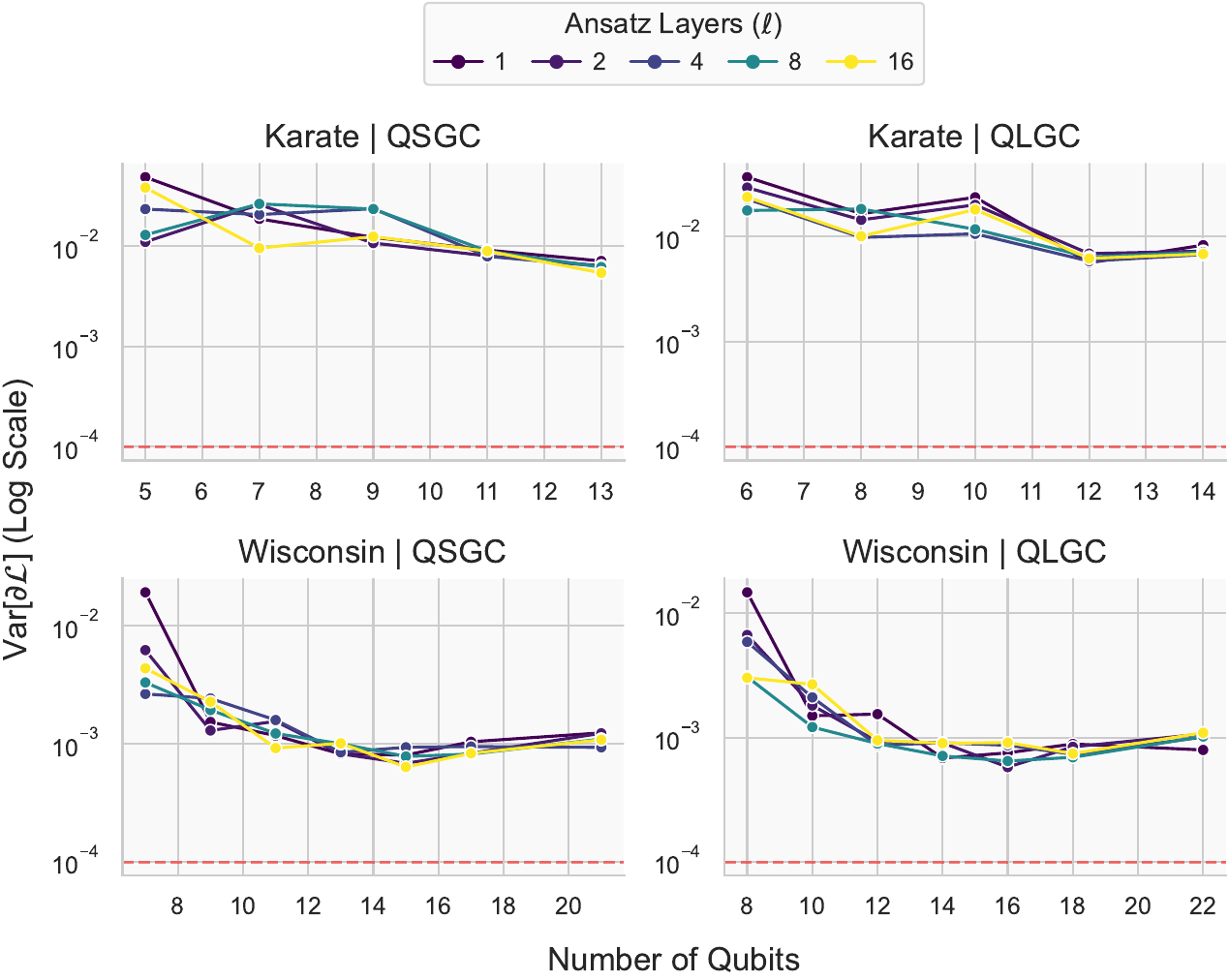}
    \label{fig:BP_karate}
     \caption{Variance of the cost function gradient $\text{Var}[\partial\mathcal{L}]$ as a function of the number of qubits for the Karate and Wisconsin datasets. The plots illustrate the scaling behavior of the gradient for different ansatz layers ($\ell$) for the QSGC (left) and QLGC (right) models. The dashed red line represents the precision limit for 10,000 shots, where the signal-to-noise ratio becomes insufficient for gradient-based optimization. The flattening of the variance for larger qubit counts suggest that circuit expressivity does not incur the scaling limits typically associated with random-like quantum landscapes, maintaining trainability.}
    \label{fig:BP}
\end{figure}

\section{Trainability and Gradient Variance Bounds}\label{sec: Trainability and Gradient Variance Bounds}

The concept of trainability is broad, but it often relates to characterizing the parameter-loss landscape with the goal of finding efficient, or identifying inefficient, ways to navigate it. In particular, the main obstacle hindering the trainability of variational quantum algorithms is the presence of vanishing gradients, a phenomenon labeled as the barren plateau problem. This is a landscape whose geometry becomes exponentially flatter with the system size. Namely,
\begin{equation}
    \Var(\partial \mathcal{L}/\partial \theta_\mu) \propto O\left(\frac{1}{b^n}\right),
\end{equation}
where $\mathcal{L}$ is the cost and its gradient is taken over any direction $\theta_\mu$, $b>1$ is some real number, and $n$ is the system size (in our case, the number of qubits).

In this section, the main results of a cost gradient analysis are provided (and expanded on in Appendix \ref{app:Detailed trainability calculations}). This is done by first isolating the one-layer QSGC model and, then by showing that extending the analysis to weighted graphs, all other architectures reduce to this model (done in Appendix \ref{app: Generalizing the trainability results across different QGCNs}).

\subsection{Trainability of a single-layer QSGC}

The investigated cost function is defined in Eq.~\ref{eq:inner_product_loss}. Encodings of the input data and target classes are given in Eq.~\ref{eq: input encoding} and Eq.~\ref{eq: target encoding}. Explicitly,
\begin{equation}\label{eq: explicit cost}
    \mathcal{L}_{QSGC} = -\frac{1}{N}\sum_{i,j=1}^{N}\hat{A}_{ij}\cdot \Re\left(\bra{\bold{y}_j}U(\thetaB)\ket{\bold{x}_i}\right). 
\end{equation}
Here, $\hat{A}$ is the normalized adjacency matrix with self-loops, $\ket{\bold{x}_i}$ and $\ket{\bold{y}_j}$ are embedded input and target states respectively\footnote{Thanks to the amplitude encoding method, the target states are just columns from the matrix $Y_{sf}$ in Eq.~\ref{eq: target encoding}.}. 
For the ansatz $U(\thetaB)$, a generic sequence of generalized rotations defined by 
\begin{equation}\label{eq: ansatz}
    U(\mathbf{\thetaB}) := \prod_{\eta=D}^1 \exp{(-i\theta_\eta V_\eta)} \cdot W_\eta
\end{equation}
is chosen, where $V_\eta^2=\mathds{1}$ and $D$ is the number of weights $\theta_\eta$.

\paragraph{Remarks:} The trainability calculation below is done for a fully-supervised classification task (hence assuming each node has a target state $\ket{\bold{y}_j}$). A simple generalization for semi-supervised tasks is outlined below.\\

\noindent
As derived in the Appendix, the explicit form of the cost function variance is
\begin{equation}\label{eq: variance}
    \Var(\partial_\mu \mathcal{L}_{QSGC}) = 
    \frac{1}{2^{D+1} N^2}\sum_{\Vec{\alpha}\in\{0,1\}^D} 
    \left[(-1)^{ w(\Vec{\alpha})}\Re\left(Z_{\Vec{\alpha}}^2\right) + \left|Z_{\Vec{\alpha}}^2\right|\right],
\end{equation}
where $\Vec{\alpha}$ sums over all $D$-bit strings and the argument 
\begin{equation}\label{eq: Z_a}
    Z_{\Vec{\alpha}} = \sum_{i,j=1}^{N}\hat{A}_{ij} \bra{\bold{y}_j} \underbrace{\prod_{\eta=D}^1 V_\eta^{\alpha_\eta} W_\eta}_{D_{\Vec{\alpha}}} \ket{\bold{x}_i}.
\end{equation}

\subsubsection{Assumptions}

Whilst Eq.~\ref{eq: variance} is non-trivial, one only needs to determine how the variance scales with the system size. That is, the number of nodes $N$, the edge density (average node degree) $k$, and the number of features per node, $C$. For this matter, three assumptions are made. 

Firstly, the graph is approximated to be $k$-regular. Due to the self-loops, it follows that every edge brings an equal contribution through $\hat{A}= 1/(k+1)$. This reduces Eq.~\ref{eq: Z_a} to $Z_{\Vec{\alpha}}=(k+1)^{-1}\sum_\text{edges}\bra{\bold{y}_j}D_{\Vec{\alpha}}\ket{\bold{x{_i}}}$.

The second assumption is that the whole Hilbert space of the data register is used to embed input and target states, $\ket{\bold{x}_i}$ and $\ket{\bold{y}_j}$. Namely, 
the number of features $C\equiv N_k$, where $N_k=2^{n_k}$ is the Hilbert space dimension of register $Reg(k)$.\footnote{In Appendix \ref{What is $C < N_k$?}, the assumption to any $C\in [N_k/2+1, N_k]$ is loosened by padding the feature vectors with zeros and applying a random permutation to the components. This action turns out to preserve the trainability results described below.}

The final approximation is an averaging over the possible input and target states (labelled ``i,t''). Namely, the distribution of variances over these states is described as a spread centered in the expected variance:
\begin{equation}
    \Var(\partial_\mu \mathcal{L}_{QSGC}) \approx \mathbb{E}_{\text{i,t}}[\Var(\partial_\mu \mathcal{L}_{QSGC})] \pm \Delta _{\text{i,t}}[\Var(\partial_\mu \mathcal{L}_{QSGC})].
\end{equation}
This approximation is reasonable for discrete, classical data as opposed to Haar-random quantum data which concentrates the loss around zero \cite{Thanasilp_2023, Cerezo_2021, Holmes_Sharma_Cerezo_Coles_2022}.

Furthermore, the analysis follows the experimental setting from the sections above: the target states (class labels) are one-hot encoded and the input states are multihot-encoded as
\begin{equation}\label{eq: target input}
    \begin{split}
        \mathbb{E}_\text{i,t}[\ket{\bold{y}_i}] =& \ket{\bold{e}_a}, \\
        \mathbb{E}_\text{i,t}[\ket{\bold{x}_i}] =& \frac{1}{\sqrt{ w(\vec{\gamma})}}\sum_{b=1}^{N_k}\gamma_b\ket{\bold{e}_b}.
    \end{split}
\end{equation}
These are inserted in the variance calculation, with the condition that repeated nodes map to identical basis states $\ket{\bold{e}_a}$ used to average the inputs and targets. Whilst the algebra becomes more tedious (Appendix \ref{app: Expected values over the input and target states}), looking at the expected gradient variance sheds light on the behavior of relevant quantities: $N$, $k$, and $C$.

\subsubsection{Bounds and expected behavior}

With the assumptions above in mind, an upper bound of the expected variance is found:
\begin{equation}\label{eq: upper bound}
    \mathbb{E}_{\text{i,t}}[\Var(\partial_\mu \mathcal{L}_{QSGC})]_\text{upper} = \Theta(C^{-1}).
\end{equation}
Furthermore, this bound is attained with the Instantaneous Quantum Polynomial (IQP) Ansatz, which is defined as
\begin{equation}
    U(\theta)=H^{\otimes n_k}\exp\left(
    \sum_{j=1}^{n_k} i\theta_j Z_j + \sum_{j<k} i\theta_{jk} Z_jZ_k
    \right)H^{\otimes n_k}.
\end{equation}
For this PQC, the commuting $Z$-Pauli gates, whilst limiting the expressivity, help reach the exact upper bound calculated in Eq.~\ref{eq: app upper bound} multiplied by a factor of $1/2$.

For the lower bound, it can be observed (and proven in Lemma \ref{lemma: app alternating sum}) that for each ansatz resulting in 
\begin{equation}
    \mathbb{E}_\text{i,t}\left[\sum_{\Vec{\alpha}\in\{0,1\}^D} 
    (-1)^{ w(\Vec{\alpha})}\Re\left(Z_{\Vec{\alpha}}^2\right)\right]<0,
\end{equation}
changing the gate $W_\lambda\to V_\lambda W_\lambda$ in the ansatz flips the sign of the alternating sum above, which increases the expected variance. Because the addition of $V_\lambda$ is a unitary process, this map does not diminish expressivity. It is straightforward to check the sign of this sum (numerically, using e.g.~a tensor network approximation), which means one can easily triage the ansätze and correct those with a negative sign. Such ansätze are labeled ``inefficient''. Therefore, it is safe to consider only the correct (``efficient'') ansätze when computing our lower bound. This is
\begin{equation}\label{eq: lower bound}
    \mathbb{E}_{\text{i,t}}[\Var(\partial_\mu \mathcal{L}_{QSGC})]_\text{lower} = \Theta(N^{-1}C^{-1}).
\end{equation}

Now that the bounds are established, it is worth asking the question: what is the distribution of the expected variances with respect to the choice of ansatz $U(\thetaB)$? As seen in the calculations in Appendices \ref{app: Upper bound} and \ref{app: Lower bound}, a strong constraint imposed on the ansatz is that $|\bra{\bold{1}}D_{\Vec{\alpha}}\ket{\bold{1}}|$ is either 1 or 0, respectively, for each bound. In reality, by taking an arbitrary ansatz, this quantity will suffer from the \textit{curse of dimensionality} and will scale inversely with the size of the states' Hilbert space. Therefore, by choosing $|\bra{\bold{1}}D_{\Vec{\alpha}}\ket{\bold{1}}|\sim O(C^{-1})$, the expected behavior (derived in Appendix \ref{app: Note on the curse of dimensionality}) turns out to be 
\begin{equation}\label{eq: expected trend}
    \mathbb{E}_{\text{i,t}}[\Var(\partial_\mu \mathcal{L}_{QSGC})]_\text{expected} = \Theta(C^{-2}+N^{-1}C^{-1}).
\end{equation}

\subsubsection{Interpreting the results}

What do Eqs.~\ref{eq: upper bound}, \ref{eq: lower bound}, and \ref{eq: expected trend} say about the single-layer QSGC trainability? On the one hand, the lower bound suggests the presence of the barren plateau as the variance decays exponentially with the number of qubits in registers $Reg(i)$ and $Reg(k)$, which is $\lceil\log(NC)\rceil$.\footnote{This excludes the number of ancillas which can themselves depend linearly as $O(NC)$. However, employing this number of ancillas defeats the purpose of a variational quantum algorithm altogether.} On the other hand, both the upper bound and expected behavior manage to avoid this by having a term that is independent of $N$. 

This means that by restricting the type of tasks considered for the quantum model, the barren plateau phenomenon can be avoided altogether. Specifically, if one focuses on the graph scaling and makes the number of features scale slowly with $N$, as e.g.~$C=O(\polylog(N))$, the trainability may be preserved. This choice fits the purpose of the GCNs considered --- if $C$ were comparable to $N$, a learning model would not be able to generalize correctly anyway.

\paragraph{Semi-supervision.} The calculation was performed for fully supervised tasks to match the experimental set-up. For semi-supervised approaches, the only difference mathematically is the collapse of some of the targets $\ket{\bold{y}_j}$ to zero for unlabeled states. By accounting for this in the expected value calculation, all expected variances gain an extra factor $M/N$, where $M=|Y_L|$ is the number of labeled nodes, out of all $N$ nodes. Therefore, if $M\ll N$, this may impact the results. Intuitively, the less nodes are labeled, the more the training process has to account for this lack of information. The model is still efficiently trained if the fraction $M/N=\Theta(1)$ --- for example, a constant $10\%$ of the nodes are labeled.

\subsection{Generalizing the results to all QGCNs}

In order to extend the single-layer QSGC conclusions, the results above are shown to hold for graphs with weighted edges. One can redefine the adjacency matrix of such a graph to be $\hat{A}_{ij}=w_{ij}$, where $\langle i,j\rangle$ form an edge with weight $w_{ij}$. Furthermore, the $k$-regularity assumption is loosened to a weighted $k$-regularity: edges connected to a node must sum up their weights to $k$. In this case, Appendix \ref{app: Graphs with weighted edges} shows that the results on the bounds and expected behavior for arbitrary ansätze are identical to the ones above.

In addition, Appendix \ref{app: Weighted graphs} highlights that multiplying (or taking a convex sum between) two commuting normalized adjacency matrices with self-loops $\hat{A}$ and $\hat{B}$ lead to a valid normalized adjacency matrix with self-loops $\hat{C}$ of a new graph. Furthermore, if the initial graphs are weighted $k$-regular, then the new graph is also weighted $k$-regular. An immediate corollary is that any polynomial of $\hat{A}$, denoted $P(\hat{A})$ can be seen as some adjacency matrix $\hat{B}$ of a new weighted $k$-regular graph. Thus, all QGCN models reduce to a single-layer QSGC, making the results above valid for any QGCN.

\paragraph{Note:} Models such as $p$-th order QSGCs (QLGCs) are still needed because running $\hat{A}^p$ ($P(\hat{A})$) on quantum circuits is more scalable than computing it classically, complexity-wise. However, on current devices, the overheads may invert the statement.

\section{Note on Complexity and Dequantization}\label{sec: Note on Complexity and Dequantization}

The QGCN complexity results in \cite{liao_graph_2024} suggest a trade-off between the algorithm qubit number and circuit depth. These complexities are respectively proportional to the space and time complexities of the forward pass. This section adds that the backpropagation (using the \textit{parameter-shift rule} \cite{Mitarai_Negoro_Kitagawa_Fujii_2018, Banchi_Branford_Waghela_2025}) contributes with a query complexity of $O(D)$ ($O(D+p)$) for the QSGC (QLGC) models. Furthermore, the \textit{shot noise} brings an additional query complexity, inversely proportional to the cost gradient variance established in the previous section \cite{Sweke_Wilde_Meyer_Schuld_Faehrmann_Meynard-Piganeau_Eisert_2020, Gu_Lowe_Dub_Coles_Arrasmith_2021}. By considering the expected behavior of an arbitrary ansatz, the shot-noise contributes a factor of $O(1/(C^{-2}+N^{-1}C^{-1}))\approx O(C^2)$ for $C\ll N$.

Finally, the most significant contribution to the final space allocation and runtime is given by the \textit{input problem}: as one uploads classical data onto the quantum computer using the amplitude and block encoding schemes, the circuit needs to be compiled before it is run. This compilation time will require at least storing and reading each of the classical inputs. As such, the compilation complexities alone end up matching the classical GCN complexities, hence canceling any desirable advantage. The only way to bypass this problem is to assume that the circuits are compiled efficiently using an oracle, or QRAM. For this reason, the classical and quantum GCN algorithms are first compared in isolation, assuming an efficient data loading process for both. Then, a dequantized QGCN model is introduced and regimes where it outperforms the quantum model are found . The full complexity study is done in Appendix \ref{app: Complexity and dequantization expanded}.

\subsection{Complexities and regimes of advantage}

The full space and time complexities of the classical and quantum models are given below:
\begin{align}
    S_{CGCN} =& O(Nk+NC+C^2), \label{eq: space CGCN}\\
    T_{CGCN} =& O(NkC+NC^2), \label{eq: time CGCN}\\
    S_{QGCN} =& O(n_{anc}+n_{anc}'), \label{eq: space QGCN}\\
    T_{QGCN} =& \tilde{O}\left(\underbrace{N\frac{\log(n_{anc})}{n_{anc}}}_{\text{amplitude enc.}} + \underbrace{N\log(N) \cdot s\log(s)\frac{\log(n'_{anc})}{n'_{anc}}}_{\text{block enc.}}+\underbrace{n'_{anc}}_{\text{QSVT}}\right)\cdot \underbrace{O(D)}_{\text{backprop.}}\cdot \underbrace{O(C^2)}_{\text{shots}}. \label{eq: time QGCN} 
\end{align}
Here, $n_{anc}$ and $n'_{anc}$ are respectively the numbers of ancillary qubits required for the amplitude and block encoding gates. The number $s$ is the sparsity of the adjacency matrix $\hat{A}$, and $k$ is the average degree.\footnote{All fixed precision factors are omitted, and the fixed QGCN order $p$ is absorbed into the Big-$O$ notation. The Big-$\tilde{O}$ notation also ignores doubly logarithmic functions.} The number of ansatz weights $D$ is considered to be varying much more slowly than $C$, and can thus be ignored as $D\ll C\ll N$. Additionally, the QSGC models omit the QSVT runtime term, which is specific to the linear graph convolution. The numbers of ancillas are bounded as follows:
\begin{align}
    \Omega(\log(NC))\leq & n_{anc} \leq O(NC), \label{eq: nanc bounds}\\
    \Omega(\log(N))\leq & n'_{anc}\leq O(N\log(N) \cdot s\log(s)). \label{eq: nancprime bounds}
\end{align}

\paragraph{Regimes of advantage.} There is freedom in choosing the number of ancillary qubits, $n_{anc}$ and $n'_{anc}$. More qubits can perform the amplitude and block encoding with a significantly smaller depth, and vice-versa. 

For the QSGC models, this gives a trade-off between the overall space and time complexities, and effectively splits the advantage in three regimes: \textbf{(1)} By choosing $O(\polylog(N))$ ancillary qubits, the QSGC gains an \textit{exponential space reduction} and a polylogarithmic speed-up w.r.t.~its classical counterpart. \textbf{(2)} Choosing a sublinear number of ancillas $n_{anc}=O(N^\alpha)$ and $n'_{anc}=O(N^{\alpha'})$ with $\alpha,\alpha'\in(0,1)$ leads to \textit{polynomial advantage in both space and time complexities}. \textbf{(3)} Lastly, by choosing $O(N/\polylog(N))$ ancillas, we flip the advantage in the first case --- now having an \textit{exponential speed-up} with a polylogarithmic space reduction.

For the QLGC models, the additional runtime of the QSVT when implementing the polynomial $P(\hat{A})$ limits the trade-off that the QSGC benefits from. Namely, the runtime in this case is lower-bounded by 
\begin{equation}
    \left(T_{QLGC}\right)_{\text{min}} = \tilde{O}\left(
    NC^3\frac{\log n_{anc}}{n_{anc}} + C^2\sqrt{Ns \log^2 N \log s}
    \right).
\end{equation}
As this is at least sublinear in $N$, the exponential speed-up case \textbf{(3)} given above is lost. However, regimes \textbf{(1)} and \textbf{(2)} are maintained.

\subsection{Low-rank simulability}

To even out the QRAM assumption, Ewin Tang's breakthrough algorithm \cite{Tang_2019, Chia_Gilyen_Li_Lin_Tang_Wang_2020} constructs an efficient classical oracle (a tree data-structure) for an $M$-by-$N$ matrix, so that it is not needed to query all elements of the matrix. Instead, a sampling distribution is created in $O(\polylog(MN)\poly(r,1/\varepsilon))$ time that can reproduce a rank-$r$ approximation of the matrix with precision $\varepsilon$. This method is a baseline for dequantization of algorithms that claim quantum advantage. As derived in Appendix \ref{app: Low-rank simulability}, this algorithm is applied to find simulability conditions for the QGCN. An estimator for the cost function
\begin{equation}
    \mathcal{L}_{QGCN}^\text{Deq}= \Tr\left(Y^T\cdot P(\hat{A})\cdot X \cdot U(\thetaB)\right)
\end{equation}
is found using the sampling method, where $P(\hat{A})$ is either the polynomial of a QLGC or $\hat{A}^p$ of a QSGC. As the matrices $X$, $Y$, and $\hat{A}$ are large, efficient data structures are assumed for each of them. Then, by sampling sets of rows and columns from the adjacency matrix, one can estimate the diagonalization of $\hat{A}$ and use it to approximate the polynomial $P(\hat{A})$. To this end, the final trace is computed by also sampling the input and target matrices and computing an estimator of the cost $\mathcal{L}$. A full description of the dequantized algorithm is done in Appendix \ref{app: Low-rank simulability}. The resulting space and time complexities (including exact exponents of the sampling method) are:
\begin{align}
    S_{QGCN}^\text{Deq} &= O(r^{11}+C^2), \label{eq: s deq} \\
    T_{QGCN}^\text{Deq} &= O(\log^2(N)\cdot r^{16.5}+Cr^{6.5}). \label{eq: t deq}
\end{align}
Comparing them to the quantum model's complexities, the dequantized model gains a speed-up over the quantum one for $r=O(N^{0.11})$, and achieves an additional space suppression for $r=O(N^{0.09})$. 

As $r$ is a truncated rank, the QGCN models considered may be dequantized if the normalized adjacency matrix with self-loops of the graph considered, $\hat{A}$, has a stable rank
\begin{equation}
    r_\text{stable} \equiv \frac{\sum_i\sigma_i^2}{\sigma_1^2} \approx O(N^{0.11}).
\end{equation}
Here, $\sigma_i$ are the singular values of $\hat{A}$, in descending order. Because the sparsity $s$ of this adjacency matrix is lower-bounded by the fraction $N/r$ (derived in Lemma \ref{lemma: rank vs sparsity}), then a high rank implies a sparse graph. Hence, the quantum model may still have an advantage for either \textit{sparse enough graphs}, or for \textit{dense, unstructured graphs} (such that $r_\text{stable}$ is large enough).

\section{Conclusion}\label{sec:conclusion_qgnn}

In this article, a family of \gls{qgnn} architectures based on graph convolutions has been introduced and studied from a theoretical and a numerical perspective. Starting from the classical formulation of graph convolutional models, quantum counterparts were proposed by replacing the trainable linear maps with parameterized quantum circuits acting on amplitude-encoded node features, i.e.\ the PQC replaces the classical weight matrix $W$ rather than implementing an arbitrary rectangular linear map. This construction preserves the general message-passing intuition of graph learning while exploiting the compact representation and expressive power of quantum states. In particular, two concrete realizations were analyzed: the QSGC, as the quantum analogue of Simplified Graph Convolutions, and the QLGC, inspired by linear graph convolutional models.\\

The main theoretical motivation for these models lies in their potential computational advantage. As discussed throughout the article, the quantum implementation can encode high-dimensional node features using a logarithmic number of qubits and can generate effective weight matrices with a number of trainable parameters that is significantly smaller than in the classical case. Under suitable assumptions on state preparation and data access, this leads to favorable time- and space-complexity scalings for large graphs. We stress that this potential advantage concerns primarily the number of trainable parameters and the asymptotic scaling of the algorithm; it does not by itself constitute a demonstrated end-to-end quantum advantage, which would additionally require efficient state preparation, block encoding, and readout. Although such assumptions are nontrivial and currently hardware-dependent, they provide a clear indication that graph learning is a natural candidate for quantum enhancement, especially in regimes where memory consumption becomes the main bottleneck for classical methods.\\

The numerical experiments support the practical relevance of these proposals. Both quantum architectures were shown to be trainable and capable of achieving competitive results on several benchmark datasets for semi-supervised node classification. Even when the classical models remain the strongest baseline overall, the quantum models often approach their performance closely and, in some cases, surpass a classical counterpart despite using a more compact parametrization. The hyperparameter study further indicates that circuit depth improves the expressivity of the models, while increasing the graph propagation order yields the expected trade-off between richer structural information and the risk of oversmoothing. In addition, the empirical analysis of gradient concentration suggests that the proposed architectures do not exhibit a severe barren plateau behavior in the explored regimes, which is an encouraging sign for their trainability.\\

On the theoretical side, these numerical results are validated by the cost gradient analysis presented above: the model can be made trainable independently of the graph size or structure. This is valid only if the feature representation is reasonably compact, as the variance is expected to scale inversely with the number of features per node $C$. Thanks to the entanglement structure between the address and data registers, the graph scaling can thus be decoupled from the PQC processing the features. 

Additionally, the QGCN models are shown to achieve better complexities than their classical counterparts in certain regimes. Depending on the allocation of ancillary qubits, this can be an exponential space reduction with a small time speed-up, or vice versa, or a balanced polynomial space-time advantage. However, as the model assumes an efficient oracle to avoid the input problem, the situation is not as fortunate in reality. To this end, the low-rank simulablity analysis provided indicates that the model can be dequantized for structured graphs. \\

Nevertheless, the work shows that quantum graph learning is a promising research direction. The proposed models establish a concrete bridge between the theory of graph neural networks and the framework of variational quantum circuits, offering architectures that are not only conceptually well motivated but also empirically viable. Future developments may include deeper analyses of expressivity and generalization, the extension to attentional or message-passing quantum architectures, the incorporation of noise-aware designs for near-term devices, and the study of real quantum implementations. Trainability-wise, stricter bounds can be found, or alternative encoding methods can be compared to further optimize the model. Altogether, the present results position \gls{qgnn}s as an appealing avenue for exploring how quantum computing may contribute to scalable machine learning on structured data.\\

\section*{Acknowledgements}
Many thanks to Léo Monbroussou for helpful discussions, as well as observations pointed out in earlier drafts of of the trainability and classical simulability sections. This work was supported by the Engineering and Physical Sciences Research Council [grant number EP/Y035046/1].


\newpage
\printbibliography

\newpage
\appendix
\renewcommand{\appendixpagename}{\textbf{APPENDIX}}
\begin{center}  
\appendixpagename
\end{center}

\input{appendix}

\end{document}

%% file: glossary.tex
\newacronym{ml}{ML}{Machine Learning}
\newacronym{qml}{QML}{Quantum Machine Learning}
\newacronym{pac}{PAC}{Probably Approximately Correct}
\newacronym{nisq}{NISQ}{Noisy Intermediate-Scale Quantum}
\newacronym{hhl}{HHL}{Harrow-Hassidim-Lloyd}
\newacronym{qpe}{QPE}{Quantum Phase Estimation}
\newacronym{qft}{QFT}{Quantum Fourier Transform}
\newacronym{lcu}{LCU}{Linear Combination of Unitaries}
\newacronym{qsvt}{QSVT}{Quantum Singular Value Transform}
\newacronym{qubo}{QUBO}{Quadratic Unconstrained Binary Optimization}
\newacronym{pqc}{PQC}{Parameterized Quantum Circuit}
\newacronym{vqa}{VQA}{Variational Quantum Algorithm}
\newacronym{qaoa}{QAOA}{Quantum Approximate Optimization Algorithm}
\newacronym{vqe}{VQE}{Variational Quantum Eigensolver}
\newacronym{qnn}{QNN}{Quantum Neural Network}
\newacronym{vqc}{VQC}{Variational Quantum Classifier}
\newacronym{qram}{QRAM}{Quantum Random Access Memory}
\newacronym{dla}{DLA}{Dynamical Lie Algebra}
\newacronym{zne}{ZNE}{Zerno Noise Extrapolation}
\newacronym{gpu}{GPU}{Graphics processing unit}
\newacronym{qpu}{QPU}{Quantum Processing Unit}
\newacronym{bps}{BP}{Barren Plateaus}
\newacronym{gqml}{GQML}{Geometric Quantum Machine Learning}
\newacronym{eqnn}{EQNN}{Equivariant Quantum Neural Networks}
\newacronym{etf}{ETF}{Equiangular Tight Frame}
\newacronym{cnn}{CNN}{Convolutional Neural Network}
\newacronym{hea}{HEA}{Hardware-Efficient Ansatz}
\newacronym{gnn}{GNN}{Graph Neural Networks}
\newacronym{qgnn}{QGNN}{Quantum Graph Neural Network}
\newacronym{sgc}{SGC}{Simplified Graph Convolution}
\newacronym{lgc}{LGC}{Linear Graph Convolution}
\newacronym{gcn}{GCN}{Graph Convolutional Networks}
\newacronym{ntca}{NTCA}{ Nonlinear Transformation of Complex Amplitudes}

%% file: appendix.tex

\section{Graph preliminaries}\label{app: Graph preliminaries}

\subsection{Unweighted graphs}\label{app: Unweighted graphs}

A graph is a collection of nodes $\mathcal{N}$ and edges $\mathcal{E}$ between the nodes. It usually describes structured information where different parts of the data are linked via some relations. In our work, we define addresses of the nodes and encode them into the register $Reg(i)$, thus called the address register. In order to encode the edges, we employ the \textit{normalized adjacency matrix with self-loops}. To unwrap this notion, we define the adjacency matrix $A$ as a symmetric matrix such that
\begin{equation}
    A_{ij} = \begin{cases}
        1 \quad \text{if nodes }i,j\text{ are connected}, \\
        0 \quad \text{otherwise}.
    \end{cases}
\end{equation}
By definition, $A_{ii}=0$. In a GNN, it is important to reinforce the information of a node onto itself, so we introduce self-loops in the graph (literally meaning that each note has an edge linking to itself). The self loops count as one extra edge for each node. 
We define this as $\tilde{A}\equiv A+I_N$, where $N$ is the total number of nodes.
Finally, for the purpose of encoding this matrix into a unitary gate, the matrix itself needs to have a sub-unitary spectral norm, meaning $\|A\|_2\leq1$. For this reason,  $\tilde{A}$ must also be normalized. In order not to bias more connected nodes, we use the following normalization: 
\begin{equation}\label{eq: app normalizing adj matr}
    \hat{A}=\tilde{D}^{-\frac{1}{2}} \tilde{A} \tilde{D}^{-\frac{1}{2}},
\end{equation}
where $\tilde{D}_{i i}=\sum_{j} \tilde{A}_{i j}$ is the degree matrix. This normalizes each edge by the degrees of the connected nodes as 
\begin{equation}\label{eq: app norm adj matrix}
    \hat{A}_{ij} = \frac{A_{ij}}{\sqrt{\deg(i)\deg(j)}}.
\end{equation}
Here, the degree $\deg(i)=\tilde{D}_{i i}$ of node $i$ is the number of edges connected to it.

\subsection{Weighted graphs}\label{app: Weighted graphs}

For a graph with \textit{weighted edges}, each edge has an associated weight $w_{ij}>0$. In this case the adjacency matrix is defined as 
\begin{equation}\label{eq: app adj matr weighted edges}
    A_{ij} = \begin{cases}
        w_{ij} \quad \text{if nodes }i,j\text{ are connected}, \\
        0 \quad \text{otherwise}.
    \end{cases}
\end{equation}

By adding the self loops, we now have extra freedom in choosing the weights of the self loops to be $w_{ii}>0$. Therefore, the new adjacency matrix is also 
\begin{equation}
    \tilde{A}_{ij} = \begin{cases}
        A_{ij} \quad \text{if }i\neq j, \\
        w_{ii} \quad \text{otherwise}.
        \end{cases}
\end{equation}

The definitions for adding self-loops and the degrees are the same as above, therefore inheriting the changes from Eq.~\ref{eq: app adj matr weighted edges}. We note that by node degree and degree matrix in this case we refer to weighted degrees, meaning that, in accordance with the definitions above, we have
\begin{equation}
    \tilde{D}_{ii} = \sum_j \tilde{A}_{ij} = \sum_jw_{ij},
\end{equation}
with $\deg(i) := \tilde{D}_{ii}$. This means that node $i$ is not restricted to having $\deg(i)$ edges connected to it, but instead it sums up the weights of all its connected edges.

Next, we will prove two important results that will help us generalize our first order QSGC results to both general QSGCs and QLGCs. The generalization is done in Appendix \ref{app: Generalizing the trainability results across different QGCNs}. For further details we refer the reader to \cite{Godsil2001-ms, gallier2016spectraltheoryunsignedsigned}.

\subsubsection{Multiplying normalized adjacency matrices}

Here we show that multiplying two commuting normalized adjacency matrices with self-loops produces a new, valid normalized adjacency matrix with self-loops. This is done in Proposition \ref{prop: app adj matr prod is adj matr}, and to prove it, we will employ the following two lemmas:

\begin{lemma}\label{lemma: app eigenvector}
    For a graph whose normalized adjacency matrix with self-loops is $\hat{A}$ and the degree matrix is $\tilde{D}$, the vector $\tilde{D}^{1/2}\bold{1}$ is an eigenvector of $\hat{A}$, with eigenvalue 1. Here, $\bold{1}$ is a vector whose elements are ones.
\end{lemma}

\begin{proof}
    For the usual adjacency matrix with self-loops $\tilde{A}$, we notice that
    \begin{equation}
        (\tilde{A}\bold{1})_i = \sum_j \tilde{A}_{ij}\cdot 1 = \tilde{D}_{ii} \qimpl \tilde{A}\bold{1} = \tilde{D}\bold{1}.
    \end{equation}
    Therefore, calculating
    \begin{equation}
        \hat{A}\cdot\tilde{D}^{1/2}\bold{1} = \tilde{D}^{-1/2}\tilde{A}\tilde{D}^{-1/2}\tilde{D}^{1/2}\bold{1} = \tilde{D}^{-1/2} \tilde{A}\bold{1}=\tilde{D}^{1/2}\bold{1},
    \end{equation}
    thus $\tilde{D}^{1/2}\bold{1}$ is an eigenvector of $\hat{A}$ corresponding to eigenvalue 1.
\end{proof}

\begin{lemma}\label{lemma: app prop degrees}
    If the normalized adjacency matrices with self-loops of two graphs commute, then their degree matrices are proportional.
\end{lemma}

\begin{proof}
    For graphs $\mathcal{G}_A$ and $\mathcal{G}_B$ we define the normalized adjacency matrices with self-loops to be $\hat{A}$ and $\hat{B}$, respectively. We also denote the degree matrices as $\tilde{D}_A$ and $\tilde{D}_B$.
    
    If $\hat{A}$ and $\hat{B}$ commute, then they share a mutual basis of eigenvectors. Furthermore, because they are symmetric, the basis is orthogonal. 
    
    Using Lemma \ref{lemma: app eigenvector}, we know that $\tilde{D}_A^{1/2}\bold{1}$ and $\tilde{D}_B^{1/2}\bold{1}$ are in this basis. If they do not coincide, they should be orthogonal, which is impossible because their inner product is
    \begin{equation}
        \bold{1}^T(\tilde{D}_B^{1/2})^T\tilde{D}_A^{1/2}\bold{1} = \sum_i\sqrt{(\tilde{D}_A)_{ii}(\tilde{D}_B)_{ii}} > 0.
    \end{equation}
    Therefore, they must be proportional, and so $(\tilde{D}_A)_{ii}\propto(\tilde{D}_B)_{ii}$, which completes the proof.
\end{proof}

\begin{proposition}\label{prop: app adj matr prod is adj matr}
    Let $\mathcal{G}_A$ and $\mathcal{G}_B$ be two graphs with weighted edges and the same number of nodes, such that their normalized adjacency matrices with self-loops are respectively $\hat{A}$ and $\hat{B}$. Then there exists a graph $\mathcal{G}_C$ whose normalized adjacency matrix with self-loops matches $\hat{C}=\hat{A}\hat{B}$ only if the two matrices $\hat{A}$ and $\hat{B}$ commute.
\end{proposition}

\begin{proof}
    We solve this from an algebraic point of view. The three conditions that $\hat{C}$ needs to meet are (1) symmetry, (2) non-negativity, and (3) the degree relationship given in Eq.~\ref{eq: app norm adj matrix}, based on constructing a valid degree matrix $\tilde{D}_C$ associated to it.

    The first condition is solved by the symmetry of each of the matrices $\hat{A}$ and $\hat{B}$, as well as the commutation relation:
    \begin{equation}
        \hat{C}^T=\left(\hat{A}\hat{B}\right)^T = \hat{B}^T \hat{A}^T = \hat{B} \hat{A} = \hat{A}\hat{B} = \hat{C}.
    \end{equation}

    The second condition means that all of the elements $\hat{C}_{ij}=\sum_k\hat{A}_{ik}\hat{B}_{kj}$ must be positive, which is true because all elements in $\hat{A}$ and $\hat{B}$ are also positive.

    Finally, to address the third condition, we start constructing the desired graph $\mathcal{G}_C$. As we know from Lemma \ref{lemma: app prop degrees}, the degree matrices of graphs $\mathcal{G}_A$ and $\mathcal{G}_B$ are proportional; let them be related as $\tilde{D}_B = b\tilde{D}_A$, where $\tilde{D}_i$ is the degree matrix of graph $\mathcal{G}_i$ and $a$ is a proportionality constant. 
    
    Because $\hat{C}=\hat{A}\hat{B}$ commutes with both, if we want $\mathcal{G}_C$ to have $\hat{C}$ as the normalized adjacency matrix with self-loops, we need to further define the degree matrix $\tilde{D}_C\equiv c\tilde{D}_A$.

    Finally, we check the relation in Eq.~\ref{eq: app norm adj matrix} for $\hat{C}$ and $\tilde{D}_C$:
    \begin{align}
        \tilde{D}_C^{\frac{1}{2}}\hat{C}\tilde{D}_C^{\frac{1}{2}} =& \tilde{D}_C^{\frac{1}{2}} \hat{A}\hat{B}\tilde{D}_C^{\frac{1}{2}} \nonumber \\
        =& \tilde{D}_C^{\frac{1}{2}} \tilde{D}_A^{-\frac{1}{2}} \tilde{A} \tilde{D}_A^{-\frac{1}{2}} \tilde{D}_B^{-\frac{1}{2}} \hat{B} \tilde{D}_B^{-\frac{1}{2}} \tilde{D}_C^{\frac{1}{2}} \nonumber \\
        =& \frac{c}{b}\tilde{D}_A^{\frac{1}{2}} \tilde{D}_A^{-\frac{1}{2}} \tilde{A} \tilde{D}_A^{-\frac{1}{2}} \tilde{D}_A^{-\frac{1}{2}} \hat{B} \tilde{D}_A^{-\frac{1}{2}} \tilde{D}_A^{\frac{1}{2}} \nonumber \\
        =& \frac{c}{b}\tilde{A}\tilde{D}_A \tilde{B} \equiv \tilde{C}.
    \end{align}
    After simplifications, we establish that the weights of the desired graph should be 
    \begin{equation}
        \tilde{C}_{ij} = \frac{c}{b}\sum_k\frac{\tilde{A}_{ik}\tilde{B}_{kj}}{(\tilde{D}_B)_{kk}}.
    \end{equation}
    Assuming $b$ is given, $c$ is still a choice in our definition. Therefore, not one, but a whole family of graphs $\{\mathcal{G}_C^c\}_{c>0}$ is defined such that $\hat{A}\hat{B}$ is the (mutual) normalized adjacency matrix with self-loops. 

    To this end, we need to check that the relation $(\tilde{D}_C)_{ii}=\sum_j\tilde{C}_{ij}$ holds:
    \begin{equation}
        \sum_j \tilde{C}_{ij} = \frac{c}{b}\sum_k\frac{\tilde{A}_{ik}\sum_j\tilde{B}_{kj}}{(\tilde{D}_B)_{kk}} = c\sum_k\frac{\tilde{A}_{ik}}{(\tilde{D}_A)_{kk}} = c(\tilde{D}_A)_{ii} \underbrace{\sum_k\frac{\tilde{A}_{ik}}{(\tilde{D}_A)_{ii}(\tilde{D}_A)_{kk}}}_{\sum_k \hat{A}_{ik}=1} = (\tilde{D}_C)_{ii}.
    \end{equation}
    Note that this check is equivalent to checking that Lemma \ref{lemma: app eigenvector} holds for $\mathcal{G}_C$.
\end{proof}

For $k$-regular graphs we can also give the following corollary:

\begin{corollary}\label{cor: app commuting k-regular graphs}
    Following Proposition \ref{prop: app adj matr prod is adj matr}, if graphs $\mathcal{G}_A$ and $\mathcal{G}_B$ are weighted $k$-regular, then $\mathcal{G}_C$ can be chosen to be $k$-regular.
\end{corollary}

\begin{proof}
    For a regular graph, all nodes share the same degree value $k+1$ (including the self-loops). As such, the degree matrices are $\tilde{D}_A=\tilde{D}_B=k\mathds{1}$ (so $b=1$). As seen in the previous proof, $\tilde{D}_C=c\tilde{D}_A$, which makes the new degree matrix be
    \begin{equation}
        \tilde{D}_C=c(k+1)\mathds{1}.
    \end{equation}
    Therefore, to make the new graph $k$-regular we set $c=1$.
\end{proof}

\subsection{Convex sums of normalized adjacency matrices}

In this section we further extend the properties of commuting normalized adjacency matrices with self-loops by also showing that taking the convex sum of two such matrices leads to a new, valid normalized adjacency matrix with self-loops.

\begin{proposition}\label{prop: app convex sum 2 terms}
    For two graphs $\mathcal{G}_A$ and $\mathcal{G}_B$ with commuting normalized adjacency matrices with self-loops $\hat{A}$ and $\hat{B}$, there exists a graph $\mathcal{G}_C$ such that its normalized adjacency matrix with self-loops is a convex sum $\hat{C}=\alpha \hat{A} + \beta \hat{B}$ for some non-negative numbers $\alpha$ and $\beta$ such that $\alpha+\beta=1$.
\end{proposition}

\begin{proof}
    As done in the previous Proposition, we confirm that a linear combination of two symmetric, non-negative matrices gives a symmetric, non-negative matrix as long as the coefficient $\alpha$ and $\beta$ are non-neegative.

    By employing Lemma \ref{lemma: app prop degrees}, because $\hat{A}$ and $\hat{B}$, the graphs' degree matrices are proportional, $\tilde{D}_B=b\tilde{D}_A$. If we want to construct a graph $\mathcal{G}_C$ whose normalized adjacency matrix with self-loops is $\hat{C}$, then this would also commute with $\hat{A}$ and $\hat{B}$, and so a necessary condition is that the degree matrix of $\mathcal{G}_C$ is proportional to that of $\mathcal{G}_A$, as $\tilde{D}_C=c\tilde{D}_A$ for some positive number $c$ which we set.

    Then, to complete the construction of $\mathcal{G}_C$, we calculate the weights given by the unnormalized adjacency matrix $\tilde{C}$ as follows:
    \begin{align}
        \tilde{C} =& \tilde{D}_C^{\frac{1}{2}} \hat{C} \tilde{D}_C^{\frac{1}{2}} \nonumber \\
        =& c\tilde{D}_A^{\frac{1}{2}}\left(\alpha \hat{A} + \beta \hat{B}\right)\tilde{D}_A^{\frac{1}{2}} \nonumber \\
        =& c\tilde{D}_A^{\frac{1}{2}}\left(\alpha 
        \tilde{D}_A^{-\frac{1}{2}}\tilde{A}\tilde{D}_A^{-\frac{1}{2}}
        + \beta 
        \tilde{D}_B^{-\frac{1}{2}}\tilde{B}\tilde{D}_B^{-\frac{1}{2}}
        \right)\tilde{D}_A^{\frac{1}{2}} \nonumber \\
        =& c\left(\alpha \tilde{A} + \frac{\beta}{b}\tilde{B}\right).
    \end{align}
    As done above, a last check that $(\tilde{D}_C)_{ii}=\sum_j\tilde{C}_{ij}$:
    \begin{equation}
        \sum_j \tilde{C}_{ij} = c\left(\alpha \sum_j\tilde{A}_{ij} + \frac{\beta}{b}\sum_j\tilde{B}_{ij}\right) = c\left(\alpha + \beta\right)(\tilde{D}_A)_{ii} = (\alpha+\beta)(\tilde{D}_C)_{ii} = (\tilde{D}_C)_{ii},
    \end{equation}
    which is only valid because $\alpha+\beta=1$. Here, everything except $c$ is given and $c$ can be set to be any positive number.
\end{proof}

This can be extended to a convex sum over $n$ normalized adjacency matrices with self-loops:

\begin{corollary}\label{cor: app app convex sum n terms}
    For graphs $\mathcal{G}_{A_p}$ with normalized adjacency matrices with self-loops $\hat{A}_p$ for $p$ ranging from 1 to $n$, the convex sum $\hat{C}=\sum_p\alpha_p\hat{A}_p$ is a valid adjacency matrix with self loops of some graph $\mathcal{G}_C$, for some $\alpha_p$ such that $\sum_p\alpha_p=1$. 
\end{corollary}

\begin{proof}
    By induction, the two-term case is done in the Proposition above. Let us assume it holds for $n$ such graphs:
    \begin{equation}
        \hat{A} \equiv \sum_{p=1}^n\alpha_p\hat{A}_p, \quad\text{where}\quad \sum_{p=1}^n\alpha_p=1.
    \end{equation}
    In order to add a new graph, $\mathcal{G}_{A_{n+1}}$ with $\hat{A}_{n+1}$, we apply Proposition \ref{prop: app convex sum 2 terms} again: For some $\alpha$ and $\beta$, the following is a valid normalized adjacency matrix with self-loops:
    \begin{equation}
        \hat{C}=\alpha\sum_{p=1}^n\alpha_p\hat{A}_p  + \beta\hat{A}_{n+1}, \quad\text{where}\quad \alpha\sum_{p=1}^n\alpha_p+\beta=1.
    \end{equation}
    By relabeling $\alpha\alpha_p \equiv \alpha_p'$ and $\beta\equiv\alpha_{n+1}'$, the condition is met for the $n+1$ graphs as well.
\end{proof}

Finally, a corollary on $k$-regularity:

\begin{corollary}\label{cor: app convex sum k regular}
    Following Proposition \ref{prop: app adj matr prod is adj matr}, if graphs $\mathcal{G}_A$ and $\mathcal{G}_B$ are weighted $k$-regular, then $\mathcal{G}_C$ can be chosen to be $k$-regular.
\end{corollary}

\begin{proof}
    Similar to the proof of Corollary \ref{cor: app commuting k-regular graphs}. Because $\tilde{D}_C=c\tilde{D}_A=c(k+1)\mathds{1}$, then the free parameter $c$ can be set to 1.
\end{proof}

\section{Detailed trainability calculations}\label{app:Detailed trainability calculations}

Here we describe in more detail the motivation behind the cost function as well as derivations for the results in Section \ref{sec: Trainability and Gradient Variance Bounds}. Furthermore, we work through a simplified version of the large-graph limit and give an example through the use of the IQP ansatz.

\paragraph{Why QSGC?} 

Here, we focus on the first order QSGC because it lies at the core of all other models. The generalization for any $p$-th order QSGC as well as QLGC models is done in Appendix \ref{app: Generalizing the trainability results across different QGCNs}. 

\paragraph{Redefinition of input, ouput, and target states:} For the following gradient analysis, we do not need to include the Hadamard test qubit into the calculations; for that reason, we discard it in our states. Specifically, we rewrite:
\begin{equation}
    \bigg(\Tr_H(\ket{\psi_X}) \ , \ \Tr_H(\ket{\psi_{out}}) \ , \  \Tr_H(\ket{\psi_Y})\bigg) \ \ \equiv \ \ \bigg(\ket{\psi'_X} \ , \ \ket{\psi'_{out}}  , \ \ket{\psi'_Y}\bigg),
\end{equation}
where the $\Tr_H$ traces out the qubit responsible for the (modified) Hadamard test. \\

To this end, we consider the \textit{fully supervised} case for the trainability calculations for simplicity and comment about how the results change for the semi-supervised case.

\subsection{Motivation of the cost function}

A natural measure of the overlap between two states (the output and the target) is their square distance, which in this case is
\begin{equation}
    \mathcal{L}_\text{distance} := \| \ket{\psi'_{out}} - \ket{\psi'_Y} \|^2.
\end{equation}

By expanding this definition, we reach

\begin{align*}
    \mathcal{L}_\text{distance} =& \braket{\psi'_{out}|\psi'_{out}} - \braket{\psi'_{out}|\psi'_{Y}} - \braket{\psi'_{Y}|\psi'_{out}} + \braket{\psi'_{Y}|\psi'_{Y}} \\
    =& 2\left(1-\Re\braket{\psi'_{out}|\psi'_{Y}}\right) \\
    =& 2\left(1+\mathcal{L}_{QGCN}\right).
\end{align*}

Therefore, the constant term as well as the prefactor of 2 can be dropped as they do not contribute to the dependence of the gradient variance on the input size. For future reference, we write the cost function $\mathcal{L}_{QGCN}=(\mathcal{L}_1+\mathcal{L}_1^*)/2$, with $\mathcal{L}_1=-\braket{\psi'_{out}|\psi'_Y}$.

\subsection{Notation and explicit form of $\mathcal{L}_1$}\label{app: Notation and explicit form of L_1}

We give the notations for the index and data register qubit numbers $n_i\equiv\log_2(N_i)$, $n_k\equiv\log_2(N_k)$, so $N_i$ and $N_k$ are the dimensions of the respective Hilbert spaces. This means that the number of nodes $N$ and the dimension of the feature vectors $C$ are bounded by $N_i/2 < N \leq N_i$ and $N_k/2 < C \leq N_k$. Therefore, their orders match, $N_i=O(N)$ and $N_k=O(C)$.

Let us now calculate $\mathcal{L}_1$ explicitly.

The normalized input and target states are
\begin{align}
    \ket{\psi'_X} &= \frac{1}{\sqrt{N}}\sum_{i=1}^{N}\ket{i}\ket{\bold{x}_i}, \\
    \ket{\psi'_Y} &= \frac{1}{\sqrt{N}}\sum_{j=1}^{N}\ket{j}\ket{\bold{y}_j}.
\end{align}

Here, $i,j$ label the nodes of the graph. Next is the output state:
\begin{equation}
    \ket{\psi'_{out}} = \left(\hat{A}\otimes U(\thetaB)\right)\ket{\psi'_X},
\end{equation}
where $\hat{A}$ is the normalized adjacency matrix with self loops as defined in Appendix \ref{app: Graph preliminaries}. This is not a unitary gate; it is implemented through block encoding in a larger circuit\footnote{We note that this can be done as $\hat{A}$ is a subunitary, in the sense that its operator norm is $\|\hat{A}\|_\text{op}\leq 1$.}, but can be used as such in our calculation. 

This leads to the inner product:
\begin{align}
    \mathcal{L}_1 &= -\frac{1}{N} \sum_{i,j=1}^{N}\bra{j}\hat{A}\ket{i}\bra{\bold{y}_j}U(\thetaB)\ket{\bold{x}_i} \\
    &= -\frac{1}{N}\sum_{i,j=1}^{N}\hat{A}_{ij}\bra{\bold{y}_j}U(\thetaB)\ket{\bold{x}_i}. \label{eq: app explicit cost}
\end{align}

Here, the ansatz $U(\thetaB)$ is generic; however, with minimal loss of generality, we can assume that the ansatz takes the following form\footnote{We note that we could add a fixed gate $W_{D+1}$ at the end of the circuit, as $U(\mathbf{\thetaB}) := W_{D+1}\cdot\prod_{\eta=D}^1 \exp{(-i\theta_\eta V_\eta)} \cdot W_\eta$, as is usual in many ansätze. The results are derived identically if one keep this gate into account. \label{foot: extra gate}}:
\begin{equation}\label{eq: app ansatz}
    U(\mathbf{\thetaB}) := \prod_{\eta=D}^1 \exp{(-i\theta_\eta V_\eta)} \cdot W_\eta.
\end{equation}
The parameters $\thetaB=(\theta_1,\theta_2,\dots,\theta_D)$ are trainable, $D$ is thus the dimension of $\thetaB$. This is not guaranteed to be equal to the depth of the ansatz, as usually some of the gates containing these parameters can be performed simultaneously. The number $D$ is usually related to the size of the data register $n_k$ only. Finally, the product ranges from $D$ to 1 (descending) as the order of the parametrized gates is from right to left in order to be applied in ascending order to the state $\ket{\bold{x}_i}$.

The unitary operators $V_\eta$ and $W_\eta$ are not trainable, but ansatz-dependent. $V_k$ is also Hermitian and squares to the identity, which leads to a useful formula:
\begin{equation}\label{eq:helper exponential expansion}
    \exp{(-i\theta_\eta V_\eta)} = \cos\theta_\eta \mathds{1} - i\sin\theta_\eta V_\eta.
\end{equation}
This can be seen by taking the Taylor expansion of the LHS and separating by even and odd terms.

\subsection{Derivation of the gradient variance}\label{app: Derivation of the gradient variance}

To derive the gradient variance for the Quantum Simple Graph Convolution, $\mathcal{L}_{QSGC}$, we start by decomposing it in terms of $\mathcal{L}_1$, as
\begin{equation}\label{eq: app var sum of terms}
    \Var\left(\partial_\mu \mathcal{L}_{QSGC}\right) = \frac{1}{4}\left(\Var(\partial_\mu \mathcal{L}_1) + \Var(\partial_\mu \mathcal{L}_1^*) + 2\Cov(\partial_\mu \mathcal{L}_1, \partial_\mu \mathcal{L}_1^*)\right).
\end{equation}
Here, $\partial_\mu\equiv \partial/\partial\theta_\mu$ is the $\mu$-th direction of the gradient. Computing the builing block of Eq.~\ref{eq: app var sum of terms}:
\begin{align}
    \partial_\mu \mathcal{L}_1 =& -\frac{\partial}{\partial\theta_\mu}\frac{1}{N}\sum_{i,j=1}^{N}\hat{A}_{ij}\bra{\bold{y}_j}U(\thetaB)\ket{\bold{x}_i} \\
    =& -\frac{1}{N}\sum_{i,j=1}^{N}\hat{A}_{ij}\bra{\bold{y}_j}\partial_\mu U(\thetaB)\ket{\bold{x}_i}.
\end{align}

For the ansatz, we have
\begin{align}
    \partial_\mu U(\thetaB) =& \partial_\mu\prod_{\eta=D}^1 \exp{(-i\theta_\eta V_\eta)} \cdot W_\eta \\
    =& \underbrace{\left(\prod_{\eta=D}^{\mu+1} \exp{(-i\theta_\eta V_\eta)} \cdot W_\eta\right)}_{U_L}
    \partial_\mu \exp{(-i\theta_\mu V_\mu)} \cdot W_\mu
    \underbrace{\left(\prod_{\eta=\mu-1}^1 \exp{(-i\theta_\eta V_\eta)} \cdot W_\eta\right)}_{U_R} \\
    =& U_L \left( -iV_\mu \exp{(-i\theta_\mu V_\mu)} \cdot W_\mu\right) U_R. \label{eq:derivative of U}
\end{align}

\paragraph{Expected value is zero.} By taking the expected value of $\partial_\mu \mathcal{L}_1$:
\begin{align}
    \mathbb{E}_{\thetaB}[\partial_\mu \mathcal{L}_1] =& \frac{1}{(2\pi)^D}
    \int \dd\thetaB
    (\partial_\mu \mathcal{L}_1) \nonumber \\
    =& \frac{1}{(2\pi)^D}
    \int \dd\thetaB' \int \dd\theta_\mu
    \frac{1}{N}\sum_{i,j=1}^{N}
    \hat{A}_{ij} 
    \bra{\bold{y}_j}\partial_\mu U(\thetaB)\ket{\bold{x}_i} \nonumber \\
    =& \frac{1}{(2\pi)^DN}\sum_{i,j=1}^{N} \hat{A}_{ij}\int \dd \thetaB'
    \bra{\bold{y}_j} U(\thetaB',\theta_\mu=\pi)-U(\thetaB',\theta_\mu=-\pi)\ket{\bold{x}_i}.
\end{align}
Here, all weights $\theta_
\eta$ range from $-\pi$ to $\pi$. The full integral over $\thetaB$ was separated into an integral over $\theta_\mu$ and the integral over all other parameters, denoted $\thetaB'$. We notice that thanks to Eq.~\ref{eq:helper exponential expansion}, $\exp{(-i(\pm\pi) V_\mu)}=\mathds{1}$ and thus the two terms in the braket above cancel out. By a similar calculation, this extends to the whole cost gradient, so $\mathbb{E}_{\thetaB}[\mathcal{L}_{QSGC}]=0$. \\

Now we can calculate the gradient variance considering the range of all parameters to be $[-\pi,\pi]\ni\theta_\eta$. When integrating over the parameters, this range is assumed.
\begin{align}
    \Var(\partial_\mu \mathcal{L}_1) =& \frac{1}{(2\pi)^D}\int \dd\thetaB \left(\partial_\mu \mathcal{L}_1\right)^2 \\
    =& \frac{1}{(2\pi)^D} \int \dd\thetaB' \int\dd\theta_\mu \frac{1}{N^2}\sum_{i,j=1}^{N}\hat{A}_{ij}\bra{\bold{y}_j}\partial_\mu U(\thetaB)\ket{\bold{x}_i}
    \sum_{k,l=1}^{N}\hat{A}_{kl}\bra{\bold{y}_l}\partial_\mu U(\thetaB)\ket{\bold{x}_k} \\
    =& \frac{1}{(2\pi)^DN^2}\sum_{i,j,k,l=1}^{N} \hat{A}_{ij}\hat{A}_{kl}\int \dd \thetaB' \int \dd \theta_\mu
    \bra{\bold{y}_j}\partial_\mu U(\thetaB)\ket{\bold{x}_i}
    \bra{\bold{y}_l}\partial_\mu U(\thetaB)\ket{\bold{x}_k}.
\end{align}

The same separation into $\theta_\mu$ and $\thetaB'$ was applied as for the expected value calculation.

Let us start with the term $(i,j,k,l)$ and integrate it over $\theta_\mu$, using Eqs.~\ref{eq:helper exponential expansion} and \ref{eq:derivative of U}. We denote this by $I_\mu$:
\begin{align}
    I_\mu =& \int \dd \theta_\mu
    \bra{\bold{y}_j}\partial_\mu U(\thetaB)\ket{\bold{x}_i}
    \bra{\bold{y}_l}\partial_\mu U(\thetaB)\ket{\bold{x}_k} \\
    =& \int \dd \theta_\mu
    \bra{\bold{y}_j}U_L \left( -iV_\mu \exp{(-i\theta_\mu V_\mu)} \cdot W_\mu\right) U_R\ket{\bold{x}_i}
    \bra{\bold{y}_l}U_L \left( -iV_\mu \exp{(-i\theta_\mu V_\mu)} \cdot W_\mu\right) U_R\ket{\bold{x}_k} \\
    =& -\int \dd \theta_\mu 
    \bra{\bold{y}_j}U_L V_\mu \left(\cos\theta_\mu \mathds{1} - i\sin\theta_\mu V_\mu \right) W_\mu U_R\ket{\bold{x}_i} \nonumber\\ &\phantom{..\int \dd \theta_\mu}\cdot
    \bra{\bold{y}_l}U_L V_\mu \left(\cos\theta_\mu \mathds{1} - i\sin\theta_\mu V_\mu \right) W_\mu U_R\ket{\bold{x}_k} \label{eq:line decompose exp}\\
    =& -\int\dd\theta_\mu (\cos\theta_\mu)^2
    \bra{\bold{y}_j}U_L V_\mu W_\mu U_R\ket{\bold{x}_i} 
    \bra{\bold{y}_l}U_L V_\mu W_\mu U_R\ket{\bold{x}_k} \nonumber\\
    &+ \int\dd\theta_\mu (\sin\theta_\mu)^2 
    \bra{\bold{y}_j}U_L V_\mu^2 W_\mu U_R\ket{\bold{x}_i}
    \bra{\bold{y}_l}U_L V_\mu^2 W_\mu U_R\ket{\bold{x}_k} \nonumber\\
    &-i\int\dd\theta_\mu\cos\theta_\mu\sin\theta_\mu 
    \bra{\bold{y}_j}U_L V_\mu W_\mu U_R\ket{\bold{x}_i}
    \bra{\bold{y}_l}U_L V_\mu^2 W_\mu U_R\ket{\bold{x}_k} \nonumber\\
    &-i\int\dd\theta_\mu\cos\theta_\mu\sin\theta_\mu
    \bra{\bold{y}_j}U_L V_\mu^2 W_\mu U_R\ket{\bold{x}_i}
    \bra{\bold{y}_l}U_L V_\mu W_\mu U_R\ket{\bold{x}_k}.
\end{align}

The second equality comes from Eq.~\ref{eq:helper exponential expansion}, the third from Eq.~\ref{eq:derivative of U}, which we then separate by decomposing the product of brakets. One can now check that the first two trigonometric integrals are equal to $\pi$ and the last two are 0. Considering also that $V_\mu^2=\mathds{1}$, we reach the final relation
\begin{equation}
    I_\mu = \pi\left( \bra{\bold{y}_j}U_L W_\mu U_R\ket{\bold{x}_i} 
    \bra{\bold{y}_l}U_L W_\mu U_R\ket{\bold{x}_k} - 
    \bra{\bold{y}_j}U_L V_\mu W_\mu U_R\ket{\bold{x}_i}
    \bra{\bold{y}_l}U_L V_\mu W_\mu U_R\ket{\bold{x}_k}\right).
\end{equation}

Fortunately, after this tedious calculation, the other integrals follow. The only difference is that in this one extra $V_\mu$ factors were present, which lead to a sign cancellation only as the resulting terms were swapped. Therefore, each integral (over e.g.~$\theta_\lambda$) takes the current terms and splits each into two: one product of brakets that do not contain the $V_\lambda$ minus one one product of brakets that do. By counting all of these terms with a $D$-bit counter $\Vec{\alpha}$, we obtain the desired result for the variance of $\partial_\mu \mathcal{L}_1$:
\begin{equation}\label{eq:app variance of L1 gradient}
    \Var(\partial_\mu \mathcal{L}_1) = \frac{1}{2^D N^2}\sum_{i,j,k,l=1}^{N}\hat{A}_{ij}\hat{A}_{kl}\sum_{\Vec{\alpha}\in\{0,1\}^D} (-1)^{ w(\Vec{\alpha})} 
    \bra{\bold{y}_j} \prod_{\eta=D}^1 V_\eta^{\alpha_\eta} W_\eta \ket{\bold{x}_i}
    \bra{\bold{y}_l} \prod_{\eta=D}^1 V_\eta^{\alpha_\eta} W_\eta \ket{\bold{x}_k}.
\end{equation}

The string $\Vec{\alpha}$ tracks the branch we took to follow this term as we computed the integrals, and the number of minuses factored comes from the numbers of 1s present in the string, thus the Hamming weight $ w(\Vec{\alpha})$. Reducing it to a simpler form:
\begin{equation}\label{eq:app var L1}
    \Var(\partial_\mu \mathcal{L}_1) = \frac{1}{2^D N^2}\sum_{\Vec{\alpha}\in\{0,1\}^D} (-1)^{ w(\Vec{\alpha})} 
    \left( \sum_{i,j=1}^{N}\hat{A}_{ij}
    \bra{\bold{y}_j} \prod_{\eta=D}^1 V_\eta^{\alpha_\eta} W_\eta \ket{\bold{x}_i}
    \right)^2.
\end{equation}

We can do a similar calculation for $\Var(\partial_\mu \mathcal{L}_1^*)$, but first we note that this is not trivially the complex conjugate of Eq.~\ref{eq:app variance of L1 gradient}.
The derivative of $\mathcal{L}_1^*$ is 
\begin{equation}
    \partial_\mu \mathcal{L}_1^* = \partial_\mu \left[\frac{1}{N}\sum_{i,j=1}^{N}\bra{\bold{x}_i}\prod_{\eta=1}^D W_\eta^\dagger \exp(+i\theta_\eta V_\eta) \ket{\bold{y}_j}\right],
\end{equation}
thus instead of extracting a $-i V_\mu$ factor from the exponential as done previously, now a plus sign appears. Therefore, 
\begin{equation}
    \partial_\mu \mathcal{L}_1^* = - (\partial_\mu \mathcal{L}_1)^*.
\end{equation}

When computing the variance, the minus sign above is squared and vanishes, hence 
\begin{equation}\label{eq:app var L1 conj}
    \Var(\partial_\mu \mathcal{L}_1^*) = 
    \frac{1}{2^D N^2}\sum_{\Vec{\alpha}\in\{0,1\}^D} (-1)^{ w(\Vec{\alpha})} 
    \left( \sum_{i,j=1}^{N}\hat{A}_{ij}
    \bra{\bold{x}_i} \prod_{\eta=1}^D W_\eta^\dagger V_\eta^{\alpha_\eta} \ket{\bold{y}_j}
    \right)^2.
\end{equation}

Finally, the calculation for the covariance is very similar, except the decomposition in line \ref{eq:line decompose exp}. Here, one of the brakets is complex conjugated, as
\begin{equation}
    \bra{\bold{y}_j}U_L V_\mu \left(\cos\theta_\mu \mathds{1} - i\sin\theta_\mu V_\mu \right) W_\mu U_R\ket{\bold{x}_i}
    \bra{\bold{x}_k}U_L W_\mu^\dagger V_\mu \left(\cos\theta_\mu \mathds{1} + i\sin\theta_\mu V_\mu \right)  U_R\ket{\bold{y}_l},
\end{equation}
and the plus means that when separating the terms, both the $\cos^2$ and the $\sin^2$ have a positive sign, hence we drop the $(-1)^{ w(\Vec{\alpha})}$ in the answer for the covariance:
\begin{equation}\label{eq:app cov L1 L1 conj}
    \Cov(\partial_\mu \mathcal{L}_1, \partial_\mu \mathcal{L}_1^*) = \frac{1}{2^DN^2} \sum_{\Vec{\alpha}\in\{0,1\}^D} \left|\sum_{i,j=1}^{N} \hat{A}_{ij}\bra{\bold{y}_j} \prod_{\eta=D}^1 V_\eta^{\alpha_\eta} W_\eta \ket{\bold{x}_i}
    \right|^2.
\end{equation}

Adding up these results as given in Eq.~\ref{eq: app var sum of terms} leads to the full cost gradient variance given in Eq.~\ref{eq: variance}:
\begin{equation}\label{eq: app variance}
    \Var(\partial_\mu \mathcal{L}_{QSGC}) = 
    \frac{1}{2^{D+1} N^2}\sum_{\Vec{\alpha}\in\{0,1\}^D} 
    \left[(-1)^{ w(\Vec{\alpha})}\Re\left(Z_{\Vec{\alpha}}^2\right) + \left|Z_{\Vec{\alpha}}^2\right|\right],
\end{equation}
where $\Vec{\alpha}$ sums over all $D$-bit strings and the argument 
\begin{equation}
    Z_{\Vec{\alpha}} = \sum_{i,j=1}^{N}\hat{A}_{ij} \bra{\bold{y}_j} \underbrace{\prod_{\eta=D}^1 V_\eta^{\alpha_\eta} W_\eta}_{\equiv D_{\Vec{\alpha}}} \ket{\bold{x}_i}.
\end{equation}

\subsection{Results for regular graphs}\label{app: Results for regular graphs}

For a $k$-regular graph, each node has a degree of $k+1$ (including self-connections). By applying Eq.~\ref{eq: app norm adj matrix}, $\hat{A}_{ij} = 1/\sqrt{\deg(i)\deg(j)}=1/(k+1)$ if nodes $i$ and $j$ are connected by an edge, and 0 otherwise. Here we noted the degree of node $i$ by $\deg(i)$. It is easier to work with the separate terms derived above, namely Eqs.~\ref{eq:app var L1} and \ref{eq:app cov L1 L1 conj}. The former equation becomes
\begin{equation}
    \Var(\partial_\mu \mathcal{L}_1) = 
    \frac{1}{2^DN^2 (k+1)^2} 
    \sum_{\Vec{\alpha}\in\{0,1\}^D} 
    \left(
    \sum_{<i,j>} 
    \bra{\bold{y}_j} D_{\Vec{\alpha}}\ket{\bold{x}_i}
    \right)^2,
\end{equation}
whilst the latter is simply
\begin{equation}
    \Cov(\partial_\mu \mathcal{L}_1, \partial_\mu \mathcal{L}_1^*) = 
    \frac{1}{2^DN^2 (k+1)^2} 
    \sum_{\Vec{\alpha}\in\{0,1\}^D} 
    \Biggl|
    \sum_{<i,j>} 
    \bra{\bold{y}_j} D_{\Vec{\alpha}}\ket{\bold{x}_i}
    \Biggr|^2.
\end{equation}
By $<i,j>$ we mean that nodes $i$ and $j$ must be connected (this includes the self loops $<i,i>$). For further reference we use the $D_{\Vec{\alpha}}$ for the ansatz terms.

\subsection{Expected values over the input and target states}\label{app: Expected values over the input and target states}

The expected value of the cost gradient variance is done by considering the distributions of input and target feature vectors respectively, and averaging each node feature over these distributions. 

\paragraph{Feature set-up.} Firstly, we take a look at the encoding of the input and target states, respectively $\ket{\bold{x}_i}$ and $\ket{\bold{y}_i}$, for node $i$. Throughout our experiments, the target states are one-hot encoded onto the first $C$ states of some basis (usually, the computational basis); let us call this basis
$\{\ket{\bold{e}_a}\}_{a\in[1,N_k]}$. For simplicity, we consider the calculation for $C=2^{n_k}\equiv N_k$.

The input states are multi-hot encoded, meaning that each state $\ket{\bold{x}_i}$ is equal to some equal superposition of a subset of basis states, based on the encoded input bit string, call it $\vec{\gamma}$.\footnote{Notice in Eq.~\ref{eq:app target input} that because of the equal superposition of basis states, the normalisation factor depends on the Hamming weight of the bit string $\Vec{\gamma}$.} We write this as
\begin{equation}\label{eq:app target input}
    \begin{split}
        \ket{\bold{y}_i} =& \ket{\bold{e}_a}, \\
        \ket{\bold{x}_i} =& \frac{1}{\sqrt{ w(\vec{\gamma})}}\sum_{b=1}^{N_k}\gamma_b\ket{\bold{e}_b}.
    \end{split}
\end{equation}
The total number of possible target states is $N_k$, equal to the that of basis states. The number of possible input states is that of all non-zero $N_k$-bit strings, which is $2^{N_k}-1$.

To this end, the input and target states are considered independent random variables\footnote{In practice, the input and target will have hidden correlations; otherwise, the model will not learn any useful information. However, for our purposes, we take them as independent for the sake of estimating the expected value.}, and are independent across the nodes of the graph. Therefore, the expected value operator $\mathbb{E}[\cdot]$ acts on any function $f$ of the states $\ket{\bold{x}_i}$ and $\ket{\bold{y}_i}$ as 
\begin{equation}\label{eq:app target input average}
    \begin{split}
        \mathbb{E}[f(\ket{\bold{y}_i})] =& 
        \frac{1}{N_k}\sum_{a=1}^{N_k}f(\ket{\bold{e}_a}), \\
        \mathbb{E}[f(\ket{\bold{x}_i})] =& 
        \frac{1}{2^{N_k}-1}\sum_{\substack{\Vec{\gamma}\in\{0,1\}^{N_k} \\ \Vec{\gamma}\neq\Vec{0}}}f\left(\frac{1}{\sqrt{ w(\vec{\gamma})}}\sum_{b=1}^{N_k}\gamma_b
        \ket{\bold{e}_b}\right).
    \end{split}
\end{equation}
Note that we use states, rather than density operators. This is because the distribution actually runs over decisions for the (classical) inputs and targets, that are then implemented onto the QGNN. This is a classical sampling process where each of the runs has a determined input and target, rather than running any of them with a probability and ``forgetting'' which decision was taken. 

\paragraph{Semi-supervision.} In the semi-supervised case, we only label $M$ out of the $N$ nodes. In our calculation, this is equivalent to allowing our target states $\ket{\bold{y}_j}$ to collapse to zero with a (classical) probability $p_s=(N-M)/N$. Although this state is not realizable physically, we can use it as a mathematical trick to generalize our results. In this case, in our expected value over targets we now have a convex sum between the relation above and zero, as
\begin{equation}
    \mathbb{E}'[f(\ket{\bold{y}_i})] = 
        (1-p_s)\cdot\frac{1} {N_k}\sum_{a=1}^{N_k}f(\ket{\bold{e}_a}) + p_s\cdot 0  = \frac{M}{NN_k}\sum_{a=1}^{N_k}f(\ket{\bold{e}_a}).
\end{equation}
As expected, this scales linearly with the fraction of labeled nodes. The fewer these nodes are, the harder it is to train. Therefore, we consider that a significant fraction of the nodes is labeled, such that $M/N \sim O(1)$. \\

Returning to the expected variance calculation, in order to find the expected value of the variance, $\mathbb{E}\left[\Var(\partial_\mu \mathcal{L}_{QSGC})\right]$, we will average the three separate terms using Eq.~\ref{eq: app var sum of terms} first:
\begin{align}
    \mathbb{E}\left[\Var(\partial_\mu \mathcal{L}_1)\right] &= 
    \frac{1}{2^DN^2 (k+1)^2} 
    \sum_{\Vec{\alpha}\in\{0,1\}^D}
    (-1)^{w(\Vec{\alpha})}
    \mathbb{E}\left[
    \left(
    \sum_{<i,j>} 
    \bra{\bold{y}_j}D_{\Vec{\alpha}}\ket{\bold{x}_i}
    \right)^2\right], \label{eq: app avg var L1}\\
    \mathbb{E}\left[\Var(\partial_\mu \mathcal{L}_1^*)\right] &= \mathbb{E}\left[\Var(\partial_\mu \mathcal{L}_1)\right]^*, \\
    \mathbb{E}\left[\Cov(\partial_\mu \mathcal{L}_1,\partial_\mu \mathcal{L}_1^*)\right] &= \frac{1}{2^DN^2 (k+1)^2} 
    \sum_{\Vec{\alpha}\in\{0,1\}^D} 
    \mathbb{E}\left[
    \left|
    \sum_{<i,j>} 
    \bra{\bold{y}_j}D_{\Vec{\alpha}}\ket{\bold{x}_i}
    \right|^2\right]. \label{eq: app avg cov L1 L1 conj}
\end{align}

By expanding the square terms from Eqs.~\ref{eq: app avg var L1} and \ref{eq: app avg cov L1 L1 conj},
\begin{equation}\label{eq:app relevant averages}
    \mathbb{E}\left[
    \sum_{\substack{{<i,j>} \\ <k,l>}} 
    \bra{\bold{y}_j}D_{\Vec{\alpha}}\ket{\bold{x}_i}
    \bra{\bold{y}_l}D_{\Vec{\alpha}}\ket{\bold{x}_k}
    \right]\equiv S_{\Vec{\alpha}}^v, \quad \text{and} \quad
    \mathbb{E}\left[
    \sum_{\substack{{<i,j>} \\ <k,l>}} 
    \bra{\bold{y}_j}D_{\Vec{\alpha}}\ket{\bold{x}_i}
    \bra{\bold{x}_k}D_{\Vec{\alpha}}^\dagger \ket{\bold{y}_l}
    \right]\equiv S_{\Vec{\alpha}}^c,
\end{equation}
we observe that the averages are taken over pairs of nodes $(i,j)$ and $(k,l)$. As such, one needs to take into account whether the (directed) edges are disconnected, or have one mutual node, or both, as each of these cases requires a different distribution over the input and target random variables. For simplicity, we denote the expressions by $S_{\Vec{\alpha}}^v$ and $S_{\Vec{\alpha}}^c$ (which will be used for the variance and covariance terms respectively).

\subsubsection{Dependence of the variance with the graph size $N$}\label{app: Dependence of the variance with the graph size N_i}

When considering the distribution over $\ket{\bold{x}_i}$, $\ket{\bold{y}_j}$, $\ket{\bold{x}_k}$, and $\ket{\bold{y}_l}$, we need to separate the following cases based on whether $i=k$ or $j=l$ or both, or neither. 

\vspace{1em}

\begin{table}[h!]
\begin{center}
\caption{Counting the number of possible combinations of pairs of edges based on a constraint (in the first column). The values are denoted by $R_{||}$, $R_\land$, $R_\lor$, and $R_{|}$, respectively.}
\vspace{0.5em}
\begin{tabular}{||c c c c||} 
    \hline
    Constraint & Independent variables & Number of cases & Denoted as  \\ [0.5ex] 
    \hline\hline
    $i\neq k$, $j\neq l$ & $\ket{\bold{x}_i}$, $\ket{\bold{x}_k}$, $\ket{\bold{y}_j}$, $\ket{\bold{y}_l}$ & $N(N-1)(k+1)^2-Nk(k+1)$ & $R_{||}$ \\ 
    \hline
    $i= k$, $j\neq l$ & $\ket{\bold{x}_i}$, $\ket{\bold{y}_j}$, $\ket{\bold{y}_l}$ & $Nk(k+1)$ & $R_{\land}$ \\
    \hline
    $i\neq k$, $j= l$ & $\ket{\bold{x}_i}$, $\ket{\bold{x}_k}$, $\ket{\bold{y}_j}$ & $Nk(k+1)$ & $R_\lor$ \\
    \hline
    $i= k$, $j= l$ & $\ket{\bold{x}_i}$, $\ket{\bold{y}_j}$ & $N(k+1)$ & $R_|$ \\ [1ex] 
    \hline
\end{tabular}
\label{tab: app R coefficients}
\end{center}
\end{table}

\vspace{1em}
In Table \ref{tab: app R coefficients}, we count all the pairs of edges categorized by the constraint on whether input or target feature vectors coincide or not and denote them using the $R$-notation given in the last column\footnote{The indices correspond to whether the edges have zero common nodes (``$||$''), or the $\ket{\bold{x}}$ or $\ket{\bold{y}}$ in common (``$\land$'' or ``$\lor$''), or both (``$|$'').}. Rewriting the expected values from Eq.~\ref{eq:app relevant averages}:
\begin{equation}\label{eq: app Sav term}
    \begin{split}
        S_{\Vec{\alpha}}^v = & 
        R_{||} \mathbb{E}\left[
        \bra{\bold{y}_j}D_{\Vec{\alpha}}\ket{\bold{x}_i}
        \bra{\bold{y}_l}D_{\Vec{\alpha}}\ket{\bold{x}_k}
        \right] + R_{\land} \mathbb{E}\left[
        \bra{\bold{y}_j}D_{\Vec{\alpha}}\ket{\bold{x}_i}
        \bra{\bold{y}_l}D_{\Vec{\alpha}}\ket{\bold{x}_i}
        \right] \\
        &+ R_{\lor} \mathbb{E}\left[
        \bra{\bold{y}_j}D_{\Vec{\alpha}}\ket{\bold{x}_i}
        \bra{\bold{y}_j}D_{\Vec{\alpha}}\ket{\bold{x}_k}
        \right] +R_{|} \mathbb{E}\left[
        \bra{\bold{y}_j}D_{\Vec{\alpha}}\ket{\bold{x}_i}
        \bra{\bold{y}_j}D_{\Vec{\alpha}}\ket{\bold{x}_i}
        \right],
    \end{split}
\end{equation}
and
\begin{equation}\label{eq: app sac term}
    \begin{split}
        S_{\Vec{\alpha}}^c = & 
        R_{||} \mathbb{E}\left[
        \bra{\bold{y}_j}D_{\Vec{\alpha}}\ket{\bold{x}_i}
        \bra{\bold{x}_k}D_{\Vec{\alpha}}\ket{\bold{y}_l}
        \right] + R_{\land} \mathbb{E}\left[
        \bra{\bold{y}_j}D_{\Vec{\alpha}}\ket{\bold{x}_i}
        \bra{\bold{x}_i}D_{\Vec{\alpha}}\ket{\bold{y}_l}
        \right] \\
        &+ R_{\lor} \mathbb{E}\left[
        \bra{\bold{y}_j}D_{\Vec{\alpha}}\ket{\bold{x}_i}
        \bra{\bold{x}_k}D_{\Vec{\alpha}}\ket{\bold{y}_j}
        \right] +R_{|} \mathbb{E}\left[
        \bra{\bold{y}_j}D_{\Vec{\alpha}}\ket{\bold{x}_i}
        \bra{\bold{x}_i}D_{\Vec{\alpha}}\ket{\bold{y}_j}
        \right].
    \end{split}
\end{equation}

\paragraph{Simple asymptotic analysis.} In order to calculate the asymptotic behavior of the gradient variance with respect to the graph size (quantified by $N$ and $k$), one can notice that the expected values of different node pairs are only dependent on the size of the distributions, thus on $N_k$, and the terms $D_{\Vec{\alpha}}$ derived from the ansatz. Whilst our analysis consideres $N$ and $N_k$ as independent parameters, all standard ansätze used depend strictly on the data register size, hence also on the logarithm of $N_k$. It is thus safe to assume that all of the edge pair expected values are equal to $O_\text{gr}(1)$ in both the node number $N$ and the average connectivity $k$.\footnote{Here we use the ``graph'' subscript $O_\text{gr}$ (and $\Theta_\text{gr}$ etc.) to restrict the result to large $N$ or $k$ asymptotes, for fixed $N_k$.} As such, both of the summations above have the order of the leading term, $S_{\Vec{\alpha}}^v, S_{\Vec{\alpha}}^c = O_\text{gr}(R_{||}(N,k))$. This implies that, when keeping $N_k$ constant,
\begin{align}
    \mathbb{E}[\Var(\partial_\mu \mathcal{L}_{QSGC})] =& \frac{1}{4}(\mathbb{E}[\Var(\partial_\mu L)] + \mathbb{E}[\Var(\partial_\mu L)^*] + 2\mathbb{E}[\Cov(\partial_\mu \mathcal{L}_1, \partial_\mu \mathcal{L}_1^*)]) \nonumber \\
    =& \frac{1}{2^{D+2}N^2 (k+1)^2} \sum_{\Vec{\alpha}\in\{0,1\}^D}
    \left[ (-1)^{
     w(\Vec{\alpha})}((S_{\Vec{\alpha}}^v) + (S_{\Vec{\alpha}}^v)^*) + 2(S_{\Vec{\alpha}}^c)\right] \nonumber\\
     =& \frac{1}{2^{D+1}N^2 (k+1)^2} \sum_{\Vec{\alpha}\in\{0,1\}^D}
    \left[ (-1)^{
     w(\Vec{\alpha})}\Re(S_{\Vec{\alpha}}^v) + S_{\Vec{\alpha}}^c\right] \label{eq: app EVar in terms of Sav and Sac} \\
     = & \frac{1}{2^{D+1}N^2 (k+1)^2} \sum_{\Vec{\alpha}\in\{0,1\}^D}
     \left[ (-1)^{
     w(\Vec{\alpha})}O_\text{gr}(R_{||}) + O_\text{gr}(R_{||})\right] \nonumber \\
     =& \frac{1}{2^{D+1}N^2 (k+1)^2} O_\text{gr}(N^2 (k+1)^2) \nonumber \\
     =& O_\text{gr}(1).
\end{align}
Therefore, to leading order, the normalization of the adjacency matrix exactly cancels the summation of the different terms. Unless both the expected values $\mathbb{E}\left[\bra{\bold{y}_j}D_{\Vec{\alpha}}\ket{\bold{x}_i}\bra{\bold{y}_l}D_{\Vec{\alpha}}\ket{\bold{x}_k}\right]$ and $\mathbb{E}\left[\bra{\bold{y}_j}D_{\Vec{\alpha}}\ket{\bold{x}_i}\bra{\bold{x}_k}D_{\Vec{\alpha}}^\dagger\ket{\bold{y}_l}\right]$ vanish\footnote{We will see later that this requires the ansatz itself to be very restrictive.}, the bound is strict as function of the graph size and connectivity, $\mathbb{E}[\Var(\partial_\mu \mathcal{L}_{QSGC})] = \Theta_\text{gr}(1)$.\\

We remind the reader that this is only an evaluation of the average variance over all possible inputs and targets. One can also find approximate (asymptotic) bounds of the full variance by also considering the spread $\Delta_\text{i,t}[\Var(\partial_\mu L)]$. One option is to approximate it further with a standard deviation, taking the variance $\sqrt{\Var_\text{i,t}(\Var(\partial_\mu \mathcal{L}_{QSGC}))}$, but this may be tricky as it assumes a symmetric spread which may not be true. We leave this task for future work.

\subsection{Calculation of $S_{\Vec{\alpha}}^v$ and $S_{\Vec{\alpha}}^c$ and dependence on $N_k$}\label{app: Dependence of the variance with the number of features}

Here, we calculate the explicit forms of the terms $S_{\Vec{\alpha}}^v$ and $S_{\Vec{\alpha}}^c$, by computing each of the terms separately. We begin with the ``$||$'' term in Eq.~\ref{eq: app Sav term}.

\begin{align*}
    \mathbb{E} &\left[
    \bra{\bold{y}_j}D_{\Vec{\alpha}}\ket{\bold{x}_i}
    \bra{\bold{y}_l}D_{\Vec{\alpha}}\ket{\bold{x}_k}
    \right] = \\
    &= \frac{1}{N_k^2(2^{N_k}-1)^2}
    \left(\sum_{a=1}^{N_k}\bra{e_a}\right)
    D_{\Vec{\alpha}}
    \left(\sum_{\Vec{\gamma}\neq 0}\frac{1}{\sqrt{w(\Vec{\gamma})}}\sum_{b=1}^{N_k}\gamma_b\ket{e_b}\right)
    \left(\sum_{c=1}^{N_k}\bra{e_c}\right)
    D_{\Vec{\alpha}}
    \left(\sum_{\Vec{\delta}\neq 0}\frac{1}{\sqrt{w(\Vec{\delta})}}\sum_{d=1}^{N_k}\delta_d\ket{e_d}\right).
\end{align*}
Here, we used expressions given in Eq.~\ref{eq:app target input average}, where $a$, $b$, $c$, and $d$ label the basis indices, and $\Vec{\gamma}$ and $\Vec{\delta}$ the bit string labels ranging from 1 to $2^{N_k}$. Rearranging the summations, the expression becomes
\begin{equation}
    \frac{1}{N_k^2(2^{N_k}-1)^2}
    \left[
    \sum_{a,b,c,d} \bra{e_a}D_{\Vec{\alpha}}\ket{e_b}
    \bra{e_c}D_{\Vec{\alpha}}\ket{e_d}
    \underbrace{\left(
    \sum_{\Vec{\gamma} \neq 0} \frac{\gamma_b}
    {\sqrt{w(\Vec{\gamma})}}
    \right)}_{H'_1}
    \underbrace{\left(
    \sum_{\Vec{\delta} \neq 0} \frac{\delta_d}
    {\sqrt{w(\Vec{\delta})}}
    \right)}_{H'_1}
    \right].
\end{equation}

The last two summations can be done independently (using Lemma \ref{lemma: app H numbers} in Appendix \ref{app: Helper lemmas}). For now, we simply denote the sums by $H'_1$, and will plug them in at the end. The expression for the desired term becomes \begin{equation}\label{eq: app first term Sav}
    \mathbb{E} 
    \left[
    \bra{\bold{y}_j}D_{\Vec{\alpha}}\ket{\bold{x}_i}
    \bra{\bold{y}_l}D_{\Vec{\alpha}}\ket{\bold{x}_k}
    \right] =
    \frac{(H'_1)^2}{N_k^2(2^{N_k}-1)^2}
    \left[
    \sum_{a,b,c,d} \bra{e_a}D_{\Vec{\alpha}}\ket{e_b}
    \bra{e_c}D_{\Vec{\alpha}}\ket{e_d}
    \right].
\end{equation}

Next, by considering the following relation
\begin{equation}\label{eq: app full superposition}
    \frac{1}{\sqrt{N_k}}\sum_{a=1}^{N_k}\ket{e_a} = \ket{\Vec{1}} =: \ket{+}^{\otimes n_k} = H^{\otimes n_k} \ket{0}^{\otimes n_k},
\end{equation}
the expression in Eq.~\ref{eq: app first term Sav} reduces to 
\begin{equation}
    \mathbb{E} 
    \left[
    \bra{\bold{y}_j}D_{\Vec{\alpha}}\ket{\bold{x}_i}
    \bra{\bold{y}_l}D_{\Vec{\alpha}}\ket{\bold{x}_k}
    \right] =
    \frac{(H'_1)^2}{(2^{N_k}-1)^2}\left(\bra{\Vec{1}}D_{\Vec{\alpha}}\ket{\Vec{1}}\right)^2.
\end{equation}

For simplicity, we continue with the ``$\lor$'' term in Eq.~\ref{eq: app Sav term},
\begin{align}
    \mathbb{E}\left[
    \bra{\bold{y}_j}D_{\Vec{\alpha}}\ket{\bold{x}_i}
    \bra{\bold{y}_j}D_{\Vec{\alpha}}\ket{\bold{x}_k}
    \right] =&
    \frac{1}{N_k(2^{N_k}-1)^2}\sum_{a,b,d}\bra{e_a}D_{\Vec{\alpha}}\ket{e_b}\bra{e_a}D_{\Vec{\alpha}}\ket{e_d}
    \sum_{\Vec{\gamma}\neq 0}\frac{\gamma_b}{\sqrt{w(\Vec{\gamma})}}
    \sum_{\Vec{\delta}\neq 0}\frac{\delta_d}{\sqrt{w(\Vec{\delta})}} \nonumber\\
    =& \frac{(H'_1)^2}{N_k(2^{N_k}-1)^2}\sum_{a,b,d}[D_{\Vec{\alpha}}]_{ab}[D_{\Vec{\alpha}}]_{ad} \nonumber\\
    =& \frac{(H'_1)^2}{(2^{N_k}-1)^2} \bra{\Vec{1}}D_{\Vec{\alpha}}^TD_{\Vec{\alpha}}\ket{\Vec{1}}, \label{eq: app term 2 in Sav}
\end{align}
where in the last two lines we wrote the matrix elements of $D_{\Vec{\alpha}}$ explicitly and then contracted using the state $\ket{\Vec{1}}$.

The ``$\land$'' term is calculated as
\begin{equation}
    \mathbb{E}\left[
    \bra{\bold{y}_j}D_{\Vec{\alpha}}\ket{\bold{x}_i}
    \bra{\bold{y}_l}D_{\Vec{\alpha}}\ket{\bold{x}_i}
    \right] =
    \frac{1}{N_k^2(2^{N_k}-1)} \sum_{a,b,c,d}\bra{e_a}D_{\Vec{\alpha}}\ket{e_b}\bra{e_c}D_{\Vec{\alpha}}\ket{e_d} \sum_{\Vec{\gamma}}\frac{\gamma_b\gamma_d}{w(\Vec{\gamma})}.
\end{equation}

Here, the summation over $\Vec{\gamma}$ branches into two cases, as can be seen in Eq.~\ref{eq:app H2 lemma} in Lemma \ref{lemma: app H numbers}. We name these values $H_1$ and $H_2'$, respectively for the cases where $b=d$ and $b\neq d$. The sum then reduces to
\begin{equation}
    \frac{1}{N_k^2(2^{N_k}-1)}
    \left[
    H_2' \sum_{\substack{a,b,c,d \\ b\neq d}} \bra{e_a}D_{\Vec{\alpha}}\ket{e_b}\bra{e_c}D_{\Vec{\alpha}}\ket{e_d}
    + H_1 \underbrace{\sum_{a,b,c}\bra{e_a}D_{\Vec{\alpha}}\ket{e_b}\bra{e_c}D_{\Vec{\alpha}}\ket{e_b}}_{=N_k\bra{\Vec{1}}D_{\Vec{\alpha}}D_{\Vec{\alpha}}^T\ket{\Vec{1}}}
    \right],
\end{equation}
where the last sum was simplified in a similar manner to the previous term (Eq.~\ref{eq: app term 2 in Sav}). We can then artificially add and subtract a copy of the second sum into the first one, to get the following expression:
\begin{align}
    \mathbb{E}\left[
    \bra{\bold{y}_j}D_{\Vec{\alpha}}\ket{\bold{x}_i}
    \bra{\bold{y}_l}D_{\Vec{\alpha}}\ket{\bold{x}_i}
    \right] =& \frac{1}{N_k^2(2^{N_k}-1)}
    \Bigl[
    H_2' \left(N_k^2(\bra{\Vec{1}}D_{\Vec{\alpha}}\ket{\Vec{1}})^2 - N_k\bra{\Vec{1}}D_{\Vec{\alpha}}D_{\Vec{\alpha}}^T\ket{\Vec{1}}\right) \nonumber \\
    &+ H_1 N_k\bra{\Vec{1}}D_{\Vec{\alpha}}D_{\Vec{\alpha}}^T\ket{\Vec{1}} 
    \Bigr] \nonumber \\
    =& \frac{H_2'}{2^{N_k}-1} \left[ 
    (\bra{\Vec{1}}D_{\Vec{\alpha}}\ket{\Vec{1}})^2 - \frac{1}{N_k}\bra{\Vec{1}}D_{\Vec{\alpha}}D_{\Vec{\alpha}}^T\ket{\Vec{1}} \right] \nonumber \\
    &+ \frac{H_1}{N_k(2^{N_k}-1)}\bra{\Vec{1}}D_{\Vec{\alpha}}D_{\Vec{\alpha}}^T\ket{\Vec{1}} \label{eq:app third term in Sav}.
\end{align}
To get this expression we wrote the summations over $a$, $b$, $c$, and $d$ the same way as for the previous terms.

The final, ``$|$'' term in Eq.~\ref{eq: app Sav term} undergoes the same branching over two cases as done above; the only difference is that now, the bra vectors are entangled, rather than independent:
\begin{align}
    \mathbb{E}\left[
    \bra{\bold{y}_j}D_{\Vec{\alpha}}\ket{\bold{x}_i}
    \bra{\bold{y}_j}D_{\Vec{\alpha}}\ket{\bold{x}_i}
    \right] =& 
    \frac{1}{N_k(2^{N_k}-1)} 
    \sum_{a,b,d} \bra{e_a}D_{\Vec{\alpha}}\ket{e_b}\bra{e_a}D_{\Vec{\alpha}}\ket{e_d}
    \sum_{\Vec{\gamma}\neq0}\frac{\gamma_b\gamma_d}{w(\Vec{\gamma})} \nonumber \\
    =& \frac{H_2'}{2^{N_k}-1}
    \left[ \bra{\Vec{1}}D_{\Vec{\alpha}}^TD_{\Vec{\alpha}}\ket{\Vec{1}} - \frac{1}{N_k}\Tr\left(D_{\Vec{\alpha}}D_{\Vec{\alpha}}^T\right)
    \right] \nonumber \\
    & + 
    \frac{H_1}{N_k(2^{N_k}-1)}
    \Tr\left(D_{\Vec{\alpha}}D_{\Vec{\alpha}}^T\right), \label{eq:app fourth term in Sav}
\end{align}
where the traces arise from the summation of the form $\sum_{a,b}(\bra{e_a}D_{\Vec{\alpha}}\ket{e_b})^2$.

To this end, we can now add up the terms to reconstruct $S_{\Vec{\alpha}}^v$ using Eq.~\ref{eq: app Sav term}, as well as the earlier calculations for each term.
\begin{align}
    S_{\Vec{\alpha}}^v = & 
    \frac{R_{||}(H'_1)^2}{(2^{N_k}-1)^2} \left(\bra{\Vec{1}}D_{\Vec{\alpha}}\ket{\Vec{1}}\right)^2 \nonumber\\
    & + \frac{R_{\lor}(H'_1)^2}{(2^{N_k}-1)^2} \bra{\Vec{1}}D_{\Vec{\alpha}}^TD_{\Vec{\alpha}}\ket{\Vec{1}} \nonumber\\
    & + \frac{R_{\land}H_2'}{2^{N_k}-1} \left[ 
    (\bra{\Vec{1}}D_{\Vec{\alpha}}\ket{\Vec{1}})^2 - \frac{1}{N_k}\bra{\Vec{1}}D_{\Vec{\alpha}}D_{\Vec{\alpha}}^T\ket{\Vec{1}} \right] + \frac{R_{\land}H_1}{N_k(2^{N_k}-1)}\bra{\Vec{1}}D_{\Vec{\alpha}}D_{\Vec{\alpha}}^T\ket{\Vec{1}} \nonumber \\
    & + \frac{R_{|}H_2'}{2^{N_k}-1}
    \left[ \bra{\Vec{1}}D_{\Vec{\alpha}}^TD_{\Vec{\alpha}}\ket{\Vec{1}} - \frac{1}{N_k}\Tr\left(D_{\Vec{\alpha}}D_{\Vec{\alpha}}^T\right)
    \right] + 
    \frac{R_{|}H_1}{N_k(2^{N_k}-1)}
    \Tr\left(D_{\Vec{\alpha}}D_{\Vec{\alpha}}^T\right),
\end{align}
which after rearranging terms, becomes
\begin{align}
    S_{\Vec{\alpha}}^v = &
    \left(\bra{\Vec{1}}D_{\Vec{\alpha}}\ket{\Vec{1}}\right)^2
    \left[
    \frac{R_{||}(H'_1)^2}{(2^{N_k}-1)^2} + \frac{R_{\land}H_2'}{2^{N_k}-1}
    \right] \nonumber \\
    & + \bra{\Vec{1}}D_{\Vec{\alpha}}^TD_{\Vec{\alpha}}\ket{\Vec{1}}
    \left[
    \frac{R_{\lor}(H'_1)^2}{(2^{N_k}-1)^2} + \frac{R_{|}H_2'}{2^{N_k}-1}
    \right] \nonumber \\
    & + \bra{\Vec{1}}D_{\Vec{\alpha}}D_{\Vec{\alpha}}^T\ket{\Vec{1}}
    \left[ 
    \frac{R_{\land}(H_1-H_2')}{N_k(2^{N_k}-1)}
    \right] \nonumber \\
    & +
    \Tr\left(D_{\Vec{\alpha}}D_{\Vec{\alpha}}^T\right)
    \left[
    \frac{R_{|}(H_1-H_2')}{N_k(2^{N_k}-1)}
    \right]. \label{eq: app Sav final form}
\end{align}

Using this rearranged formula, we gain a clearer view of the contributions to $S_{\Vec{\alpha}}$: There are four different terms whose prefactors depend solely on the ansatz (which can in turn depend on $N_k$), but usually not on $N$. The factors in square brackets are then different terms of the form $R_i\cdot f(N_k)$, where $R_i$ depends strictly on $N$ and the rest depends on $N_k$. 

The calculation for $S_{\Vec{\alpha}}^c$ defined in Eq.~\ref{eq: app sac term} is almost identical, except the second bracket in every term is flipped. Therefore, we can quickly write the final formula for this by performing the following changes:
\begin{equation}\label{eq: app Sav map to Sac}
    \begin{split}
    \left(\bra{\Vec{1}}D_{\Vec{\alpha}}\ket{\Vec{1}}\right)^2 = \bra{\Vec{1}}D_{\Vec{\alpha}}\ket{\Vec{1}}\bra{\Vec{1}}D_{\Vec{\alpha}}\ket{\Vec{1}}
    \quad\mapsto &\quad  \bra{\Vec{1}}D_{\Vec{\alpha}}\ket{\Vec{1}}\bra{\Vec{1}}D_{\Vec{\alpha}}^\dagger\ket{\Vec{1}} = 
    \left|\bra{\Vec{1}}D_{\Vec{\alpha}}\ket{\Vec{1}}\right|^2, \\
    \bra{\Vec{1}}D_{\Vec{\alpha}}^TD_{\Vec{\alpha}}\ket{\Vec{1}} = \sum_a \bra{e_a}D_{\Vec{\alpha}}\ket{\Vec{1}}\bra{e_a}D_{\Vec{\alpha}}\ket{\Vec{1}}
    \quad\mapsto &\quad 
    \sum_a \bra{e_a}D_{\Vec{\alpha}}\ket{\Vec{1}}\bra{\Vec{1}}D_{\Vec{\alpha}}^\dagger\ket{e_a} = 1, \\
    \bra{\Vec{1}}D_{\Vec{\alpha}}D_{\Vec{\alpha}}^T\ket{\Vec{1}} = \sum_a \bra{\Vec{1}}D_{\Vec{\alpha}}\ket{e_a}\bra{\Vec{1}}D_{\Vec{\alpha}}\ket{e_a}
    \quad\mapsto &\quad 
    \sum_a \bra{\Vec{1}}D_{\Vec{\alpha}}\ket{e_a}\bra{e_a}D_{\Vec{\alpha}}^\dagger\ket{\Vec{1}} = 1, \\
    \Tr(D_{\Vec{\alpha}}D_{\Vec{\alpha}}^T) = 
    \sum_{a,b} \bra{e_a}D_{\Vec{\alpha}}\ket{e_b}\bra{e_a}D_{\Vec{\alpha}}^\dagger\ket{e_b}
    \quad\mapsto &\quad 
    \sum_{a,b} \bra{e_a}D_{\Vec{\alpha}}\ket{e_b}\bra{e_b}D_{\Vec{\alpha}}^\dagger\ket{e_a} = N_k.
\end{split}
\end{equation}

Using this map, the expression for $S_{\Vec{\alpha}}^c$ reduces to:
\begin{align}\label{eq: app Sac final form}
    S_{\Vec{\alpha}}^c = &
    \left|\bra{\Vec{1}}D_{\Vec{\alpha}}\ket{\Vec{1}}\right|^2
    \left[
    \frac{R_{||}(H'_1)^2}{(2^{N_k}-1)^2} + \frac{R_{\land}H_2'}{2^{N_k}-1}
    \right] \nonumber \\ & +
    \left[
    \frac{R_{\lor}(H'_1)^2}{(2^{N_k}-1)^2} +
    \frac{R_{\land}(H_1-H_2')}{N_k(2^{N_k}-1)} +
    \frac{R_{|}(H_1)}{(2^{N_k}-1)}
    \right].
\end{align}


Here we note the following bounds:
\begin{align}
    \left|\left(\bra{\Vec{1}}D_{\Vec{\alpha}}\ket{\Vec{1}}\right)^2\right| \leq & \left|\bra{\Vec{1}}D_{\Vec{\alpha}}\ket{\Vec{1}}\right|^2 
    \leq 1, \label{eq: app bound 1}\\
    \left|\bra{\Vec{1}}D_{\Vec{\alpha}}^TD_{\Vec{\alpha}}\ket{\Vec{1}}\right|
    \leq & 1, \label{eq: app bound 2}\\
    \left|\bra{\Vec{1}}D_{\Vec{\alpha}}D_{\Vec{\alpha}}^T\ket{\Vec{1}}\right|
    \leq & 1, \label{eq: app bound 3}\\
    \left|\Tr\left(D_{\Vec{\alpha}}D_{\Vec{\alpha}}^T\right)\right|
    \leq & N_k. \label{eq: app bound 4}
\end{align}
These bounds respectively match the lines in the map \ref{eq: app Sav map to Sac} we used above to derive $S_{\Vec{\alpha}}^c$ from $S_{\Vec{\alpha}}^v$. Therefore, by taking the modulus in Eq.~\ref{eq: app Sav final form}, we can further obtain 
\begin{align}
    |S_{\Vec{\alpha}}^v| \leq &
    \left|\left(\bra{\Vec{1}}D_{\Vec{\alpha}}\ket{\Vec{1}}\right)^2\right|
    \left[
    \frac{R_{||}(H'_1)^2}{(2^{N_k}-1)^2} + \frac{R_{\land}H_2'}{2^{N_k}-1}
    \right] + \left|\bra{\Vec{1}}D_{\Vec{\alpha}}^TD_{\Vec{\alpha}}\ket{\Vec{1}}\right|
    \left[
    \frac{R_{\lor}(H'_1)^2}{(2^{N_k}-1)^2} + \frac{R_{|}H_2'}{2^{N_k}-1}
    \right] \nonumber \\
    & + \left|\bra{\Vec{1}}D_{\Vec{\alpha}}D_{\Vec{\alpha}}^T\ket{\Vec{1}}\right|
    \left[ 
    \frac{R_{\land}(H_1-H_2')}{N_k(2^{N_k}-1)}
    \right] +
    \left|\Tr\left(D_{\Vec{\alpha}}D_{\Vec{\alpha}}^T\right)\right|
    \left[
    \frac{R_{|}(H_1-H_2')}{N_k(2^{N_k}-1)}
    \right] \nonumber \\
    \leq & S_{\Vec{\alpha}}^c. \label{eq: app Sav less than Sac}
\end{align}
This will become relevant for the final calculation of the asymptotic behavior of $\mathbb{E}[\Var(\partial_\mu \mathcal{L}_{QSGC})]$.
\newline

To this end, we can simplify the results in Eqs.~\ref{eq: app Sav final form} and \ref{eq: app Sac final form} by approximating to leading order in $N_k$. For this, we use the results in Lemma \ref{lemma: app H approx}:
\begin{align}
    S_{\Vec{\alpha}}^v = &
    \left(\bra{\Vec{1}}D_{\Vec{\alpha}}\ket{\Vec{1}}\right)^2
    \left[
    \frac{R_{||}(H'_1)^2}{(2^{N_k}-1)^2} + \frac{R_{\land}H_2'}{2^{N_k}-1}
    \right] + \bra{\Vec{1}}D_{\Vec{\alpha}}^TD_{\Vec{\alpha}}\ket{\Vec{1}}
    \left[
    \frac{R_{\lor}(H'_1)^2}{(2^{N_k}-1)^2} + \frac{R_{|}H_2'}{2^{N_k}-1}
    \right] \nonumber \\
    & + \bra{\Vec{1}}D_{\Vec{\alpha}}D_{\Vec{\alpha}}^T\ket{\Vec{1}}
    \left[ 
    \frac{R_{\land}(H_1-H_2')}{N_k(2^{N_k}-1)}
    \right] +
    \Tr\left(D_{\Vec{\alpha}}D_{\Vec{\alpha}}^T\right)
    \left[
    \frac{R_{|}(H_1-H_2')}{N_k(2^{N_k}-1)}
    \right] \nonumber \\
    = & \left(\bra{\Vec{1}}D_{\Vec{\alpha}}\ket{\Vec{1}}\right)^2
    (R_{||}+R_\land)\left(\frac{1}{2N_k}+O(N_k^{-2})\right)+
    \bra{\Vec{1}}D_{\Vec{\alpha}}^TD_{\Vec{\alpha}}\ket{\Vec{1}}
    (R_{\lor}+R_{|})\left(\frac{1}{2N_k}+O(N_k^{-2})\right) \nonumber \\
    & + \bra{\Vec{1}}D_{\Vec{\alpha}}D_{\Vec{\alpha}}^T\ket{\Vec{1}}
    \frac{R_\land}{N_k} \left(\frac{1}{2N_k}+O(N_k^{-2})\right) +
    \Tr\left(D_{\Vec{\alpha}}D_{\Vec{\alpha}}^T\right)
    R_{|}\left(\frac{1}{2N_k^2}+O(N_k^{-3})\right) \nonumber \\
    = & \left[
    \left(\bra{\Vec{1}}D_{\Vec{\alpha}}\ket{\Vec{1}}\right)^2
    (R_{||}+R_\land) +
    \bra{\Vec{1}}D_{\Vec{\alpha}}^TD_{\Vec{\alpha}}\ket{\Vec{1}}
    (R_{\lor}+R_{|}) +
    \bra{\Vec{1}}D_{\Vec{\alpha}}D_{\Vec{\alpha}}^T\ket{\Vec{1}}
    \frac{R_\land}{N_k} +
    \Tr\left(D_{\Vec{\alpha}}D_{\Vec{\alpha}}^T\right)
    \frac{R_{|}}{N_k}
    \right]
    \nonumber \\
    & \cdot \left(\frac{1}{2N_k}+O(N_k^{-2})\right). \label{eq: app Sav approx}
\end{align}
Here, we used Lemma \ref{lemma: app H approx} to plug in the approximations for the coefficients of the $R$-terms, and then grouped them into a single one by adding a factor $N_k^{-1}$ to the trace.

Similarly, for $S_{\Vec{\alpha}}^c$, using the map given above in Eq.~\ref{eq: app Sav map to Sac}, we get 
\begin{equation}\label{eq: app Sac approx}
    S_{\Vec{\alpha}}^c = 
    \left[
    \left|\bra{\Vec{1}}D_{\Vec{\alpha}}\ket{\Vec{1}}\right|^2
    (R_{||}+R_\land)+
    1\cdot
    (R_{\lor}+\frac{R_\land}{N_k}+2R_{|})
    \right]
    \left(\frac{1}{2N_k}+O(N_k^{-2})\right).
\end{equation}
Using the second inequality in \ref{eq: app bound 1}, We know that in the asymptotic limit of $N_k$ the first term in the square brackets above is $O(1)$, and that $R_\lor+2R_{|}=\Theta(1)$, so we can infer
\begin{equation}
    S_{\Vec{\alpha}}^c = \Theta(1)\left(\frac{1}{2N_k}+O(N_k^{-2})\right) = \Theta(N_k^{-1}).
\end{equation}
For $S_{\Vec{\alpha}}^v$, the bound cannot be tightened more without adding constraints to the ansatz.

Assembling the expected variance using Eq.~\ref{eq: app EVar in terms of Sav and Sac}:
\begin{align}
    \mathbb{E}[\Var(\partial_\mu \mathcal{L}_{QSGC})] =&
    \frac{1}{2^{D+1}N^2 (k+1)^2} 
    \sum_{\Vec{\alpha}\in\{0,1\}^D}
    \left[ (-1)^{w(\Vec{\alpha})}
    \Re(S_{\Vec{\alpha}}^v) 
    + 
    S_{\Vec{\alpha}}^c\right] \nonumber \\
    =& \frac{1}{2^{D+1}N^2 (k+1)^2} 
    \left[
    \sum_{\Vec{\alpha}\in\{0,1\}^D}
    (-1)^{w(\Vec{\alpha})}
    \Re(S_{\Vec{\alpha}}^v)
    +
    2^D\Theta(N_k^{-1})
    \right] \nonumber \\
    =& 
    \frac{1}{2^{D+1}N^2 (k+1)^2}\sum_{\Vec{\alpha}\in\{0,1\}^D}
    (-1)^{w(\Vec{\alpha})}
    \Re(S_{\Vec{\alpha}}^v)
    +\Theta(N_k^{-1}). \label{eq: app EVar alternating sum problem}
\end{align}
If the remaining summation in the RHS is negative, it could cancel out the leading term in $\Theta(N_k^{-1})$. However, one can employ a trick to constrain the ansatz slightly so that the sum is guaranteed to be positive. We prove this in Lemma \ref{lemma: app alternating sum}. Therefore, the leading order of the variance is indeed $\Theta(N_k^{-1})$.

\subsection{Bounds for the expected variance}\label{app: Bounds for the expected variance}

Finally, we compute upper and lower bounds for the expected variance $\mathbb{E}[\Var(\partial_\mu \mathcal{L}_{QSGC})]$.

\subsubsection{Upper bound}\label{app: Upper bound}

Given the bound \ref{eq: app Sav less than Sac}, we can calculate an upper bound of the expected variance as such:
\begin{align}
    \mathbb{E}[\Var(\partial_\mu \mathcal{L}_{QSGC})] =&
    \frac{1}{2^{D+1}N^2 (k+1)^2} 
    \sum_{\Vec{\alpha}\in\{0,1\}^D}
    \underbrace{\left[ (-1)^{w(\Vec{\alpha})}
    \Re(S_{\Vec{\alpha}}^v) 
    + S_{\Vec{\alpha}}^c\right]}_{\geq 0} \nonumber \\
    =&
    \frac{1}{2^{D+1}N^2 (k+1)^2} 
    \sum_{\Vec{\alpha}\in\{0,1\}^D}
    \Bigl| (-1)^{w(\Vec{\alpha})}
    \Re(S_{\Vec{\alpha}}^v) + 
    \underbrace{S_{\Vec{\alpha}}^c}_{\geq0}\Bigr| \nonumber \\
    \leq & 
    \frac{1}{2^{D+1}N^2 (k+1)^2} 
    \sum_{\Vec{\alpha}\in\{0,1\}^D}
    \biggl[\underbrace{\left| (-1)^{w(\Vec{\alpha})}
    \Re(S_{\Vec{\alpha}}^v)\right|}_{\leq S_{\Vec{\alpha}}^c} + 
    S_{\Vec{\alpha}}^c\biggr] \nonumber \\
    \leq& \frac{1}{2^DN^2 (k+1)^2}
    \sum_{\Vec{\alpha}\in\{0,1\}^D} S_{\Vec{\alpha}}^c.
\end{align}

Next, we can apply bound \ref{eq: app bound 1} on $S_{\Vec{\alpha}}^c$ to 
\begin{equation}
    S_{\Vec{\alpha}}^c \leq 
    \left[
    (R_{||}+R_\land)+
    (R_{\lor}+\frac{R_\land}{N_k}+2R_{|})
    \right]
    \left(\frac{1}{2N_k}+O(N_k^{-2})\right)\equiv S.
\end{equation}
After neglecting the $R_\land/N_k$ term (that contributes as $O(N_k^{-2})$ and plugging in the values from Table \ref{tab: app R coefficients}, we get
\begin{equation} \label{eq: app maximum Sac}
    S = \left[
    N^2(k+1)^2+N(k+1)
    \right]
    \left(\frac{1}{2N_k}+O(N_k^{-2})\right).
\end{equation}
Therefore, the expected variance is upper bounded by 
\begin{align}
    \mathbb{E}[\Var(\partial_\mu \mathcal{L}_{QSGC})] \leq&
    \frac{1}{2^DN^2 (k+1)^2}2^D
    S \nonumber \\
    =& \left[
    1+\frac{1}{N(k+1)}
    \right]
    \left(\frac{1}{2N_k}+O(N_k^{-2})\right) = \Theta(N_k^{-1}) \equiv \Theta(C^{-1}). \label{eq: app upper bound}
\end{align}
Finding whether this bound is tight is a work in progress, as one needs to find the right ansatz that maximizes the expected variance.

\subsubsection{Lower bound}\label{app: Lower bound}

Using Lemma \ref{lemma: app alternating sum}, we can always force the sum over $\Re(S_{\Vec{\alpha}}^v)$ to be nonnegative. Therefore, given the additional constraint in the lemma, a lower bound for the expected variance is the following:
\begin{align}
    \mathbb{E}[\Var(\partial_\mu \mathcal{L}_{QSGC})] =&
    \frac{1}{2^{D+1}N^2 (k+1)^2} 
    \Biggl[
    \underbrace{\sum_{\Vec{\alpha}\in\{0,1\}^D}
    (-1)^{w(\Vec{\alpha})}
    \Re(S_{\Vec{\alpha}}^v)}_{\geq 0}
    +
    \sum_{\Vec{\alpha}\in\{0,1\}^D}S_{\Vec{\alpha}}^c
    \Biggr] \nonumber \\
    \geq& \frac{1}{2^{D+1}N^2 (k+1)^2}
    \sum_{\Vec{\alpha}\in\{0,1\}^D}S_{\Vec{\alpha}}^c.
\end{align}
To give a lower bound for $S_{\Vec{\alpha}}^c$, the term $\left|\bra{\Vec{1}}D_{\Vec{\alpha}}\ket{\Vec{1}}\right|^2$ can be taken to 0:
\begin{align}
    S_{\Vec{\alpha}}^c = &
    \biggl[
    \underbrace{\left|\bra{\Vec{1}}D_{\Vec{\alpha}}\ket{\Vec{1}}\right|^2}_{\geq 0}
    (R_{||}+R_\land)+
    1\cdot
    (R_{\lor}+\underbrace{\frac{R_\land}{N_k}}_{=O(N_k^{-1})} +2R_{|})
    \biggr]
    \left(\frac{1}{2N_k}+O(N_k^{-2})\right) \nonumber \\
    \geq & (R_\lor+2R_{|})
    \left(\frac{1}{2N_k}+O(N_k^{-2})\right) \nonumber \\
    =& [N(k+1)^2+N(k+1)]\left(\frac{1}{2N_k}+O(N_k^{-2})\right),
\end{align}
which means the expected variance is ultimately lower-bounded by
\begin{equation}
    \mathbb{E}[\Var(\partial_\mu \mathcal{L}_{QSGC})] \geq 
    \frac{1}{4}\left[\frac{1}{N}+\frac{1}{N(k+1)}\right]\left(\frac{1}{N_k}+O(N_k^{-2})\right)=\Theta(N^{-1}N_k^{-1}) \equiv \Theta(N^{-1}C^{-1}).
\end{equation}

We emphasize on the fact that this bound can only be attained if the term $\left|\bra{\Vec{1}}D_{\Vec{\alpha}}\ket{\Vec{1}}\right|^2=0$, which is a significant constraint to the ansatz.

\subsection{Note on the curse of dimensionality}\label{app: Note on the curse of dimensionality}

The variance bounds given above are computed by considering the strict bounds on the leading term, namely, $$0\leq\left|\bra{\Vec{1}}D_{\Vec{\alpha}}\ket{\Vec{1}}\right|^2\leq1.$$ We note here that in reality, for most VQCs\footnote{By this we mean that the set of generating gates $V_\eta$ and the fixed gates $W_\eta$ are arbitrary.}, attaining either of these ends is highly unlikely. However, because of the vanishing overlap between two arbitrary states (the \textbf{curse of dimensionality})\cite{Cerezo_Larocca_Garcia-Martin_Diaz_Braccia_Fontana_Rudolph_Bermejo_Ijaz_Thanasilp_etal._2025}, the distribution of possible moduli $\left|\bra{\Vec{1}}D_{\Vec{\alpha}}\ket{\Vec{1}}\right|^2$ strongly peaks around $N_k^{-1}$. We prove this with Lemma \ref{lemma: app curse of dim}, in Appendix \ref{app: Helper lemmas}.

Therefore, most (arbitrarily chosen) ansätze lead to
\begin{align}
    S_{\Vec{\alpha}}^c = &
    \biggl[
    \underbrace{\left|\bra{\Vec{1}}D_{\Vec{\alpha}}\ket{\Vec{1}}\right|^2}_{\sim N_k^{-1}}
    (R_{||}+R_\land)+
    R_{\lor}+\frac{R_\land}{N_k} +2R_{|}
    \biggr]
    \left(\frac{1}{2N_k}+O(N_k^{-2})\right) \nonumber \\
    \sim & \left[R_\lor+2R_{|}+\frac{R_{||}+2R_\land}{N_k}\right]\left(\frac{1}{2N_k}+O(N_k^{-2})\right) \nonumber \\
    =& \Theta(Nk^2N_k^{-1})+O(N^2k^2N_k^{-2}).
\end{align}

Plugging this in Eq.~\ref{eq: app EVar in terms of Sav and Sac}, we see that apart from the $S_{\Vec{\alpha}}^v$ alternating sum (which as seen above, can be bounded between 0 and the $S_{\Vec{\alpha}}^c$ sum), most ansätze will have an additional $O(N_k^{-2})$ term\footnote{By expanding the $H$ terms in Lemma \ref{lemma: app H approx} to next order, one can check that the leading term is present, so we upgrade the Big-$O$ notation to a stricter $\Theta(N_k^{-2})$.}:
\begin{align}
    \mathbb{E}[\Var(\partial_\mu \mathcal{L}_{QSGC})] \sim& \Theta(N^{-1}N_k^{-1}) + \Theta(N_k^{-2}) \nonumber \\
    =& \Theta(N^{-1}C^{-1})+\Theta(C^{-2}).
\end{align}

This implies on the one hand that most ansätze encounter the flattening with graph size, but at a much harsher tradeoff with the $N_k$ scaling.

\subsection{Example: IQP ansatz}\label{app: Example: IQP ansatz}

In this section we show the application of the results above on a worked example: the IQP ansatz. This takes the form:
\begin{equation}
    U(\theta)=H^{\otimes n_k}\exp\left(
    \sum_{j=1}^{n_k} i\theta_j Z_j + \sum_{j<k} i\theta_{jk} Z_jZ_k
    \right)H^{\otimes n_k}.
\end{equation}
In terms of our general ansatz formula, each $\theta_j$ and $\theta_{jk}$ is a parameter, so the depth is
\begin{equation}
    D=n_k+\frac{n_k(n_k-1)}{2} = \frac{n_k(n_k+1)}{2}.
\end{equation}
The $W_\eta$ gates are:
\begin{equation}
    W_\eta = \begin{cases}
        \mathds{1} \quad \text{for} \quad \eta \in\{1,D+1\}, \\
        H^{\otimes n_k} \quad \text{otherwise,}
    \end{cases}\\
\end{equation}
and the $V_\eta$ are:
\begin{equation}
    V_\eta = \begin{cases}
        Z_j, \quad\text{for} \quad 1\leq \eta=j \leq n_k, \\
        Z_jZ_k, \quad \text{for} \quad n_k\leq \eta = j\cdot n_k + k \leq D.
    \end{cases}
\end{equation}
Then, when computing the expected variance we need $D_{\Vec{\alpha}}$, which takes the following form:
\begin{equation}
    D_{\Vec{\alpha}}  = H^{\otimes n_k} \cdot \prod_{\eta}V_\eta^{\alpha_\eta}\cdot H^{\otimes n_k},
\end{equation}
where the product ends up being a tensor product of identities and $Z$ gates. Note that all gates in $D_{\Vec{\alpha}}$ are real, square to identity, and the $Z$ gates commute, so $D_{\Vec{\alpha}}=D_{\Vec{\alpha}}^T=D_{\Vec{\alpha}}^\dagger$. Therefore, the following simplifications occur:
\begin{align}
    \bra{\Vec{1}}D_{\Vec{\alpha}}\ket{\Vec{1}} =& \bra{\Vec{1}}H^{\otimes n_k} \cdot \prod_{\eta}V_\eta^{\alpha_\eta}\cdot H^{\otimes n_k}\ket{\Vec{1}} = \bra{0}\prod_{\eta}V_\eta^{\alpha_\eta}\ket{0} = 1 \\
    \bra{\Vec{1}}D_{\Vec{\alpha}}^TD_{\Vec{\alpha}}\ket{\Vec{1}} =& \bra{\Vec{1}}D_{\Vec{\alpha}}D_{\Vec{\alpha}}^T\ket{\Vec{1}} = \braket{\Vec{1}|\Vec{1}} = 1 \\
    \Tr(D_{\Vec{\alpha}}D_{\Vec{\alpha}}^T) =&
    \Tr(\mathds{1})=N_k.
\end{align}
Hence one can observe that the bounds from Eqs.~\ref{eq: app bound 1}-\ref{eq: app bound 4} are reached. To this end, the bound in Eq.~\ref{eq: app Sav less than Sac} is also attained, 
\begin{equation}
    S_{\Vec{\alpha}}^v=S_{\Vec{\alpha}}^c=
    \left[
    (R_{||}+R_\land)+
    (R_{\lor}+\frac{R_\land}{N_k}+2R_{|})
    \right]
    \left(\frac{1}{2N_k}+O(N_k^{-2})\right)= S.
\end{equation}
Using Eq.~\ref{eq: app maximum Sac}, the full expected variance is
\begin{align}
    \mathbb{E}[\Var(\partial_\mu \mathcal{L}_{QSGC})] =&
    \frac{1}{2^{D+1}N^2 (k+1)^2} 
    \sum_{\Vec{\alpha}\in\{0,1\}^D}
    \left[ (-1)^{w(\Vec{\alpha})}
    \Re(S_{\Vec{\alpha}}^v) 
    + 
    S_{\Vec{\alpha}}^c\right] \nonumber \\
    =& \frac{1}{2^{D+1}N^2 (k+1)^2}S
    \underbrace{\sum_{\Vec{\alpha}\in\{0,1\}^D}\left[(-1)^{w(\Vec{\alpha})}+1\right]}_{2^D} \nonumber \\
    =& \left[
    1+\frac{1}{N(k+1)}
    \right]\left(\frac{1}{4N_k}+O(N_k^{-2})\right).
\end{align}
Therefore, to leading order, this expected variance does flatten in the limit of large $N$ and $k$, and decays as $N_k^{-1}$.

We note that whilst $S_{\Vec{\alpha}}^v$ and $S_{\Vec{\alpha}}^c$ themselves have been maximised using this ansatz, the summation over $\Vec{\alpha}$ has not. This is because the $S_{\Vec{\alpha}}^v$ contribution cancels out. In fact, this expected variance is exactly half the maximum possible. This also means that this ansatz has an optimal asymptotic form, up to a factor of 2.

\subsection{What is $C < N_k$?}\label{What is $C < N_k$?}

The expected variance calculation in Sec.~\ref{app: Dependence of the variance with the number of features} reduces nicely the moment we obtain the full superposition in Eq.~\ref{eq: app full superposition}, together with the corresponding $H$ numbers ($H_1$, $H_1'$, and $H_2$). They all depend heavily on the assumption that $C=N_k=2^{n_k}$. This is usually not the case, and for this reason we need to address possible extensions of our modeling of the trainability.

\paragraph{Padding.} In this case, a standard amplitude encoding scheme will assume a padding of the dimension-$C$ feature vectors (and labels) with $N_k-C$ zeros, so that it fits register $Reg(k)$. This changes the estimators in Eq.~\ref{eq:app target input average} to 
\begin{equation}\label{eq:app new target input average}
    \begin{split}
        \mathbb{E}[f(\ket{\bold{y}_i})] =& 
        \frac{1}{C}\sum_{a=1}^{C}f(\ket{\bold{e}_a}), \\
        \mathbb{E}[f(\ket{\bold{x}_i})] =& 
        \frac{1}{2^{C}-1}\sum_{\substack{\Vec{\gamma}\in\{0,1\}^{C} \\ \Vec{\gamma}\neq\Vec{0}}}f\left(\frac{1}{\sqrt{ w(\vec{\gamma})}}\sum_{b=1}^{C}\gamma_b
        \ket{\bold{e}_b}\right).
    \end{split}
\end{equation}

In our previous calculation, this effectively cancels out all terms that lie in the $\{\ket{\bold{e}_b}\}_{b\in(C, N_k]}$. One can see that this is equivalent to mapping 
\begin{equation}
    D_{\Vec{\alpha}}\mapsto D'_{\Vec{\alpha}} \equiv \Pi\cdot D_{\Vec{\alpha}} \cdot \Pi,
\end{equation}
where $\Pi$ projects onto the feature space. As this new matrix is not unitary, we cannot infer the same results on the bounds and general behavior of the expected variance.

\paragraph{Random permutations.} A solution that recovers the previous results approximately is to take the padded feature vectors and encoded labels and randomly permute their components (classically) before loading them on the qantum circuit using the amplitude encoding.\footnote{We assume that this process is included in the classical compilation performed with an efficient QRAM, which we must assume anyway in Sec.~\ref{sec: Note on Complexity and Dequantization}.} This means that the classical estimators now become
\begin{equation}\label{eq:app permuted target input average}
    \begin{split}
        \mathbb{E}[f(\ket{\bold{y}_i})] =& 
        \frac{1}{N_k}\sum_{a=1}^{N_k}f(\ket{\bold{e}_a}), \\
        \mathbb{E}[f(\ket{\bold{x}_i})] =& 
        \frac{1}{V(N_k,C)-1}\sum_{\substack{\Vec{\gamma}\in\{0,1\}^{N_k} \\ w(\Vec{\gamma})\in [1,C]}}f\left(\frac{1}{\sqrt{ w(\vec{\gamma})}}\sum_{b=1}^{N_k}\gamma_b
        \ket{\bold{e}_b}\right),
    \end{split}
\end{equation}
where $V(n,w)$ is the number of $n$-bit strings with Hamming weight up to $w$. Due to the uniformity of the expected input and target states for all nodes, this random permutation recovers all of the previous results up to an overhead $O(1)$. Specifically, the only difference is recalculating the $H$-terms. In Lemma \ref{lemma: app reduced H terms}, we show that because $C>N_k/2$, these terms are comparable to our initial $H_1$, $H_2$, and $H_1'$. For the same reason, $2^{N_k-1}\approx V(N_k,N_k/2)\leq V(N_k,C)\leq 2^{N_k}$, so $V(N_k,C)=\Theta(2^{N_k})$. Therefore, applying this random permutation recovers the expected variance results. \\

Although permuting the components seems like an ad-hoc solution, so is the specific use of the first $C$ basis vectors out of the full Hilbert space --- in the sense that a generic ansatz will not be trained specifically for those basis components. Furthermore, if one does want to recover the training for the standard choice of amplitude encoding, it is enough to encode the inverse permutation in the circuit. Formally speaking, if we precompile a permutation $\sigma\in S_{N_k}$ for the amplitude encoding process, for any ansatz $D_{\Vec{\alpha}}$, there exists another ansatz $D'_{\Vec{\alpha}}$ that permutes the elements back into position and is trained as originally planned. This new ansatz takes the form 
\begin{equation}
    D'_{\Vec{\alpha}} \equiv W(\sigma) \cdot D_{\Vec{\alpha}} \cdot W(\sigma^{-1}),
\end{equation}
where $W(\sigma)$ and $W(\sigma^{-1})$ temporarily permute the feature vector components to first $C$ axes, apply the VQC, and then permute them back. These gates can be realized as presented in \cite{nielsen_quantum_2010, herbert2024almostoptimalcomputationalbasisstate}. As expected, this change of ansatz does not change the expected variance calculated above, as that is averaged over all permutations.

Roughly speaking, this calculation tells us that having to encode the data into a strictly larger Hilbert space will be expected to have the same trainability, possibly requiring an additional shuffle of the input and target state components in the pre-compilation step --- which is a degree of freedom we can exploit.



\section{Generalizing the trainability results across different QGCNs}\label{app: Generalizing the trainability results across different QGCNs}

In this section, we show that by maintaining the two assumptions from Sections \ref{app: Results for regular graphs} and \ref{app: Expected values over the input and target states}, namely the $k$-regularity and averaging across different inputs and targets, we can easily generalize our first order QSGC result to any QSGC or even QLGC architecture. For that, we first need to make a stronger generalization of the results above to \textit{graphs with weighted edges}. 


\subsection{Graphs with weighted edges}\label{app: Graphs with weighted edges}

By generalizing the architecture to weighted graphs, we allow the matrix $\hat{A}_{ij}$ to contain custom edge weights by inheriting them from the unnormalized adjacency matrix $\tilde{A}_{ij}$; see Appendix \ref{app: Weighted graphs} for details.

Therefore the results on the bias and variance of the cost function gradient stay the same for arbitrary choices of the matrix $\hat{A}$. We recall the formulas:
\begin{align}
    \mathbb{E}_{\thetaB}[\partial_\mu \mathcal{L}_1] =& \frac{1}{(2\pi)^DN}\sum_{i,j=1}^{N} \hat{A}_{ij}\int \dd \thetaB'
    \bra{\bold{y}_j} U(\thetaB',\theta_\mu=\pi)-U(\thetaB',\theta_\mu=-\pi)\ket{\bold{x}_i} = 0; \\
    \Var_{\thetaB}[\partial_\mu \mathcal{L}_1] =& \frac{1}{2^D N^2}\sum_{i,j,k,l=1}^{N}\hat{A}_{ij}\hat{A}_{kl}\sum_{\Vec{\alpha}\in\{0,1\}^D} (-1)^{ w(\Vec{\alpha})} 
    \bra{\bold{y}_j} D_{\Vec{\alpha}} \ket{\bold{x}_i}
    \bra{\bold{y}_l} D_{\Vec{\alpha}}\ket{\bold{x}_k}, \label{eq: app recall variance}\\
    \Cov_{\thetaB}[\partial_\mu \mathcal{L}_1, \partial_\mu \mathcal{L}_1^*] =& \frac{1}{2^D N^2}\sum_{i,j,k,l=1}^{N}\hat{A}_{ij}\hat{A}_{kl}\sum_{\Vec{\alpha}\in\{0,1\}^D}
    \bra{\bold{y}_j} D_{\Vec{\alpha}}\ket{\bold{x}_i}
    \bra{\bold{x}_k} D_{\Vec{\alpha}}^\dagger \ket{\bold{y}_l}.
\end{align}
We see that for any choice of $\hat{A}$, the bias is zero, thus maintaining the results for the variance and covariance terms. For these terms, we first take the expected value over the input and target \textit{before} making the regular graph assumption:
\begin{align*}
    \mathbb{E}_\text{i,t}[\Var[\partial_\mu \mathcal{L}_1]] =& \frac{1}{2^D N^2}\sum_{i,j,k,l=1}^{N}\hat{A}_{ij}\hat{A}_{kl}\sum_{\Vec{\alpha}\in\{0,1\}^D} (-1)^{ w(\Vec{\alpha})} 
    \mathbb{E}_\text{i,t}[\bra{\bold{y}_j} D_{\Vec{\alpha}} \ket{\bold{x}_i}
    \bra{\bold{y}_l} D_{\Vec{\alpha}}\ket{\bold{x}_k}], \\
    \mathbb{E}_\text{i,t}[\Cov[\partial_\mu \mathcal{L}_1, \partial_\mu \mathcal{L}_1^*]] =& \frac{1}{2^D N^2}\sum_{i,j,k,l=1}^{N}\hat{A}_{ij}\hat{A}_{kl}\sum_{\Vec{\alpha}\in\{0,1\}^D}
    \mathbb{E}_\text{i,t}[\bra{\bold{y}_j} D_{\Vec{\alpha}}\ket{\bold{x}_i}
    \bra{\bold{x}_k} D_{\Vec{\alpha}}^\dagger \ket{\bold{y}_l}].
\end{align*}
Here, we observe the same splitting of the sum as in Appendix \ref{app: Expected values over the input and target states}, depending on whether $i=k$ and/or $j=l$. The $R$ values calculated above can thus be generalized as:
\begin{align}
    \hat{R}_{||} :=& \sum_{\substack{i \neq k \\ j \neq l}} \hat{A}_{ij}\hat{A}_{kl} = \sum_{i,j,k,l} \hat{A}_{ij}\hat{A}_{kl} - \hat{R}_{\land} - \hat{R}_{\lor} - \hat{R}_{|}, \label{eq: app R hat ||}\\
    \hat{R}_{\land} :=&  \sum_{i,j\neq l} \hat{A}_{ij}\hat{A}_{il} = \sum_{i,j,l} \hat{A}_{ij}\hat{A}_{il} - \hat{R}_{|}, \label{eq: app R hat land}\\
    \hat{R}_{\lor} :=&  \sum_{i \neq k, j} \hat{A}_{ij}\hat{A}_{kj} = \sum_{i,k,l} \hat{A}_{ij}\hat{A}_{kj} - \hat{R}_{|}, \label{eq: app R hat lor}\\
    \hat{R}_{|} :=& \sum_{i,j}\hat{A}_{ij}\hat{A}_{ij} = \Tr\left(\hat{A}^2\right). \label{eq: app R hat |}
\end{align}
In the regular graph case (Appendices \ref{app: Results for regular graphs} and \ref{app: Expected values over the input and target states}) the $\hat{R}$ notation is simply the $R$ notation divided by $(k+1)^2$.

\paragraph{Weighted regular graphs.} For the case when the graph is weighted $k$-regular, meaning that $\tilde{D}_{ii}=(k+1)\mathds{1}$, we use Lemma \ref{lemma: app R hat bounds} to bound the $\hat{R}$ values above. In terms of orders, the $\hat{R}$ terms maintain the orders from the unweighted graph case, although they depend heavily on $\hat{R}_{|}$:
\begin{align}
    \hat{R}_{||} =& N(N-2)+\hat{R}_{|} = \Theta(N^2), \\
    \hat{R}_\land = \hat{R}_\lor =& N - \hat{R}_{|} = O(N), \\
    \hat{R}_{|} =& O(N).
\end{align}

Therefore, we can guarantee that to leading order, the trainability results from Appendix \ref{app:Detailed trainability calculations} hold for the case of weighted graphs.

\paragraph{Upper bound.} If $\left|\braket{\Vec{1}|D_{\Vec{\alpha}}|\Vec{1}}\right|^2=1$, then repeating the calculation in Appendix \ref{app: Upper bound} leads to $S_{\Vec{\alpha}}^c \leq S$ where
\begin{align}
    S =& \left[
    (R_{||}+R_\land)+
    (R_{\lor}+\frac{R_\land}{N_k}+2R_{|})
    \right]
    \left(\frac{1}{2N_k}+O(N_k^{-2})\right) \nonumber \\
    =& (k+1)^2\left[
        \left(\sum_{i,j}\hat{A}_{ij}\right)^2 + \hat{R}_{|} +\frac{1}{N_k}\left(\sum_{i,j,k}\hat{A}_{ij}\hat{A}_{kj}-\hat{R}_{|}\right)
    \right]
    \left(\frac{1}{2N_k}+O(N_k^{-2})\right) \nonumber \\
    =& \Theta(N^2 k^2 N_k^{-1}) \quad\quad \bigg( = \Theta(N^2 k^2 N_k^{-1})\bigg).
\end{align}
As derived for the unweighted case, this leads to 
\begin{equation}\label{eq: app upper bound general}
    \mathbb{E}[\Var(\partial_\mu \mathcal{L}_{QSGC})]_\text{upper} =
    \Theta(C^{-1}).
\end{equation}

\paragraph{Lower bound.} If $\left|\braket{\Vec{1}|D_{\Vec{\alpha}}|\Vec{1}}\right|^2=1$, then repeating the calculation in Appendix \ref{app: Lower bound} leads to 
\begin{align}
    S_{\Vec{\alpha}}^c \geq& 
    (R_\lor+2R_{|})
    \left(\frac{1}{2N_k}+O(N_k^{-2})\right) \nonumber \\
    =& (k+1)^2 \left(\sum_{i,j,k}\hat{A}_{ij}\hat{A}_{kj}+\hat{R}_{|}\right) \left(\frac{1}{2N_k}+O(N_k^{-2})\right) \nonumber \\
    =& \Theta(N k^2 N_k^{-1}),
\end{align}
which leads to the same lower bound
\begin{equation}\label{eq: app lower bound general}
    \mathbb{E}[\Var(\partial_\mu \mathcal{L}_{QSGC})]_\text{lower} =
    \Theta(N^{-1}C^{-1}).
\end{equation}

\paragraph{Expected dependence.} To this end, by adding the factor given by the curse of dimensionality, $\left|\braket{\Vec{1}|D_{\Vec{\alpha}}|\Vec{1}}\right|^2\sim N_k^{-1}$, we obtain the same result as in Appendix \ref{app: Note on the curse of dimensionality},
\begin{align}
    \mathbb{E}[\Var(\partial_\mu \mathcal{L}_{QSGC})]_\text{expected} \sim& \mathbb{E}[\Var(\partial_\mu \mathcal{L}_{QSGC})]_\text{upper} / N_k + \mathbb{E}[\Var(\partial_\mu \mathcal{L}_{QSGC})]_\text{lower} \nonumber \\
    =& \Theta(C^{-2} + N^{-1} C^{-1}). \label{eq: app expected trend general}
\end{align}

\subsection{Results for $p$-th order QSGCs}\label{app: Results for QSGCs with K layers}

The results above provide a significant extension to the first order QSGC. 
Taking a second order QSGC, for example, is equivalent to updating all of our results with $\hat{A} \mapsto \hat{A}^2$. However, as proven in Proposition \ref{prop: app adj matr prod is adj matr}, multiplying two commuting normalized adjacency matrices with self-loops gives a new normalized adjacency matrix with self-loops. Therefore, because $\hat{A}$ commutes with itself, then $\hat{A}^2$ is a valid normalized adjacency matrix with self-loops of a new graph.\footnote{This works regardless of whether $\hat{A}$ has weighted edges or not. In the unweighted case, we just define trivial weights $w_{ij}=1$ on all edges of the graphs.} Thus, our results should hold by reassigning $\hat{A}^2\equiv \hat{B}$ for some new graph $\mathcal{G}_B$. Moreover, thanks to Corrolary \ref{cor: app commuting k-regular graphs}, if the initial graph is $k$-regular, then graph $\mathcal{G}_B$ is also $k$-regular.

This can be continued. As $\hat{A}$ and $\hat{B}=\hat{A}^2$ commute, multiplying them gives a new graph with normalized adjacency matrix with self-loops $\hat{C}=\hat{A}^3$, which is also $k$-regular. Continuing this process, we confirm that our results hold for any QSGC model with an arbitrary propagation order, $p$. Explicitly, for a general, $p$-th order QSGC, the upper bound in Eq.~\ref{eq: app upper bound general}, lower bound in Eq.~\ref{eq: app lower bound general}, and expected dependence for arbitrary ansätze in Eq.~\ref{eq: app expected trend general} hold.

\subsection{Results for QLGCs}\label{app: Results for QLGCs}

As mentioned in Section \ref{sec: Linear Graph Convolutional Network}, we replace the polynomial $P(L)$ when implementing the \textit{Quantum Linear Graph Convolution}  with a polynomial of $\hat{A}$, as the two polynomials' coefficients are equally trainable. In our variance calculations, unsurprisingly, we only have to replace
\begin{equation}
    \hat{A} \mapsto P(\hat{A}) = \sum_{q=1}^p  a_q \hat{A}^q.
\end{equation}
Here, in order to implement the polynomial on the address register $Reg(i)$, its spectral norm must be subunitary, and so we choose the range of coefficients\footnote{We write the coefficients as $a$ rather than $\alpha$ to avoid confusion with the $D$-bit counter $\Vec{\alpha}$.} $ a_q\in[0,1]$, with at least one being non-zero and their sum must be one, $\sum_q a_q=1$. This is essentially a linear superposition of $q$-th order QSGC models.

Employing Corollary \ref{cor: app app convex sum n terms}, the convex sum $P(\hat{A})$ of commuting matrices $\hat{A}^q$ is a valid normalized adjacency matrix with self-loops of some new graph. This will reduce most of the QLGC trainability calculation to the first order QSGC, as seen above. However, here we note that the coefficients $a_q$ are also trainable parameters.

\subsubsection{New gradient variance calculation}

Following the result in Eq.~\ref{eq: app explicit cost}, the new cost function is given by
\begin{align}
    \mathcal{L}_{QLGC} =& -\frac{1}{N}\sum_{i,j}P(\hat{A})_{ij}\Re\left(\bra{\bold{y}_j}U(\thetaB)\ket{\bold{x}_i}\right) \nonumber \\
    =& \sum_q a_q\left(-\frac{1}{N}\sum_{i,j}(\hat{A}^q)_{ij}\Re\left(\bra{\bold{y}_j}U(\thetaB)\ket{\bold{x}_i}\right)\right) \nonumber \\
    \equiv & \sum_q a_q \mathcal{L}_{q}.
\end{align}
Here we note $\mathcal{L}_q$ the cost function of an $q$-th order QSGC.

The new set of parameters is $\theta_\eta\in[-\pi,\pi)$, along with $ a_q\in[0,1)$, with the additional condition $\sum_q a_q=1$. We already know that the volume for the $\theta_\eta$ parameters is $(2\pi)^D$, where $D$ is the number of weights in the VQC. Let the volume $\int\dd \Vec{a}\equiv \mathcal{N}_a$.\footnote{This is the hypersurface section of constant sum $\sum_q a_q=1$ through the unit hypercube, which is $\mathcal{N}_a = \sqrt{p}/(p-1)!$, but for our purposes the exact value will not be relevant.}

Therefore, there are two gradients that need to be analyzed: $\partial_\mu \mathcal{L}_{QLGC}$ and $\partial_r \mathcal{L}_{QLGC}\equiv \partial \mathcal{L}_{QLGC} / \partial a_r$.

\paragraph{Gradient $\partial_\mu \mathcal{L}_{QLGC}$.} The cost gradient bias in this case is
\begin{equation}
    \mathbb{E}_{\Vec{a}}\mathbb{E}_{\thetaB}[\partial_\mu \mathcal{L}_{QLGC}] = \mathbb{E}_{\Vec{a}} \sum_q a_q \mathbb{E}_{\thetaB}[\partial_\mu \mathcal{L}_{q}].
\end{equation}
The terms in the sum are simply the biases of the QSGC models, which are zero. Therefore, the overall bias is also zero.

The variance is
\begin{equation}
    \Var_{\Vec{a};\thetaB}[\partial_\mu \mathcal{L}_{QLGC}] = \mathbb{E}_{\Vec{a}}\mathbb{E}_{\thetaB}[(\partial_\mu \mathcal{L}_{QLGC})^2] =
    \mathbb{E}_{\Vec{a}} \sum_{q,q'} a_q a_{q'} \Cov_{\thetaB}(\partial_\mu \mathcal{L}_{q},\partial_\mu \mathcal{L}_{q'}).
\end{equation}
We see that the terms 
\begin{equation}
    \Cov_{\thetaB}(\partial_\mu \mathcal{L}_{q},\partial_\mu \mathcal{L}_{q'}) =
    -\frac{1}{N^2}\sum_{i,j,k,l}(\hat{A}^q)_{ij}(\hat{A}^{q'})_{kl}\Re\left(\bra{\bold{y}_j}U(\thetaB)\ket{\bold{x}_i}\right)\Re\left(\bra{\bold{y}_l}U(\thetaB)\ket{\bold{x}_k}\right)
\end{equation}
are a slightly more generalized version of the variance in Eq.~\ref{eq: app recall variance}, where instead of counting terms $\hat{A}_{ij}\hat{A}_{kl}$, we use terms $\hat{B}_{ij}\hat{C}_{kl}$, where the commuting normalized adjacency matrices with self-loops $\hat{B}=\hat{A}^q$ and $\hat{C}=\hat{A}^{q'}$. By adapting Eqs.~\ref{eq: app R hat ||}, \ref{eq: app R hat land}, \ref{eq: app R hat lor}, and \ref{eq: app R hat |} with this change, we see that, in fact, the $k$-regular approximations are exactly the same. By rerunning the same set of arguments, unsurprisingly we can check that the terms $\Cov_{\thetaB}(\partial_\mu \mathcal{L}_{q},\partial_\mu \mathcal{L}_{q'})$ have the same dependence as any QSGC, which is solely on the ansatz and the system size. We call this dependence $\mathcal{D}(N,C)$ for simplicity.

Therefore, the overall variance simplifies to
\begin{equation}
    \Var_{\Vec{a};\thetaB}[\partial_\mu \mathcal{L}_{QLGC}] = \mathbb{E}_{\Vec{a}}\sum_{q,q'} a_q a_q'\mathcal{D}(N,C) = \mathbb{E}_{\Vec{a}} \underbrace{\left(\sum_q a_q\right)^2}_{1}\mathcal{D}(N,C) = \mathcal{D}(N,C).
\end{equation}

\paragraph{Gradient $\partial_r \mathcal{L}_{QLGC}$.}

The gradient in this case takes the form 
\begin{equation}
    \partial_r \mathcal{L}_{QLGC} = \partial_r\left(\sum_q a_q \mathcal{L}_q\right) = \mathcal{L}_q.
\end{equation}
This does not depend on any $a_q$ and so the bias and variance simply reduce to those of the $q$-th order QSGC cost.
\begin{align}
    \mathbb{E}_{\Vec{a}}\mathbb{E}_{\thetaB}[\partial_r \mathcal{L}_{QLGC}] =& \mathbb{E}_{\thetaB}[\mathcal{L}_q], \\
    \Var_{\Vec{a};\thetaB}[\partial_r \mathcal{L}_{QLGC}] =& \Var_{\thetaB}[\mathcal{L}_q].
\end{align}

How different is calculating the bias and variance of the QSGC cost function instead of its gradient? Because of our choice of sequential ansatz $U(\thetaB)=\prod_\eta(\exp(-i\theta_\eta V_\eta)\cdot W_\eta)$, it turns out the results are exactly the same. This is thanks to the \textit{paramaeter-shift rule}: because $V_\eta^2=\mathds{1}$, taking the derivative 
\begin{equation}
    \partial_\mu \exp(-i\theta_\mu V_\mu) = -i V_\mu \exp(-i\theta_\mu V_\mu) = \exp(-i(\theta_\mu+\pi/2) V_\mu).
\end{equation}
Therefore, by relabeling $\thetaB \mapsto \thetaB'$ where all components are the same except $\theta_\mu \mapsto \theta_\mu+\pi/2$, the overall gradient $\partial_\mu U(\thetaB) \equiv U(\thetaB')$. In all of the integrals from the variance derivation, the exact endpoints of the angle domains were irrelevant; it only matters that they integrate over a full sine period. Hence, shifting $\theta_\mu \mapsto \theta_\mu+\pi/2$ does not affect the results. 

To this end, the bias of $\partial_r \mathcal{L}_{QLGC}$ is also zero, and the variance is $\mathcal{D}(N,C).$ We conclude that all QSGCs and QLGCs are equally trainable, leading to the exact same gradient cost variance for a given system size and ansatz.

\section{Complexity and dequantization expanded}\label{app: Complexity and dequantization expanded}

In this section, we show how one can benefit from the space-time complexity tradeoff for all QGCN models, and how they can be dequantized for graphs giving rise to low rank adjacency matrices $\hat{A}$. We start with the QSGC architecture.

\subsection{Classical and Quantum SGC Complexities}

We begin this by importing the complexity results from the QSGC's original paper \cite{liao_graph_2024}. For the classical SGC (CSGC) model, the general complexity bounds are
\begin{align}
    S_{CSGC} =& O(|E|+NC+C^2), \\
    T_{CSGC} =& O((p|E|C+NC^2)\log(1/\varepsilon')),
\end{align}
where $|E|=Nk/2$ is the number of edges, $N$ that of nodes, $k$ the average degree, and $C$ the number of features per node.  $p$ is the propagation order, which translates to how many times $\hat{A}$ is applied. $\varepsilon'$ is a precision parameter and is considered fixed, so we will absorb it in the Big-$O$ notation as an overhead. We mention that the logarithmic dependence makes this overhead have a small impact to the practical complexity.

The QSGC model attains the following complexities in the number of qubits and circuit depth:
\begin{align}
    n_{QSGC} =& O(\log(NC)+n_{anc}+n_{anc}'), \\
    D_{QSGC} =& \tilde{O}\left(NC\log(1/\varepsilon_1) \frac{\log(n_{anc})}{n_{anc}} + pN\log(N) \cdot s\log(s)\log(1/\varepsilon_2) \frac{\log(n'_{anc})}{n'_{anc}}\right).
\end{align}
Here, $n_{anc}$ and $\varepsilon_1$ are respectively the required number of ancillary qubits and precision of the amplitude encoding of the input and target states. Similarly, $n'_{anc}$ and $\varepsilon_2$ are those required for the block encoding of the adjacency matrix $\hat{A}$. The propagation order $p$ is taken to be independent on the graph size so can be ignored. The algorithm implementing the sparse block encoding assumes that $\hat{A}$ has sparsity $s$. This is to say, every column in $\hat{A}$ has at most $s$ non-zero values. Because this includes the self-loops, that is equivalent to saying $s-1$ corresponds to the maximum degree in the graph. The Big-$\tilde{O}$ notation excludes the doubly logarithmic factors as they scale very slowly. We note that the numbers of ancillary qubits $n_{anc}$ and $n'_{anc}$ can be adjusted, and lie in the following ranges:
\begin{align}
    \Omega(\log(NC))\leq & n_{anc} \leq O(NC), \label{eq: app nanc bounds}\\
    \Omega(\log(N))\leq & n'_{anc}\leq O(N\log(N) \cdot s\log(s)). \label{eq: app nancprime bounds}
\end{align}
The lower bound of $n_{anc}$ indicates that the logarithmic term in the total qubit number $n_{QSGC}$ can be omitted. We use this in Eq.~\ref{eq: app space QSGC}. Finally, the precisions can be ignored as done in the classical case.

\paragraph{Gradient calculation.} Firstly, the \textit{parameter-shift rule} \cite{Mitarai_Negoro_Kitagawa_Fujii_2018, Banchi_Branford_Waghela_2025} states that in order to compute the cost function gradient for a gradient descent method, one needs to shift in turn each of the weights $\theta_\eta$ of the ansatz $U(\thetaB)$. This adds a contribution of $O(2D)$ queries, where $D$ is the number of weights.

\paragraph{Shot noise.}  As mentioned in \cite{liao_graph_2024}, is of order $\Omega(\log(1/\delta)/\varepsilon^2)$. Here, $\varepsilon$ is the additive error to the gradient, and the probability of failure $\delta$ is fixed. In order to resolve points on the loss landscape so that the gradient descent can be performed, this error needs to lie below the standard deviation $\sigma_{\partial_\mu L}\sim\Var_{\thetaB} [\partial_\mu \mathcal{L}_{QSGC}]$ \cite{Sweke_Wilde_Meyer_Schuld_Faehrmann_Meynard-Piganeau_Eisert_2020, Gu_Lowe_Dub_Coles_Arrasmith_2021}. If we choose an arbitrary ansatz, the expected variance is of order $\Theta(N^{-1}C^{-1}+C^{-2})$. Then, the number of shots must be at least
\begin{equation}
    \text{\#shots} = \Omega\left(\frac{1}{\varepsilon^2}\right) \geq \Omega\left(\frac{1}{\sigma^2}\right) = \Omega\left(\frac{1}{N^{-1}C^{-1}+C^{-2}}\right) = \Omega\left(\frac{C^2}{1+C/N}\right)=\Omega(C^2),
\end{equation}
where $C/N\ll 1$ by choice.

\paragraph{Input problem.} As mentioned in Section \ref{sec: Note on Complexity and Dequantization}, in order to get the full space and time complexities, one needs to add to the general circuit complexities the full classical compilation overhead. For the QSGC, this requires that all classical data of the input, target, and adjacency matrix $\hat{A}$ to be stored and processed classically in order to generate the transpiled circuit. This adds a term $O(|E|C+NC)$ to the space complexity anf $\Omega(|E|C+NC)$ to the time complexity:
\begin{align}
    S_{QSGC} =& O(|E|C+NC) + O(n_{QSGC}), \label{eq: app space QSGC}\\
    T_{QSGC} =& \Omega(|E|C+NC) + O(2D)O(C^2)O(D_{QSGC}). \label{eq: app time QSGC}
\end{align}
We see that the input problem compromises any significant quantum advantage in this case as the first terms in both of Eqs.~\ref{eq: app space QSGC} and \ref{eq: app time QSGC} are comparable to the classical SGC space and time complexities. There could be significant differences in the overheads, but in general the problem as it stands is not scalable. Hence, we must assume an efficient oracle or QRAM to bypass the input problem for now. Hereafter, the classical compilation complexities are ignored.\\

\subsection{Complexity analysis for fixed $C$}

For our purposes, it is enough to consider $C$ as well as all precisions and failure probabilities fixed as we want to vary the graph size $N$ and connectivity $k$. The depth $D$ is also strictly dependent on the data register and should vary only with $C$. This simplifies the complexities to
\begin{align}
    S_{CSGC} =& O(N(k+C)) = O(Nk), \label{eq: app space CSGC}\\
    T_{CSGC} =& O(N(k+C)) = O(Nk), \label{eq: app time CSGC}\\
    S_{QSGC} =& O(n_{anc}+n_{anc}'), \label{eq: app new space QSGC}\\
    T_{QSGC} =& \tilde{O}\left(N\frac{\log(n_{anc})}{n_{anc}} + N\log(N) \cdot s\log(s)\frac{\log(n'_{anc})}{n'_{anc}}\right). \label{eq: app new time QSGC} 
\end{align}
In Eqs.~\ref{eq: app space CSGC} and \ref{eq: app time CSGC}, $C=O(1)$ in $N$ and $k=\Omega(1)$ so term $Nk$ dominates. Maintaining the assumption of (almost) regular graphs, we can write $k=\Theta(s)$ given that $s-1$ is the maximum degree of any node, and $k$ is the average degree. We now show that exponential advantage can be achieved for either space or time complexities, but not both, by classifying them over the choices of $n_{anc}$ and $n_{anc}'$. We split the analysis in three cases: a balanced polynomial advantage in space and time complexities, an exponential speed-up with maximized space reduction, and an exponential reduction with maximized speed-up. As we will see below, the latter case is a bit trickier and the choice of ancillas depends on the sparsity of the adjacency matrix.

\subsubsection{Polynomial speed-up and space reduction}

For this task we choose a sub-exponential number of ancillas, $n_{anc}=O(N^\alpha)$ and $n_{anc}'=O(N^{\alpha'})$, with $\alpha, \alpha'\in(0,1)$. This obeys the bounds in Eq.~\ref{eq: app nanc bounds} and Eq.~\ref{eq: app nancprime bounds}. Then, the updated space and time complexities in Eq.~\ref{eq: app new space QSGC} and Eq.~\ref{eq: app new time QSGC} are
\begin{align}
    S_{QSGC} =& O(N^\alpha), \\
    T_{QSGC} =& \tilde{O}\left(N\frac{\alpha\log(N)}{N^\alpha} + N\log(N)\cdot s\log(s)\frac{\alpha'\log(N)}{N^{\alpha'}}\right) \nonumber \\
    =& \tilde{O}\left(N^{1-\alpha}\log(N) + N^{1-\alpha'}\log^2(N)\cdot s\log(s)\right)
\end{align}

The sparsity $s\in[1,N]$ is an integer, so $\Omega(1)\leq s \leq O(N)$. Then, for each limit, we have the following cases for the space reduction:
\begin{equation}
    \begin{cases}
        s=\Theta(1) \qimpl 
        S_{CSGC}=O(N) \qimpl S_{QSGC} = O(S_{CSGC}^{\max(\alpha,\alpha')}), \\ \\
        s=\Theta(N) \qimpl 
        S_{CSGC}=O(N^2) \qimpl S_{QSGC} = O(S_{CSGC}^{\alpha'/2}),
    \end{cases}
\end{equation}
both leading to a polynomial advantage. Strictly speaking, the space reduction is sublinear.

For the time speed-up:
\begin{equation}
    \begin{cases}
        s=\Theta(1) \qimpl 
        \begin{cases}
            T_{CSGC} = O(N), \\ \\
            T_{QSGC} = \tilde{O}(N^{1-\max(\alpha,\alpha')}\log^a(N)) \qimpl \sigma_{q/cl} = \tilde{\Omega}(N^{\max(\alpha,\alpha')} / \polylog(N)),
        \end{cases} \\ \\
        s=\Theta(N) \qimpl 
        \begin{cases}
            T_{CSGC} = O(N^2), \\ \\
            T_{QSGC} = \tilde{O}(N^{2-\alpha'}\log^3(N)) \qimpl \sigma_{q/cl} = \tilde{\Omega}(N^{\alpha'/2} / \polylog(N)).
        \end{cases}
    \end{cases}
\end{equation}
In the second line, $a$ is 1 if $\alpha>\alpha'$ and 2 otherwise. 

Therefore, in both cases (and thus for any intermediate sparsity), the advantage is \textit{polynomial}. In terms of the system size $N$, the speed-up is also sublinear. This is the optimal balance as trying to obtain better scaling for one type of complexity will negatively affect the other one. The choice of $\alpha$ depends on the type of hardware onto which this model is implemented.

\subsubsection{Exponential space reduction}

Here we focus on obtaining an exponential space reduction with optimal time speed-up. For that reason, $n_{anc}$ and $n_{anc}'$ must be of order $O(\polylog(N))$. By explicitly defining the leading orders as $n_{anc}=O(\log^l(N))$ and $n_{anc}'=O(\log^{l'}(N))$, Eq.~\ref{eq: app new time QSGC} becomes
\begin{equation}
    T_{QSGC}=\tilde{O}\left(N/\log^l(N)+N\cdot s\log(s)/\log^{l'-1}(N)\right).
\end{equation}
For $s=\Theta(1)$, the terms are comparable up to a polylogarithm, and for $\min(l, l'-1)\geq 1$, the speed-up is \textit{polylogarithmic}. In fact, thanks to the lower bound of $n_{anc}$, the condition is simply $l'\geq 2$. For $s=\Theta(N)$, the second term dominates, and given that $T_{CSGC}=Ns$, the speed-up is polylogarithmic for $l'\geq3$. 

\subsubsection{Exponential speed-up}
In this case, in order for the runtime to gain an exponential speed-up, each of the two terms in Eq.~\ref{eq: app new time QSGC} must be polylogarithmic. For that to happen, we must choose $n_{anc}=O(N/\polylog(N))$ and $n_{anc}'=O(Ns/\polylog(N))$. By setting the dominant powers as $m$ and $m'$ respectively, Eq.~\ref{eq: app new time QSGC} becomes
\begin{align}
    T_{QSGC} =& \tilde{O}\left(N\frac{\log(N/\log^m(N))}{N/\log^m(N)} + N\log(N) \cdot s\log(s)\frac{\log(Ns/\log^{m'}(N))}{Ns/\log^{m'}(N)} \right) \\
    =& \tilde{O}(\log^{m+1}(N)+\log^{m'+2}(N)\log(s)) \\
    =& \tilde{O}(\polylog(N)).
\end{align}
The space reduction is then 
\begin{align}
    S_{QSGC} = O\left(N/\log^m(N) + Ns/\log^{m'}(N)\right).
\end{align}
For $s=\Theta(1)$, the terms are comparable and for $\min(m,m')\geq 1$ the space reduction is \textit{polylogarithmic}. For $s=\Theta(N)$, the second term dominates and regardless of $m\in\mathbb{Z}$, the condition reduces to $m'\geq 1$.

\subsection{A more accurate $C$-dependence}\label{app: A more accurate $C$-dependence}

In practice, for an ever-increasing graph size, having a limited number of features leads to many nodes becoming indistinguishable, having the same features. If one wants to be able to distinguish between all nodes if necessary, $C$ should increase at least logarithmically with the number of nodes. If one also desires additional expressivity from the embedding process, an even better choice that is still trainable is $C=O(\polylog(N))\equiv O(\log^c(N))$, where $c$ is the leading order. How does this change our results?

In terms of trainability, the barren plateau is still avoided for most ansätze, as
\begin{equation}
    \mathbb{E}_\text{i,t}[\Var_{\thetaB}(\partial_\mu \mathcal{L}_{QSGC})] \sim \Theta(N^{-1}\log^{-c}(N)+\log^{-2c}(N)) = \Theta(\log^{-2c}(N)).
\end{equation}
Furthermore, the number $D$ of weights $\theta_\eta$ now can depend on $N$. In order not to overfit, $D\ll O(C)$, so in the Big-$\tilde{O}$ notation, it can be dismissed. \\

With this form of $C$, the new complexities for the classical SGC are the following:
\begin{align}
    S_{CSGC} =& O(Ns+N\log^c(N)), \\
    T_{CSGC} =& O(Ns\log^c(N) + N\log^{2c}(N)).
\end{align}

Moving to the QSGC, the space and time complexities from Eq.~\ref{eq: app space QSGC} and Eq.~\ref{eq: app time QSGC} become
\begin{align}
    S_{QSGC} =& \tilde{O}(n_{anc}+n'_{anc}), \\
    T_{QSGC} =& O(\log^{2c}(N))\tilde{O}\left( N\log^{c}(N) \frac{\log(n_{anc})}{n_{anc}} + N\log(N) \cdot s\log(s)\frac{\log(n'_{anc})}{n'_{anc}}\right) \nonumber \\
    =& \tilde{O}\left( N\log^{3c}(N) \frac{\log(n_{anc})}{n_{anc}} + N\log^{c+1}(N) \cdot s\log(s)\frac{\log(n'_{anc})}{n'_{anc}}\right).
\end{align}
Similarly, the bounds on the numbers of ancillary qubits change as
\begin{align}
    \Omega(\log(N))\leq & n_{anc} \leq O(N\log^c(N)), \label{eq: app nanc bounds Cacc}\\
    \Omega(\log(N))\leq & n'_{anc}\leq O(N\log(N) \cdot s\log(s)). \label{eq: app nancprime bounds Cacc}
\end{align}
We observe that, in fact, the bounds are broadened by this dependence.

\subsubsection{Advantage regimes}\label{app: Advantage regimes}

The regimes described in the previous section are preserved as we let $C$ depend on $N$; the only difference is that the minimum dominant powers $l$, $l'$, $m$, and $m'$ generally need to make up for the extra powers $c$ in the QSGC model, but benefit from the extra polylogarithmic burden that the CSGC acquires.

\paragraph{Polynomial speed-up and space reduction.}

Because polynomial powers do not interact with logarithmic ones, this regime is completely unchanged. That is to say, $\alpha,\alpha'\in(0,1)$ is still valid.

\paragraph{Exponential space reduction.} In this case, $n_{anc}$ needs to make up for the factor of $\log^{3c}(N)$ and $n_{anc'}$ for $\log^{c+1}(N)$. Therefore, the following updates are required:
\begin{align}
    s =& \Theta(1) \qimpl \min(l-c,l'+c-1)\geq 1; \\
    s =& \Theta(N) \qimpl l'\geq 3, \quad\text{unchanged.}
\end{align}

\paragraph{Exponential speed-up.} Similar to the case above, the only update is on the conditions of the dominant logarithmic powers:
\begin{align}
    s =& \Theta(1) \qimpl \min(m,m')\geq c+1; \\
    s =& \Theta(N) \qimpl m'\geq 1, \quad \text{unchanged.}
\end{align}

\subsection{Complexity of the QLGC and advantage}

As given in \cite{liao_graph_2024}, the classical and quantum Linear Graph Convolutional models take the following forms:
\begin{align}
    S_{CLGC} =& O(|E| + pN C + C^2), \label{eq: app space CLGC} \\
    T_{CLGC} =& O(p|E|C + N C^2), \label{eq: app time CLGC} \\
    n_{QLGC} =& O(\log(NC)+n_{anc}+n'_{anc}) \label{eq: app qubits QLGC} \\
    D_{QLGC} =& \tilde{O}\left(N C \log\left(\frac{1}{\varepsilon_1}\right) \frac{\log(n_{anc})}{n_{anc}} + pN \log N \cdot s \log s 
    \log\left(\frac{1}{\varepsilon_2}\right) \frac{\log(n'_{anc})}{n'_{anc}} + pn'_{anc} \right). \label{eq: app depth QLGC}
\end{align}
Here, the propagation order $p$ takes the role of the dominant power in $P(\hat{A})$. In addition, there are now $p$ more weights encoded into rotation gates via QSVT, so the backpropagation complexity (parameter-shift rule) extends to $O(D+p)$. As in the QSGC case, we take $p$ to be fixed and the depth $D$ to be small compared to $C$. Therefore, this factor is neglected using the Big-$\tilde{O}$ notation. As before, we also consider the precision to factor in a small, fixed overhead.

Having assumed a fixed $p$, then the four LGC complexities above are identical to the SGC ones, with a single exception: the QSVT contribution in the depth, $\tilde{O}(pn'_{anc})$. Because this itself undergoes a trade-off with the term $\propto\frac{\log(n'_{anc})}{n'_{anc}}$, regardless of our choice of ancilla numbers, there is a minimum possible depth that halts the exponential speed-up advantage.

\paragraph{Optimal QLGC time complexity.} To check the trade-off in $D_{QLGC}$, we define a function
\begin{align}
    f(x) = a\frac{\log x }{x}+x,
\end{align}
In our case, $x=n'_{anc}$ and $a=N\log N \cdot s \log s$. Its minimum is given by $f'(x)=0$, which implies that
\begin{equation}
    0 = 1+a\frac{1-\log x}{x^2} \qimpl x = \sqrt{a(\log x -1)}.
\end{equation}
By assuming that $x_0=\sqrt{a}$, we can solve this iteratively:
\begin{align}
    x_1 &=\sqrt{a(\log x_0-1)} = \sqrt{a\log a} \nonumber \\
    x_2 &=\sqrt{a(\log(a\log a))-1} =\sqrt{a\log a + a\log(\log a)} = \tilde{O}\left(\sqrt{a\log a}\right).
\end{align}
The answer quickly converges, and the minimum function satisfying this is when the two terms are equal, $f(x) = \tilde{O}\left(\sqrt{a\log a}\right)$. In this case, the QLGC depth is minimised by:
\begin{align}
    (n'_{anc})_\text{m} &= \tilde{O}\left(\sqrt{Ns\log N \log s \log (Ns)}\right) = \tilde{O}\left(\sqrt{Ns\log^2 N \log s}\right), \\
    \left(D_{QLGC}\right)_{\text{min}} &= \tilde{O}\left(
    NC\frac{\log n_{anc}}{n_{anc}} + \sqrt{Ns \log^2 N \log s}
    \right).
\end{align}
Even for $n_{anc}=O(N/\polylog N)$, this depth can reach $\tilde{O}(\sqrt{N\log^2 N})$ for sparse graphs, and $\tilde{O}(N\log^{3/2}N)$ for dense ones, which means the exponential speed-up is now unachievable, with at best a polynomial one. As this case becomes weak, the remaining advantages are: (1) exponential space suppression, and (2) mixed polynomial-polynomial space-time advantage.


\subsection{Low-rank simulability}\label{app: Low-rank simulability}

An efficient oracle is a severe assumption in general; to even out the comparison, we can strictly ignore the data loading process for both the classical and the quantum models, and strictly compare the architectures themselves. In this case, Tang's algorithm \cite{Tang_2019} that approximates close-to-low-rank matrices with an $l^2$-norm sampling algorithm becomes the new reference for the classcal GCN models. Namely, the algorithm creates a distribution that can reconstruct a rank $r$ approximation of some $(M\times N)$ matrix $A$. The approximation precision $\varepsilon$ is defined such that $\|A_\text{approx}-A\|_F\leq\varepsilon\|A\|_F$ bounds the Frobenius norm of the error between the approximation and the exact matrix. Note that $r\leq\min(M,N)$.

In order to avoid the $O(M\cdot N)$ complexity, one can never access the full matrix directly, and for this reason the algorithm requires the construction of a \textit{binary search tree data structure} that groups elements together and only require $O(\log^2(MN))$ space and access time complexities per query. After creating these "efficient classical oracles", the algorithm itself can recreate a rank $r$ approximation of the matrix efficiently. Overall, from the information in \cite{Tang_2019}, we check that the space and time complexities for this algorithm are
\begin{align}
    S_\text{init} &= S_\text{compile} + S_\text{sampling} \nonumber \\
    &= O(w\log^2(MN))+O\left(\max\left\{\frac{r^{11}}{\varepsilon^{34}},\frac{r^{8}}{\varepsilon^{32}}\right\}\right), \label{eq: app s init} \\
    T_\text{init} &= T_\text{compile} + T_\text{sampling} \nonumber \\
    &= O(w\log^2(MN)) + \tilde{O}\left(\log^2(MN)\cdot \max\left\{\frac{r^{16.5}}{\varepsilon^{51}},\frac{r^{12}}{\varepsilon^{48}}\right\}\right). \label{eq: app t init}
\end{align}
Here, $w$ is the number of nonzero entries in matrix $A$. To maintain a fair comparison with the QGCN, the compilation is omitted and the data structure is assumed to exist already. For our purposes, we consider $\varepsilon$ to be a fixed precision and so will be omitted.\footnote{We mention that in reality, this can give an immense overhead due to its large powers.} Therefore, we will consider the following relevant complexities: 
\begin{align}
    S_\text{sampling}
    &= O\left(r^{11}\right), \label{eq: app s sample} \\
    T_\text{sampling}
    &= \tilde{O}\left( \log^2(MN)\cdot r^{16.5}\right). \label{eq: app t sample}
\end{align}
Once sampled, the generated distribution allows one to access the approximated matrix $A_\text{approx}$. We will apply this to our adjacency matrix $\hat{A}$, which we want to approximate to a rank $r$ matrix.

\paragraph{Sampling process (informal).} 
Assume we can efficiently construct the tree-data structures for matrices $\hat{A}$, $X$, and $Y$. We need to approximate the product $\sum_{i,j=1}^N(Y^T)_{ai}P(\hat{A})_{ij}X_{jb}$, for some $a,b\in[C]$, in order to find the QGCN cost function.\footnote{Here, $P(\cdot)$ is a $p$-th order polynomial of a QLGC (or simply the power function of a $p$-th order QSGC).} For this, we sample columns $j_1,j_2,\dots,j_n$ for $n \ll N$. Call the distribution $\mathcal{D}_{\hat{A}}$.\footnote{This sampling takes the form $\mathcal{D}_{\hat{A}}(j)=\|\hat{A}_{*j}\|^2/\|\hat{A}\|_F^2$, so peaks for larger columns' modulus. We do not need to worry about this form in our brief sketch.} By stacking these columns and normalizing each by $1/\sqrt{n\cdot\mathcal{D}_{\hat{A}}(j_u)}$ for $u\in[n]$, we form a new matrix
\begin{equation}
    S_{iu} = \frac{\hat{A}_{ij_u}}{\sqrt{n\cdot\mathcal{D}_{\hat{A}}(i_u)}}.
\end{equation}
We can do a sample for the rows of the already trimmed adjacency matrix $S$ by sampling $n$ rows from it\footnote{We use the initial sampling $\mathcal{D}_{\hat{A}}$, which we can easily access using the data structure for $\hat{A}$.}, and obtain a matrix
\begin{equation}
    W_{vu}=\frac{\hat{A}_{i_vj_u}}{\sqrt{n^2\mathcal{D}_{\hat{A}}(j_u)\mathcal{D}_{\hat{A}}(i_v)}}.
\end{equation}
This is a trimmed and normalized version of the initial adjacency matrix. Furthermore, we compute its Singular Value Decomposition, we obtain a spectral form 
\begin{equation}
    W = U_W\cdot \Sigma_W \cdot V_W^T,
\end{equation}
where $\Sigma_W$ is truncated to size $(r\times r)$. The following output is stored\footnote{Storing the isometries and singular matrix is not strictly required but it is useful as they will be called repeatedly.}:
\begin{equation}
    \biggl\{ (j_1,j_2,\dots,j_n);(i_1,i_2,\dots,i_n); W= U_W\cdot \Sigma_W \cdot V_W^T\biggr\}.
\end{equation}
This will be considered a point in our sample which is connected to the distribution $\mathcal{D}_{\hat{A}}$. We note that this is a simplified sketch of the sampling process, which takes time $T_\text{sampling}$ to produce, requiring space $S_\text{sampling}$.

\paragraph{Dequantizing a first order QSGC cost.} For this model, the cost can be easily dequantized. As we sampled rows and columns from $\hat{A}$, we need to also trim the input and target matrices accordingly: we sample 
\begin{equation}
    X'_{ua}=X_{j_ua}/\sqrt{n\mathcal{D}_{\hat{A}}(j_u)} \quad \text{and} \quad Y'_{av}=Y_{ai_v}/\sqrt{n\mathcal{D}_{\hat{A}}(i_v)}.
\end{equation}
Then, we compute the following estimator:
\begin{align}
    \mathbb{E}_{\{j\}_u,\{i\}_v} [(Y'\cdot W\cdot X')_{ab}] &= \sum_{u,v=1}^n \mathbb{E}_{\{j\}_u,\{i\}_v} [Y'_{av}\cdot W_{vu}\cdot X'_{ub}] \nonumber \\
    &= \sum_{u,v=1}^n\sum_{j_u,i_v=1}^N \mathcal{D}_{\hat{A}}(j_u) \mathcal{D}_{\hat{A}}(i_v) \frac{Y^T_{ai_v}\hat{A}_{i_vj_u}X_{j_ub}}{n^2\mathcal{D}_{\hat{A}}(j_u) \mathcal{D}_{\hat{A}}(i_v)} \nonumber \\ 
    &= \frac{1}{n^2}\sum_{u,v=1}^n(Y^T\cdot \hat{A}\cdot X)_{ab} \nonumber \\
    &= (Y^T\cdot \hat{A}\cdot X)_{ab}.
\end{align}
Because this distribution benefits from a sharp peak around the approximated matrix, we can instead write this as an approximation of the form 
\begin{equation}
    Y^T\cdot \hat{A}\cdot X \approx Y'\cdot W\cdot X'.
\end{equation}
Therefore, the cost is approximated to
\begin{equation}
    \mathcal{L}_{QSGC} \approx \Tr(Y'\cdot W\cdot X'\cdot U(\theta)).
\end{equation}
Calling the sampled values using the data structure takes $O(n\log^2N)$ runtime. The matrix contraction and trace calculation add $O(nCr+rC^2+r^2)$ to the runtime.

\paragraph{Dequantizing all QGCNs.} How do we advance the previous calculation to one for a generic polynomial $P(\hat{A})$? Because $\hat{A}$ is symmetric, we can diagonalize it and apply the polynomial to each eigenvalue. In order to access the eigenvalues, however, we need to project it onto a tractable space, much smaller than the $N$-dimensional space it normally acts on. We note that by spectrally decomposing matrix $S\equiv U_S\cdot \Sigma_S \cdot V_S$ (defined above as the collection of $n$ columns of $\hat{A}$), we can retrieve an approximation of the $n$ most significant eigenvectors of the adjacency matrix. Namely, the following relation is approximately satisfied, 
\begin{equation}
    \hat{A} \cdot \hat{A}^T \approx S\cdot S^T \qimpl \hat{A}^2 \approx U_S\cdot \Sigma_S^2\cdot U_S^T.
\end{equation}
This is a valid spectral decomposition for $\hat{A}^2$, so its eigenvectors are given by $U_S$. Furthermore, we can write out 
\begin{equation}
    U_S = S\cdot V_S \cdot \Sigma_S^{-1} \approx S\cdot \underbrace{V_W\cdot \Sigma_W^{-1}}_{\equiv P_W}.
\end{equation}
The Then, the optimal subspace projection (preserving most information) is 
\begin{equation}
    M := U_S^T\cdot \hat{A} \cdot U_S \approx P_W^T \cdot S^T \cdot \hat{A} \cdot S \cdot P_W.
\end{equation}
Next, by diagonalizing $M\equiv Z\cdot \Lambda \cdot Z$, we can apply the polynomial on its eigenvalues, thus reaching the approximation:
\begin{equation}
    P(\Lambda) \approx Z^T\cdot P_W^T \cdot S^T\cdot P(\hat{A}) \cdot S\cdot P_W\cdot Z.
\end{equation}
We note that this is normally a heavy approximation as applying a polynomial is a highly nonlinear operation, which is then approximated by passing it through a (linear) estimator. Because of the low variance of the chosen distribution, this is a reasonable approximation.

Returning to the matrix product $Y^T\cdot P(\hat{A})\cdot X$, we estimate this as
\begin{equation}
    Y^T\cdot P(\hat{A})\cdot X \approx Y^T\cdot S \cdot P_W \cdot Z \cdot P(\Lambda) \cdot Z^T \cdot P_W^T \cdot S^T \cdot X.
\end{equation}
Inside this approximation, we further take 
\begin{equation}
    Y^T\cdot S = \mathbb{E}_{\{i\}_v}[Y'^T\cdot W], \quad \text{and} \quad S^T\cdot X = \mathbb{E}_{\{j\}_u}[W^T\cdot X'],
\end{equation}
which simplifies the approximation to
\begin{equation}
    Y^T\cdot P(\hat{A})\cdot X \approx Y'^T\cdot U_W \cdot Z \cdot P(\Lambda) \cdot Z^T \cdot U_W^T \cdot X'
\end{equation}
and thus the cost function. For this, the runtime necessary for extracting the distribution points is $O(n\log^2N)$, that for the diagonalization of $M$ is $O(r^3)$. Finding the polynomial (of fixed order $p$) of $\Lambda$ using Horner's method takes another $O(pr)=O(r)$. Finally, the full matrix contraction can be calculated in $O(nCr+rC^2+Cr^2)$. In terms of space complexity, the $W$ are stored using their spectral decomposition in $O(nr+r^2)$ space, and the trainable weights in $O(C^2)$ space, apart from the efficient data structures for $Y$, $X$, and $\hat{A}$. \\

\paragraph{Complexities of the dequantized algorithm.} Considering that $n=O(r^{5.5})$ (given in \cite{Tang_2019}), the full complexities of the dequantized QGCN are:
\begin{align}
    S_{QGCN}^\text{Deq} &= S_\text{sampling} + S_\text{compute} = O(r^{11}+nr+C^2) \nonumber \\ &= O(r^{11}+\log^{2c}(N)), \label{eq: app s deq} \\
    T_{QGCN}^\text{Deq} &= T_\text{sampling} + T_\text{compute} = O(\log^2(N)\cdot r^{16.5}+nCr+rC^2+r^2C) \nonumber \\
    &=O(\log^2(N)\cdot r^{16.5}+\log^c(N)r^{6.5}). \label{eq: app t deq}
\end{align}


\paragraph{Sublinear rank.} In order for this algorithm to gain an advantage over the standard classical GCN model, $r\ll O(N)$ is required. Thus, we choose a form $r=O(N^\beta)$, where $\beta\in[0,1)$. For the comparison with the quantum GCN models, we remind the complexities (derived from Eq.~\ref{eq: app qubits QLGC} and Eq.~\ref{eq: app depth QLGC}) for $n_{anc}=O(N^\alpha)$ and $n'_{anc}=O(N^{\alpha'})$:
\begin{align}
    S_{QGCN} &= \tilde{O}\left(N^{\max(\alpha,\alpha')}\right), \\
    T_{QGCN} &= \tilde{O}\left(N^{1-\alpha}\log^{3c+1}(N) + N^{1-\alpha'}\log^{c+2}(N) \cdot s\log(s) + \kappa N^{\alpha'}\right).
\end{align}
Here, $\kappa$ is 0 for QSGCs and 1 for QLGCs.

In Lemma \ref{lemma: rank vs sparsity}, we show that for a normalized adjacency matrix with self-loops $\hat{A}$ with sparsity $s$, the rank is lower-bounded by $r\geq N/s$. Therefore, we can impose a lower-bound on the sparsity, $s=\Omega(N/r)$. The best case scenario is $s=\Theta(N^{1-\beta})$.

The advantage of the dequantized model w.r.t.~the quantum one requires the following:
\begin{align}
    \text{space advantage:} \quad &\max(\alpha,\alpha')>11\beta, \label{eq: app space deq condition}\\ 
    \text{time advantage:} \quad &\max(1-\alpha,2-\alpha'-\beta,\alpha') > 16.5\beta. \label{eq: app time deq condition}
\end{align}
Given that  $\alpha$ and $\alpha'$ lie in the interval $(0,1)$, then from the conditions above, the maximum rank where both speed-up and space suppression are achieved against the QGCN models is $\beta_\text{max}=1/11 \approx 0.09$. This is in fact also valid if one only seeks space suppression alone, as it comes from Eq.~\ref{eq: app space deq condition}. Alternatively, seeking speed-up alone, the rank has an exponent $\beta'_\text{max}=2/17.5 \approx 0.11$.

\paragraph{Polylogarithmic rank.} In the case where $r=O(\log^\rho(N))$ for some integer $\rho$, the new classical algorithm gains both exponential space reduction and speed-up compared to the standard CSGC. Compared to the QSGC which has a necessary trade-off, this is net superior, gaining an overall exponential advantage. \\

One final note is that because of the way the adjacency matrix $\hat{A}$ is approximated, the rank $r$ is a measure of the approximated matrix, so the rank inequalities can be further loosened to the stable rank \cite{tropp2008columnsubsetselectionmatrix, Cohen_Nelson_Woodruff_2016}, such that
\begin{equation}
    r \mapsto r_\text{stable}= \frac{\|S\|_F^2}{\|S\|_2^2} = \frac{\sum_i\sigma_i^2}{\sigma_1^2}.
\end{equation}
Here, the ratio of the Frobenius norm over the spectral norm simply checks whether the singular value spectrum of the matrix is dominated by its first (few) values or if it is flat. The Frobenius norm can be calculated efficiently directly from the tree data structure of $\hat{A}$. The first singular value can then be approximated by making a preliminary low-rank sampling ($r=O(1)$). Then, the rank we use will be $r\equiv \lfloor r_\text{stable}\rfloor$.

To this end, we conclude that the QSGC architecture is classically simulable for any $r=O(N^{0.11})$ and loses the advantage entirely for $r=O(\polylog(N))$.

\section{Helper lemmas}\label{app: Helper lemmas}

In this section, we define and prove a few lemmas that will be relevant for the detailed gradient analysis given in Appendix \ref{app:Detailed trainability calculations}.

\subsection{Sums involving Hamming weights}

First, we find the explicit forms for the summations labelled as $H_1$, $H_2$, and $H'_1$ in Section \ref{app: Dependence of the variance with the number of features}, by defining the following lemma:

\begin{lemma} \label{lemma: app H numbers}
    The following summations over all nonzero bit-strings of length $N_k$ hold:
    \begin{align}
        H_1 & := \sum_{\Vec{\gamma} \neq 0} \frac{\gamma_b}
        {w(\Vec{\gamma})} =
        \frac{2^{N_k}-1}{N_k},
        \label{eq: app H1 lemma}\\
        H_2 & := \sum_{\Vec{\gamma} \neq 0} \frac{\gamma_b \gamma_c}
        {w(\Vec{\gamma})} =
        \begin{cases}
            \frac{(N_k-2)2^{N_k-1}+1}{N_k(N_k-1)} \equiv H'_2 \quad \text{if} \quad b\neq c, \\ \\
            H_1 \quad \text{otherwise,}
        \end{cases} \label{eq:app H2 lemma}\\
        H'_1 & := \sum_{\Vec{\gamma} \neq 0} \frac{\gamma_b}
        {\sqrt{w(\Vec{\gamma})}} = 2^{N_k-1}
        \left[\sqrt{\frac{2}{N_k}}\left(1-\frac{1}{8N_k}\right) + O(N_k^{-5/2})\right]. \label{eq: app H1' lemma}
    \end{align}
    where $w(\Vec{\gamma})$ denotes the Hamming weight of a nonzero bit string $\Vec{\gamma}$ of length $N_k\geq 3$ and $b$ and $c$ label the bits inside $\Vec{\gamma}$.
\end{lemma}

\begin{proof}
    For $H_1$, all strings where $\gamma_b=0$ give no contribution, so the sum reduces to 
    \begin{equation}
        H_1 = \sum_{\Vec{\gamma}'}\frac{1}{w(\Vec{\gamma}')+1},
    \end{equation}
    where $\Vec{\gamma}'$ is an $(N_k-1)$-bit string equal to $\Vec{\gamma}$ with $\gamma_b$ removed. For this sum, we notice that there are $\binom{N_k-1}{k}$ strings that have a Hamming weight $w(\Vec{\gamma})\equiv k$. Therefore, the sum can be further simplified to
    \begin{equation}\label{eq: app H1 proof}
        H_1 = \sum_{k=0}^{N_k-1}\frac{1}{k+1}\binom{N_k-1}{k}
        = \frac{1}{N_k}\sum_{k=0}^{N_k-1}\binom{N_k}{k-1}
        =\frac{2^{N_k}-1}{N_k}.
    \end{equation}
    The second simplification comes from rearranging terms in the binomial formula $\binom{N_k-1}{k}=\frac{k+1}{N_k}\binom{N_k}{k-1}$, and the final result is derived using $\sum_{k=0}^{N_k} \binom{N_k}{k}=2^{N_k}$.

    For $H_2$, there are two options. If $b=c$, then in Eq.~\ref{eq:app H2 lemma}, $\gamma_b\gamma_c = \gamma_b^2 = \gamma_b$, so $H_2=H_1$. If $b\neq c$, we perform a similar calculation as done in Eq.~\ref{eq: app H1 proof}, after eliminating both bits $\gamma_b$ and $\gamma_c$ from the string $\Vec{\gamma}$, resulting into a new string $\Vec{\gamma}''$ of length $N_k-2$:
    \begin{equation}
        H_2 = \sum_{\Vec{\gamma}''}\frac{1}{w(\Vec{\gamma}'')+2} = \sum_{k=0}^{N_k-2}\frac{1}{k+2}\binom{N_k-2}{k}.
    \end{equation}
    Next, we write the terms as
    \begin{equation}
        \frac{1}{k+2}\binom{N_k-2}{k} = \frac{1}{N_k(N_k-1)}
        \left[ N_k\binom{N_k-1}{k+1}-\binom{N_k}{k+2}\right],
    \end{equation}
    which can be verified by writing out the binomial coefficients. Finally, by adding and subtracting terms to complete the summations $\sum_{k=0}^{n}\binom{n}{k}=2^{n}$ and for $n=N_k-1$ and $N_k$ respectively, we obtain
    \begin{equation}
        H_2 = \frac{1}{N_k(N_k-1)}
        \left[N_k(2^{N_k-1}-1)-(2^{N_k}-N_k-1)\right] = \frac{(N_k-2)2^{N_k-1}+1}{N_k(N_k-1)}.
    \end{equation}

    The final expression, Eq.~\ref{eq: app H1' lemma}, can be approximated using a Taylor expansion. We start with the same simplification as done in Eq.~\ref{eq: app H1 proof}:
    \begin{equation}
        H'_1 = \sum_{k=0}^{N_k-1}\frac{1}{\sqrt{k+1}}\binom{N_k-1}{k} \equiv 2^{N_k-1} \mathbb{E}_k^{\textbf{Bi}}\left[\frac{1}{k+1}\right],
    \end{equation}
    where we have normalized by the exponential in order to reach a binomial distribution $\textbf{Bi}(N_k-1, 1/2)$ --- hence the summation becomes an expected value.

    For this binomial distribution, the mean is $\mu=\frac{N_k-1}{2}$, and the variance is $\sigma^2=\frac{N_k-1}{4}$. Thus, an expansion around the mean of any function $f(k)$,
    \begin{equation}
        f(k) = f(\mu) + f'(\mu)(k-\mu)+\frac{1}{2}f''(\mu)(k-\mu)^2+ \xi
    \end{equation}
    has an expected value
    \begin{equation}
        \mathbb{E}_k^\textbf{Bi}[f(k)] = 
        f(\mu) +
        f'(\mu)\mathbb{E}_k^\textbf{Bi}[(k-\mu)] +
        \frac{1}{2}f''(\mu)\mathbb{E}_k^\textbf{Bi}[(k-\mu)^2] +
        \mathbb{E}_k^\textbf{Bi}[\xi].
    \end{equation}
    Here, $\xi$ is the error of the Taylor expansion. The first moment vanishes (as well as all odd moments in the error, $\mathbb{E}_k^\textbf{Bi}[(k-\mu)^{2l+1}]$). Furthermore, the second moment is simply the variance; therefore 
    \begin{equation}
        \mathbb{E}_k^\textbf{Bi}[f(k)] = f(\mu) + \frac{1}{2}f''(\mu)\sigma^2 + \mathbb{E}_k^\textbf{Bi}[\xi].
    \end{equation}

    In our case, the function analyzed is $f(k)=(k+1)^{-1/2}$, so calculating the derivatives leads to 
    \begin{equation}
        \mathbb{E}_k^\textbf{Bi}[(k+1)^{-1/2}] = 
        \sqrt{\frac{2}{N_k+1}} + \frac{3\sqrt{2}}{8}\frac{N_k-1}{(N_k+1)^\frac{5}{2}}+O\left(N_k^{-\frac{5}{2}}\right),
    \end{equation}
    where the $O(N_k^{-5/2})$ term comes from the error. By also involving a large-$N_k$ approximation, we reach the desired result for $H'_1$.
\end{proof}

For the use of binomial formulas and that of the expected value of the Taylor expansion for $H_1'$, we recommend reading \cite{Li_Goldman_2020} and \cite{Casella_Berger_2002}, respectively.
We also show some approximations necessary for the final formula in the following lemma:

\begin{lemma}\label{lemma: app H approx}
    The following relations hold:
    \begin{align}
        \frac{(H'_1)^2}{(2^{N_k}-1)^2} = & \frac{1}{2N_k} + O(N_k^{-2}), \label{eq: app approx 1}\\
        \frac{H_2'}{2^{N_k}-1} = & \frac{1}{2N_k} + O(N_k^{-2}), \label{eq: app approx 2}\\
        \frac{H_1-H_2'}{N_k(2^{N_k}-1)} = & \frac{1}{2N_k^2} + O(N_k^{-3}). \label{eq: app approx 3}
    \end{align}
\end{lemma}

\begin{proof}
    Following the previous lemma, we use Eq.~\ref{eq: app H1' lemma} to calculate
    \begin{align}
        \frac{(H'_1)^2}{(2^{N_k}-1)^2} =& \frac{2^{2N_k-2}\left[\frac{2}{N_k}\left(1-\frac{1}{8N_k}\right)^2+O(N_k^{-3})\right]}{(2^{N_k}-1)^2} \nonumber \\
        =&\frac{1}{4}\left[\frac{2}{N_k}\left(1-\frac{1}{8N_k}\right)^2+O(N_k^{-3})\right] + O(\exp(-N_k)) \nonumber \\
        =& \frac{1}{2N_k} + O(N_k^{-2}),
    \end{align}
    where in the first line we neglected the ``1'' term adding an exponentially small error, and then reduced to $O(N_k^{-2})$.
    The following relation is computed as
    \begin{align}
        \frac{H_2'}{2^{N_k}-1} =& \frac{(N_k-2)2^{N_k-1}+1}{N_k(N_k-1)(2^{N_k}-1)} \nonumber \\
        = &\frac{N_k-2}{2N_k(N_k-1)} + O(\exp(-N_k)) \nonumber \\
        =& \frac{1}{2N_k}\cdot\frac{1-2N_k^{-1}}{1-N_k^{-1}} \nonumber \\
        =& \frac{1}{2N_k}\left(1-\frac{2}{N_k}+O(N_k^{-2})\right)\left(1+\frac{1}{N_k}+O(N_k^{-2})\right) \nonumber \\
        =& \frac{1}{2N_k}+O(N_k^{-2}).
    \end{align}
    Finally, Eq.~\ref{eq: app approx 3} is calculated as follows:
    \begin{align}
        \frac{H_1-H_2'}{N_k(2^{N_k}-1)}=& \frac{1}{N_k} \left[\frac{1}{N_k}-\left(\frac{1}{2N_k}+O(N_k^{-2})\right)\right] \nonumber \\
        =& \frac{1}{2N_k^2}+O(N_k^{-3}).
    \end{align}
    Thus the three asymptotes are derived.
\end{proof}

To this end, we show that the $H$-terms above reduced to a maximum Hamming weight $N_k/2 < C < N_k$ are of the same order as the previous ones.

\begin{lemma}\label{lemma: app reduced H terms}
    The following summations,
\begin{align}
    H_{1,\text{red}} & := \sum_{\Vec{\gamma} \in [1,C]} \frac{\gamma_b}
    {w(\Vec{\gamma})} = \Theta(H_1)\\
    H_{2,\text{red}} & := \sum_{\Vec{\gamma} \in [1,C]} \frac{\gamma_b \gamma_c}
    {w(\Vec{\gamma})} = \Theta(H_2)\\
    H'_{1,\text{red}} & := \sum_{\Vec{\gamma} \in [1,C]} \frac{\gamma_b}
    {\sqrt{w(\Vec{\gamma})}} = \Theta(H'_1).
\end{align}
\end{lemma}

\begin{proof}
    After performing the same steps as in Lemma \ref{lemma: app H numbers} to simplify the $\gamma_b$ and $\gamma_c$ components, these sums become 
    \begin{align}
        H_{1,\text{red}} &= \sum_{k=0}^{C-1}\frac{1}{k+1}\binom{N_k-1}{k} \\
        H_{1,\text{red}} &= \sum_{k=0}^{C-2}\frac{1}{k+2}\binom{N_k-2}{k} \\
        H'_{1,\text{red}} &= \sum_{k=0}^{C-1}\frac{1}{\sqrt{k+1}}\binom{N_k-1}{k}.
    \end{align}
    All sums take the form
    \begin{equation}
        H_{f, \text{red}}=\sum_{k=0}^c f(k)\binom{n}{k},
    \end{equation}
    for corresponding $n$ and $c$ and strictly decreasing function $f(k)$. 
    
    We notice that, because $C > N_k/2$, all the sums contain the most significant binomial factors, namely $\binom{n}{\lfloor n/2 \rfloor}$. Together with the monotonicity of $f(k)$, for all terms we can infer that
    \begin{equation}
        H_{f, \text{red}} > \sum_{k=c+1}^n f(k)\binom{n}{k}
    \end{equation}, and so 
    \begin{equation}
        H_f = \sum_{k=0}^n f(k)\binom{n}{k} = \sum_{k=0}^c f(k)\binom{n}{k} + \sum_{k=c+1}^n f(k)\binom{n}{k} < 2 H_{f, \text{red}} < 2H_f,
    \end{equation}
    which proves the lemma.
\end{proof}

In fact, because the Binomial distribution has a sharp peak around the middle point, one can find a significantly better approximation, of the form $H_{f, \text{red}} = H_f + O(2^n/n)$, the error becoming small after normalization when calculating the expected variance.

\subsection{Constraining the alternating sum in the expected variance}\label{app: Constraining the alternating sum in the expected variance}

Here we prove that by adding a very loose constraint, we can force the summation in Eq.~\ref{eq: app EVar alternating sum problem} to always be positive. We do this in the following lemma:

\begin{lemma}\label{lemma: app alternating sum}
    For $\{S_{\Vec{\alpha}}^v\}_{\Vec{\alpha}}$ given in Eq.~\ref{eq: app Sav term} and for some ansatz $U(\theta)$ defined in Eq.~\ref{eq: app ansatz}, the map $W_\lambda\mapsto V_\lambda W_\lambda$ flips the sign of
    \begin{equation}
        \sum_{\Vec{\alpha}\in\{0,1\}^{D}}(-1)^{w(\Vec{\alpha})}\Re(S_{\Vec{\alpha}}^v) 
        \quad\longmapsto\quad 
        -\sum_{\Vec{\alpha}\in\{0,1\}^{D}}(-1)^{w(\Vec{\alpha})}\Re(S_{\Vec{\alpha}}^v).
    \end{equation}
\end{lemma}

\begin{proof}
    When mapping $W_\lambda \mapsto V_\lambda W_\lambda \equiv\tilde{W}_\lambda$, the unitary matrix $D_{\Vec{\alpha}}$ transforms as
    \begin{align}
        D_{\Vec{\alpha}} &= \prod_{\eta=D}^1V_\eta^{\alpha_\eta} W_\eta \nonumber \\
        &= V_D^{\alpha_D}W_D\dots V_{\lambda+1}^{\alpha_{\lambda+1}}W_{\lambda+1}\cdot V_\lambda^{\alpha_\lambda} W_\lambda \cdot V_{\lambda-1}^{\alpha_{\lambda-1}}W_{\lambda-1} \dots  V_1^{\alpha_1}W_1 \nonumber \\
        &\mapsto V_D^{\alpha_D}W_D\dots V_{\lambda+1}^{\alpha_{\lambda+1}}W_{\lambda+1}\cdot V_\lambda^{\alpha_\lambda} \tilde{W}_\lambda \cdot V_{\lambda-1}^{\alpha_{\lambda-1}}W_{\lambda-1} \dots  V_1^{\alpha_1}W_1 \nonumber \\
        &= V_D^{\alpha_D}W_D\dots V_{\lambda+1}^{\alpha_{\lambda+1}}W_{\lambda+1}\cdot V_\lambda^{\alpha_\lambda+1} W_\lambda \cdot V_{\lambda-1}^{\alpha_{\lambda-1}}W_{\lambda-1} \dots  V_1^{\alpha_1}W_1 \nonumber \\
        &\equiv D_{\Vec{\alpha}'},
    \end{align}
    where $\Vec{\alpha}'$ has all bits equal to those of $\Vec{\alpha}$, except for $\alpha_\lambda':=\alpha_\lambda \oplus 1$.
    Therefore,
    \begin{equation}
        \Re(S_{\Vec{\alpha}}^v)\mapsto \Re(S_{\Vec{\alpha}'}^v),
    \end{equation}
    and the sum
    \begin{align}
        \sum_{\Vec{\alpha}\in\{0,1\}^{D}}(-1)^{w(\Vec{\alpha})}\Re(S_{\Vec{\alpha}}^v) &\mapsto
        \sum_{\Vec{\alpha}\in\{0,1\}^{D}}(-1)^{w(\Vec{\alpha})}\Re(S_{\Vec{\alpha}'}^v) \nonumber \\
        &\equiv \sum_{\Vec{\alpha'}\in\{0,1\}^{D}}(-1)^{w(\Vec{\alpha}')\oplus 1}\Re(S_{\Vec{\alpha}'}^v) \nonumber \\
        &= - \sum_{\Vec{\alpha'}\in\{0,1\}^{D}}(-1)^{w(\Vec{\alpha}')}\Re(S_{\Vec{\alpha}'}^v).
    \end{align}
    Here we first applied the map, thus obtaining $\Re(S_{\Vec{\alpha}}^v)$, then explicitly rewrote $\Vec{\alpha}$ in terms of $\Vec{\alpha}'$. Flipping one bit is a bijection, so the sum maintains its range. Finally, after extracting the additional $-1$ factor, we can relabel the summation index with $\Vec{\alpha}$, hence proving the lemma.
\end{proof}

Therefore, if our summation is negative, we can pick any of the $W_\eta$ gates (or any odd number of them, for that matter) and multiply them by its associated $V_\eta$, in order to make the sum positive. Note that, by applying the map $W_\lambda \mapsto \tilde{W}_\lambda$, the $S_{\Vec{\alpha}}^c$ summation in the variance is unchanged.

\subsection{Curse of dimensionality}

The curse of dimensionality, first mentioned by Richard Bellman in 1957 \cite{bellman1957dynamic}, is a phenomenon that occurs in many areas of computational and data science \cite{Chen_2009, Altman2018-xp, Hu_2024}. In our case, we formally define it as a vanishingly small overlap between two random states as their Hilbert space dimension increases:
\begin{lemma}\label{lemma: app curse of dim}
    For two random $n$-qubit states $\ket{\psi}$ and $\ket{\phi}$, their overlap (modulus squared) averages to
    \begin{equation}
        \mathbb{E}_{\mathcal{H}}[|\braket{\phi|\psi}|^2]=\frac{1}{2^n}.
    \end{equation}
\end{lemma}
\begin{proof}
    By mapping state $\ket{\psi}\equiv U_\psi\ket{0}$ to $\ket{0}$ with some gate $U_\psi^\dagger$, the state $\ket{\phi}$ becomes $\ket{\phi'}=U_\psi^\dagger\ket{\phi}$. As such, the inner product becomes
    \begin{equation}
        \braket{\phi | \psi} = \braket{\phi | U_\psi U_\psi^\dagger|\psi} = \braket{\phi' | 0},
    \end{equation}
    where $\ket{\phi'}$ is random. By writing the amplitudes of $\ket{\phi'}\equiv\sum_i c_i\ket{i}$ in the standard basis, we have $|\braket{\phi'|0}|^2=|c_0|^2$.
    Finally, the state considered is randomly selected from an isotropic distribution, so 
    \begin{align}
        \mathbb{E}_\mathcal{H}[|c_0|^2] &= \mathbb{E}_\mathcal{H}[|c_i|^2] \nonumber \\
        &= \frac{1}{2^n}\sum_{i=1}^{2^n}\mathbb{E}_\mathcal{H}[|c_i|^2] \nonumber \\
        &= \frac{1}{2^n}\mathbb{E}_\mathcal{H} \left[\sum_{i=1}^{2^n}|c_i|^2\right] \nonumber \\
        &=\frac{1}{2^n},
    \end{align}
    which completes the proof.
\end{proof}

\subsection{Calculating the $\hat{R}$ values}

Here we provide a derivation for the $\hat{R}$ values used in Appendix \ref{app: Graphs with weighted edges}:

\begin{lemma}\label{lemma: app R hat bounds}
    For the $\hat{R}_{||}$, $\hat{R}_{\land}$, $\hat{R}_{\lor}$, and $\hat{R}_{|}$ defined in Eqs.~\ref{eq: app R hat ||}, \ref{eq: app R hat land}, \ref{eq: app R hat lor}, and \ref{eq: app R hat |} respectively, the bounds for weighted $k$-regular graphs are the following:
    \begin{align}
        (N-1)^2 \leq & \hat{R}_{||} \leq N(N-1) \\
        0 \leq & \hat{R}_{\land} \leq N-1  \\
        0 \leq & \hat{R}_{\lor} \leq N-1  \\
        1 \leq & \hat{R}_{|} \leq N.
    \end{align}
    We recall that $N$ is the number of nodes.
\end{lemma}

\begin{proof}
First, we recall that for a $k$-regular graph, $k+1 = \tilde{D}_{ii} = \sum_j \tilde{A}_{ij}$. Therefore, the normalized version reads $\sum_j\hat{A}_{ij}=1$.

In this case, we begin with $\hat{R}_{|}$. For the upper bound we take the square of a sum to be larger than the sum of squared terms:
\begin{equation}
    \hat{R}_{|} = \sum_{i,j}\hat{A}_{ij}^2 \leq \sum_i \left(\sum_j\hat{A}_{ij}\right)^2 = \sum_i 1 = N.
\end{equation}
The lower bound is done via the Cauchy-Schwarz inequality:
\begin{equation}
    \hat{R}_{|} = \sum_{i,j}\hat{A}_{ij}^2 = \sum_{i}\frac{1}{s_i}\left(\sum_{j'=1}^{s_i}\hat{A}_{ij}^2\right)\left(\sum_{j'=1}^{s_i}1\right) \geq \sum_i \frac{1}{s_i}\sum_{j'=1}^{s_i}\hat{A}_{ij} \geq \frac{1}{s}\sum_i\sum_j\hat{A}_{ij} = \frac{N}{s}.
\end{equation}
Here, $s_i$ is the sparsity of row $i$ and the $s$ is the row sparsity of matrix $\hat{A}$. The worst case is when $s=N$, so $\hat{R}_{|}\geq 1$.

The other $\hat{R}$ values only inherit the bounds from $\hat{R}_{|}$, the other terms being calculated precisely. In $\hat{R}_{\land}$, the first term is
\begin{equation}
    \sum_{i,j,k} \hat{A}_{ij}\hat{A}_{kj} = \sum_j \left(\sum_i \hat{A}_{ij}\right)\left(\sum_k \hat{A}_{kj}\right) = \sum_j 1\cdot 1 = N.
\end{equation}
Therefore, by subtracting the bounds of $\hat{R}_{|}$, we obtain the desired bounds for $\hat{R}_\land$. Thanks to the symmetry of $\hat{A}$, the same bounds hold for $\hat{R}_\lor$.

Finally, by rewriting Eq.~\ref{eq: app R hat ||} in terms of Eqs.~\ref{eq: app R hat land} and \ref{eq: app R hat lor}, we have
\begin{equation}
    \hat{R}_{||} = \underbrace{\left(\sum_{i,j}\hat{A}_{ij}\right)}_{=N}\underbrace{\left(\sum_{k,l}\hat{A}_{kl}\right)}_{=N} - 2\underbrace{\sum_{i,j,k}\hat{A}_{ij}\hat{A}_{kj}}_{N} + \hat{R}_{|}.
\end{equation}
Plugging in the bounds for $\hat{R}_{|}$ gives the desired bounds.
\end{proof}

\subsection{Rank-sparsity relation}

Here we prove that the normalized adjacency matrix with self-loops $\hat{A}$ of a graph has a high rank only if it is sparse.

\begin{lemma}\label{lemma: rank vs sparsity}
    For a graph, the normalized adjacency matrix with self-loops $\hat{A}$ is defined as given in Appendix \ref{app: Graph preliminaries}. Then, its rank $r$ depends on its dimension $N$ and sparsity $s$ by the bound 
    \begin{equation}
        r \geq \frac{N}{s}.
    \end{equation}
\end{lemma}
\begin{proof}
    First, $\hat{A}=\tilde{D}^{-1/2}\tilde{A}\tilde{D}^{-1/2}$, where $\tilde{D}$ are invertible operators. Its rank is the dimension of the image of $\hat{A}$, which we denote by $\rk(\hat{A})=\dim(\im(\hat{A}))$. Therefore, by mapping $\im(\hat{A}) \to \im(\tilde{D}^{1/2}\hat{A})$ with the bijection $\tilde{D}^{1/2}$ to the left, we preserve the rank. Because transposition also preserves both rank and invertibility, applying invertible operators to the right also preserves the rank. Therefore,
    \begin{align}
        \rk(\tilde{A}) &=
        \rk(\tilde{D}^{1/2}(\hat{A}\tilde{D}^{1/2})) =
        \rk(\hat{A}\tilde{D}^{1/2}) =
        \rk((\hat{A}\tilde{D}^{1/2})^T) \nonumber\\ &= 
        \rk((\tilde{D}^{1/2})^T\cdot\hat{A}^T) = \rk(\hat{A}^T) = \rk(\hat{A}) = r.
    \end{align}
    Here, we expanded $\tilde{A}$, then reduced the bijection to the left, transposed, reduced the new bijection to the left, and transposed again, in this order, to reach $r$.

    We now start from the Cauchy-Schwarz inequality applied on the sum of the eigenvalues $\lambda_i$ of $\tilde{A}$:
    \begin{equation}\label{eq: app CBS}
        \left(\sum_{i=1}^r \lambda_i\cdot1\right)^2 \leq \left(\sum_{i=1}^r 1\right)\cdot \left(\sum_{i=1}^r \lambda_i^2\right),
    \end{equation}
    however, the LHS is simply the squared trace $\Tr(\tilde{A})=N$, and the sum over squared eigenvalues can be further bounded by
    \begin{equation}
        \sum_{i=1}^r \lambda_i^2 = \Tr(\tilde{A}^2) = \sum_{i,j}\tilde{A}_{ij}\tilde{A}_{ji} = \sum_{i,j}\tilde{A}_{ij}^2\leq N\cdot s,
    \end{equation}
    which counts the maximum number of ones present in the matrix. Therefore, the Eq.~\ref{eq: app CBS} leads to
    \begin{equation}
        N^2\leq r\cdot N \cdot s \quad \Rightarrow \quad r\geq N/s.
    \end{equation}
    So a sparse graph has high rank $r$.
\end{proof}